\documentclass[12pt]{article}
\usepackage{frankstyle}
\usepackage{mathabx}

\usepackage{tikz}
\usetikzlibrary{decorations.pathreplacing,angles,quotes}

\title{The Network Propensity Score: \\ Spillovers, Homophily, and Selection into Treatment \thanks{\footnotesize I would like to thank Xu Cheng, Francis J. DiTraglia, Petra Todd, Jere Behrman, Frank Schorfheide, Francis Diebold, Wayne Gao, Karun Adusumilli, Ben Golub, Mathew Jackson, Hyungsik Roger Moon, Bryan Graham, Suyong Song, Andrin Pelican, Luis Candelaria, Aureo de Paula, and Juan Camilo Castillo, for helpful comments and suggestions, as well as seminar participants at the 2021 IAAE Webinar, 2021 IAAE Conference, the Young Economists Symposium 2020, the Econometric Society World Congress 2020, the Warwick 2019 Ph.D. Conference, Emory QTM, Universit\'{e} de Montr\'{e}al, Central European University, Lehigh University, University of Washington, Bates White, University of Exeter, Universidad de los Andes, CEMFI, Bates White, and the UPenn Empirical Micro and Econometrics Workshops. This project was supported by a Penn SAS Dissertation Completion Fellowship. Email: \href{mailto:alejandro.sanchez.becerra@emory.edu}{alejandro.sanchez.becerra@emory.edu}}}

\author[1]{
Alejandro S\'{a}nchez-Becerra, Emory QTM}

\date{ This version: \today \\ First Version: January, 2019 \vspace{-1.5em}} 

\begin{document}

\clearpage
\maketitle
\thispagestyle{empty}

\begin{abstract}
  \singlespacing

I establish primitive conditions for unconfoundedness in a coherent model that features heterogeneous treatment effects, spillovers, selection-on-observables, and network formation. I identify average partial effects under minimal exchangeability conditions. If social interactions are also anonymous, I derive a three-dimensional network propensity score, characterize its support conditions, relate it to recent work on network pseudo-metrics, and study extensions. I propose a two-step semiparametric estimator for a random coefficients model which is consistent and asymptotically normal as the number and size of the networks grows.  I apply my estimator to a political participation intervention Uganda and a microfinance application in India.


	\medskip
\end{abstract}


\linespread{1.3}
\selectfont
\newpage
\section{Introduction}

A popular strategy to identify average treatment effects in quasi-experimental settings is to compare the outcomes of individuals with similar characteristics but different treatment status. For this strategy to be valid, the researcher needs to satisfy a high-level \textit{unconfoundedness} condition --also known as \textit{selection-on-observables}-- which lists a set of observable covariates that delimit comparison subgroups, and a \textit{support}  condition, that guarantees sufficient people to compare in each subgroup. Such conditions are trivially satisfied in experiments with known assignment probabilities, but require further justification in observational settings. For example, researchers can appeal to institutional features of the program assignment rules or prior knowledge of the participants' decision-making process. While this type of strategy has a long history, researchers have traditionally ignored cases with meaningful interference/spillovers, where individual's potential outcomes depend on the treatment status of others, in addition to their own.\footnote{These situations violate the Stable Unit Treatment Value Assumption (SUTVA).} A burgeoning literature has focused on extending these notions of unconfoundedness and support to situations with network spillovers \citep{forastiere2020identification, liu2019doubly,sofrygin2017semi}.\footnote{For example, job placement programs can displace non-participants from the labor market \citep{crepon2013labor}, cash transfers can affect informal insurance networks \citep{meghir2020migration}, and professional events can encourage the adoption of business practices \citep{fafchamps2018networks}.}. In spite of these advances, much less is known about the economic content or plausibility of these assumptions outside experimental settings.

The problem is that spillovers introduce a second layer of selection through the choice of social connections. Individuals may be likely to befriend others with similar willingness to participate in the program. This implies that strategies to identify spillovers by comparing the outcomes of social groups with high and low participation rates may be misleading. Cross-sectional differences could reflect sorting patterns into high- or low-take-up groups, rather than a product of social interactions. This type of phenomenon is called homophily. Resolving this problem by comparing the outcomes of two ``similar'' individuals with differential friend take-up rates is a step in the right direction. However, it is difficult to define which pairs of individuals are actually comparable in a network sense, at least without further assumptions. The problem becomes more acute in real-world social networks, where individuals have non-overlapping sets of friends. Defending institutional/decision-based rationales for unconfoundedness is hard without a global, internally consistent model of collective decision-making due to the many different ways in which the network composition could affect the selection process in either layer.

This paper's main contribution is to propose primitive assumptions for unconfoundedness and support conditions, which suffice to identify causal effects in the presence of spillovers. To do so I propose an internally consistent model with selection-on-observables, bilateral network formation, and potential outcomes with spillovers (via a random coefficients specification). To keep things tractable I focus on binary networks, observed by the researchers, that do not quantity relationship intensity. Within my framework individuals select into treatment and form connections based on a combination of observed characteristics and i.i.d shocks. My non-parametric identification approach is constructive and builds on the notion of a graphon -- a function that can be used to represent a large class of exchangeable network processes. Most importantly, depending on the primitives of the model, the data can exhibit no homophily, homophily with spillovers, or homophily without spillovers. I argue that elementary building blocks can lead to rich selection patterns and produce bias of na\"{i}ve procedures, even without modeling strategic considerations explicitly. 

The strength of my results depends on the generality with which the researcher decides to model spillovers. First, I focus on a class of exchangeable potential outcomes models. I allow for situations where the outcome can depend on the take-up of friends, the take-up of friends of friends, or weighted averages that depend on covariates of friends. This nests a heterogeneous version of the reduced-form linear-in-means \citep{bramoulle2009identification}, interference with rooted networks \citep{auerbach2021local}, and approximate network interference \cite{leung2022causal}. Second, I specialize my results to models that satisfy an anonymous interactions assumption. This condition --which has been extensively analyzed in a number of recent papers \citep{aronowsamii2017estimating,leung2019treatment,spilloversnoncompliance,forastiere2020identification,liu2019doubly,sussman2017elements,tchetgen2017auto}-- states that the potential outcomes spillovers only enter through direct friend connections, equally-weighted. The exchangeable and anonymous interaction models coincide when individuals are fully connected within disjoint clusters.\footnote{In that case the modeling approach is sometimes known as partial interference}. It is important to emphasize that both of these generalizations nests the Stable Unit Treatment Value (SUTVA) assumption.

I then use this framework to establish two key findings. First, the researcher can satisfy the unconfoundedness condition by choosing individual determinants of treatment take-up and relationship choices --but not friends' take-up decisions. This applies to a large class of exchangeable spillover models. Second, for the subset of models with anonymous interactions \citep{aronowsamii2017estimating} there exists a three-dimensional individual statistic --that I call the network propensity score (NPS)-- which can be used as a matching variable. The first component is the individual propensity score \citep{rosenbaum1983central}, the second is a measure of friend take-up rate, and the last is the number of friends. Crucially, the validity of the support condition is easy to verify and can be motivated from the patterns of association in the network. From a structural perspective, the NPS can be expressed as an integrand of the take-up process, friend preferences over traits, and the measure of traits in the population
 
I establish the relationship between the network propensity score and recently proposed network pseudo-metrics \citep{auerbach2022identification,zeleneev2020identification}, and illustrate weak identification issues that arise from applying those approaches to study spillovers. I propose alternative assumptions to deal with unobserved heterogeneity, that encompass other strategies recently used in the literature \citep{spilloversnoncompliance,johnsson2019estimation,imbens2009identification}. I propose a two-step semi-parametric estimator, that is based on inverse-weighting in a random coefficients specification \citep{graham2022semiparametrically,wooldridge1999distribution}. In the asymptotics, I allow for the possibility that a subset of the control variables are unobserved but can be consistently estimated in large networks. My approach is agnostic about intra-network dependence, and hence the rate of convergence of the estimator is going to depend on the total number of groups/networks.
 
I apply my methodology to two empirical examples. First, I consider an intervention designed to increase political participation in Uganda \citep{eubank2019viral,ferrali2020takes}. Citizens voluntarily participated in quarterly information sessions about ways to engage with local district officials. I find evidence of spillovers because individuals with a higher number of friends participating in the sessions were more likely to be politically active, after controlling for covariates. The estimates of the spillover effects under my approach are statistically significant and about twice the size of comparable ordinary least squares (OLS) regressions with additive covariates. The network propensity score matching methodology is better equipped to handle heterogeneous spillover effects that can be correlated with the endogeneous regressors.

In the second example, I analyze the effects of an intervention to increase microfinance adoption \citep{banerjee2013diffusion}. This example has been analyzed extensively by the econometrics literature \citep{candelaria2020semiparametric,chandrasekhar2014tractable} and has lead to many follow-up projects \citep{banerjee2017gossip,breza2019social,chandrasekhar2018social}. The microfinance organization used a non-random selection rule based on occupation of household members (shopkeepers, teachers), who received in-depth information about the loans offered by the company. In practice, households with higher wealth and privileged castes were both more likely to receive treatment themselves and to be friends with others that received treatment as well. My network propensity score matching approach estimates large treatment effects but limited local spillover effects. The results suggest that while network characteristics such as centrality can affect short-term speed of information diffusion \citep{banerjee2013diffusion,akbarpour2018just}, local neighbor interactions may play a smaller role on medium-term loan adoption after accounting for covariates.

Finally this paper considers applications of the network propensity score approach to stratified experiments. I analyze experiments that exogenously assign treatment probabilities across multiple networks \citep{duflo2003role,crepon2013labor,baird,vasquezbare}. I find that the network propensity score has a simple form in both cases under perfect compliance. I also consider settings with non-compliance and spillovers  \citep{spilloversnoncompliance,vasquezbare,imai2020causal}. I discuss the applicability of the network propensity score to identify average spillover effects under  non-compliance in sparse networks.

The paper is organized as follows. Section \ref{litreview} reviews prior literature. Section \ref{model} introduces the model. Section \ref{exchangeability} presents identification under exchangeability. Section \ref{networkpropensityscore} introduces the network propensity score. Section \ref{estimation} proposes feasible estimator and presents the asymptotic results. Section \ref{example} discusses the two empirical examples. Section \ref{conclusion} concludes.

\section{Related Literature}
\label{litreview}

There have been three recent approaches in the literature that extend propensity score methods for use with network data. The first approach uses relationship data and friend covariates to relax the selection on observables assumption. \cite{jackson2020adjusting} assume that program participation is the result of a strategic game with friends (spillovers in treatment), but assume that there are no spillovers on outcomes. The second approach assumes selection on observables (without spillovers) but focuses on pairwise outcomes. \cite{arpino2015implementing}, for example, compute the propensity score of adopting tariff agreements and use it to evaluate their effect on bilateral trade between countries. The third approach, which is closest to my own, incorporates spillovers by assuming \textit{anonymous interactions} \citep{manski2013identification}, which implies heterogeneous outcomes that depend on own treatment and the total number of treated friends. This approach is sometimes called \textit{multi-treatment} matching because it assumes that individuals with different numbers of treated friends experience  different intensities that satisfy unconfoundedness \citep{{forastiere2018estimating, liu2019doubly,sofrygin2017semi}}. In this case a form of generalized propensity score \citep{hirano2004propensity} is computed for each exposure level. Recent work in economics incorporates similar uncounfoundedness assumptions \citep{leung2019treatment,viviano2019policy,ananth2020optimal}. \cite{qu2021efficient} propose an efficient estimator under heterogeneous partial inference. In most of these papers, the network is typically treated as exogenous, and often times variation in the treatment is the main source of identification.

One of the key innovations is to prove unconfoundedness from a general micro-founded setting with exchangeability. \cite{manski2013identification} consider a slightly broader class of social interaction models, but restrict attention to experiments. \cite{qu2021efficient} propose a version of exchangeability and unconfoundedness based on partitions of the sample into disjoint groups of influence. I propose a more general version that accommodates higher-order connections \citep{auerbach2021local,leung2022causal,bramoulle2009identification} or covariate-weighting to characterize peers with more influence. Covariate-weighted versions can be computed with full network data or more cost-effective approaches using aggregate relational data \citep{breza2020using,alidaee2020recovering}. My results are also novel for the anonymous interactions. A wide literature, e.g. \cite{forastiere2018estimating}, compute a version of generalized propensity score for network data by predicting each possible category of the endogenous variable. The rationale for unconfoundedness is often based on implicit arguments for network homophily. However, I show that if network formation/selection arguments are being used to select the covariates believed to satisfy unconfoundedness, then the dimensionality of the generalized propensity scores can be theoretically be brought down to three (the network propensity score).

To my knowledge, this is the first paper to combine an internally consistent non-parametric network model with selection-on-observables to justify unconfoundedness. Previously, \cite{goldsmith2013social} studied a parametric network model and use it to identify causal effects. \cite{johnsson2019estimation} extended this idea to a class of network models that satisfy a monotonicity restriction similar to \cite{gao2019nonparametric}. Both papers restrict attention to cases with an exogenous treatment. However, identification of causal effects is possible for a broader class of network models. In two influential papers \cite{aldous1981representations} and \cite{hoover1979relations} showed that \textit{any} network process whose distribution is ex-ante independent of the ordering of agents can be represented as a dyadic network with independent covariates and independent shocks. Recent work has also shown how to micro-found the dyadic model from dynamic games \citep{mele2017structural}, and how to  test the empirical validity of dyadic models with additive fixed effects \citep{pelican2019testing}.

The key empirical challenge is whether the covariates of the Aldous-Hoover representation are actually observed or whether some of them may be latent. Imposing monotonicity as in \citep{graham2017econometric, gao2019nonparametric, johnsson2019estimation} is one way to recover latent heterogeneity.  Another recent literature \citep{auerbach2019identification,zeleneev2020identification} focuses on a network pseudo-metric to analyze more general forms of unobserved heterogeneity. \cite{auerbach2019identification}, however, argues that this form of heterogeneity cannot be separately identified from spillovers in dense networks with exogenous treatment. I show why this concern carries over the case with selection-on-observables, by establishing the relationship between the network propensity score and the network pseudo-metric. Another route to recover unobserved heterogeneity is to explore group patterns in treatment decisions. \cite{spilloversnoncompliance} recover control variables in experiments with spillovers and one-sided non-compliance. I discuss conditions that allow for this type of control variables. These conditions could also nest related matching approaches that exploit exponential family forms \citep{arkhangelsky2018role}.


\section{Model}
\label{model}

I assume that there are $g = \{1,\ldots,G\}$ disjoint groups that contain $i = \{1,\ldots,N_g\}$ individuals each. We can interpret $g$ as the identifier for a school, village or city. Treatment status is denoted by a binary variable $D_{ig}$ that equals one if individual $\{ig\}$ is treated and zero if she is not. Each individual has a vector of socio-economic covariates $C_{ig}$, which can be stacked in a matrix $C_g$. A social network is denoted by an $N_g \times N_g$ adjacency matrix $A_g$ with binary entries. Each entry $A_{ijg}$ equals one if individuals $\{ig\}$ and $\{jg\}$ are friends and zero otherwise, using the convention that $A_{iig} = 0$. I define two additional measures: the total number of $\{ig\}'s$ friends $L_{ig} \equiv \sum_{j=1}^{N_g} A_{ijg}$ and the total number of $\{ig\}'s$ \textit{treated} friends by $T_{ig} \equiv \sum_{j=1}^{N_g}A_{ijg}D_{jg}$. The variables $L_{ig}$ and $T_{ig}$ are meant to capture peer influence in $\{ig\}'s$ immediate friend circle.

I analyze a model where a scalar outcome $Y_{ig}$ is determined by\footnote{In Appendix \ref{secappendixidentification} I show how to extend non-parametric identification results to models of the form $Y_{ig}=m(X_{ig},\tau_{ig})$, where $m(\cdot)$ is an arbitrary function, possibly unknown.}
\begin{equation}
Y_{ig} = X_{ig}'\tau_{ig},
\label{eq:linearoutcome_matrix}
\end{equation}
where $\tau_{ig} \equiv (\alpha_{ig},\beta_{ig},\gamma_{ig}',\delta_{ig}')' \in \mathbb{R}^{2+2k}$ is a vector of real-valued random coefficients, and $X_{ig}$ is vector of endogenous regressors defined by
$$ X_{ig}' \equiv \begin{bmatrix} 1 & D_{ig} & \varphi(i,A_g,C_g,N_g) & D_{ig}\times \varphi(i,A_g,C_g,N_g) \end{bmatrix} $$
Here, $\varphi:\mathbb{Z}_{+}^2 \to \mathbb{R}^k$ is a \textit{known} function which reflects the researcher's beliefs about how the treatment of others affects unit $\{ig\}$.\footnote{The function $\varphi(\cdot)$ could accommodate non-linearities by having different basis functions.} In the most general form, it can depend on $\{ig\}$'s position within network $A_g$, a vector of treatment of indicators $D_g$ for other people in the group, and a matrix of covariates $C_g$. It is common to restrict attention to functions that depend the treatment status of immediate neighbors. For example, if $\varphi(i,A_g,C_g,N_g) = \sum_{j=1}^{N_g}A_{ijg}D_{jg}/\sum_{j=1}^{N_g} A_{ijg} = T_{ig}/L_{ig}$, then the outcome is only determined by own treatment and the proportion of treated neighbors, and the outcome equation simplifies to
\begin{equation}
Y_{ig} = \alpha_{ig} + D_{ig}\beta_{ig} + (T_{ig}/L_{ig})\gamma_{ig} +  D_{ig} \times (T_{ig}/L_{ig}) \delta_{ig}.
\label{eq:linearoutcome}
\end{equation}

 I am interested in identifying the average partial effects for a target population $\mathcal{F}$, defined as
\begin{equation}
\tau \equiv (\alpha,\beta,\gamma,\delta) = \mathbb{E}[\tau_{ig} \mid \mathcal{F}].
\label{eq:averagepartialeffect}
\end{equation}
The average partial effects vector $\tau$ integrates the coefficients in \eqref{eq:linearoutcome}. The conditioning $\mathcal{F}$ is important to emphasize that the average is computed for a specific subpopulation (men or women, old or young, etc.). When the conditioning set is empty, i.e $\mathcal{F} = \emptyset$, the average is computed for the entire population.

\begin{rem}[No spillovers]
The potential outcomes model \citep{fisher1960design,rubin1980randomization} that is routinely used in program evaluation is a special case of \eqref{eq:linearoutcome}. In that case we set $\gamma_{ig} = \delta_{ig} = 0$ and define individual-specific outcomes by treatment status as $Y_{ig}(0) \equiv \alpha_{ig}$ and $Y_{ig}(1) \equiv \alpha_{ig} + \beta_{ig}$. The average treatment effect is defined as $\beta = \mathbb{E}[\beta_{ig} \mid \mathcal{F}] = \mathbb{E}[Y_{ig}(1) - Y_{ig}(0) \mid \mathcal{F}]$. Heterogeneity of $\beta_{ig}$ is important to capture varying responses to treatment. Researchers are often interested in testing $\beta = 0$, the null hypothesis that the treatment has no effect on average. If $\beta > 0$ then the treatment has a positive effect on the population of interest, and a negative effect if $\beta < 0$.
\end{rem}

\begin{rem}[Interpreting direct and indirect effects]The more interesting case is when $\gamma_{ig}$ and $\delta_{ig}$ are not zero. For simplicity assume that $\varphi(t,l) = t/l$ and $l > 0$, which implies that the model in \eqref{eq:linearoutcome} is a linear function of own treatment, the fraction of treated friends, and an interaction. We can define the potential outcomes as $Y_{ig}(0,t,l) \equiv \alpha_{ig} + \gamma_{ig}\times(t/l)$ and $Y_{ig}(1,t,l) \equiv \alpha_{ig} +\beta_{ig} + (\gamma_{ig} +\delta_{ig})\times(t/l)$. The \textit{direct} average treatment effect is equal to $\mathbb{E}[Y_{ig}(1,t,l) - Y_{ig}(0,t,l) \mid \mathcal{F}] = \beta + \delta\times (t/l)$. In contrast to the Fisher-Rubin model, the magnitude of the treatment effect depends on how many friends are treated. For example, if $\delta > 0$ then having more treated friends widens the gap between the treated and control. In addition to the ATE we can compute the spillover effect for control individuals $\mathbb{E}[Y_{ig}(0,t,l) -Y_{ig}(0,0,l) \mid \mathcal{F}] = \gamma \times t/l$. If $\gamma > 0$ then control individuals have better outcomes when some of their friends are treated even if they are not participating in the treatment directly. Modeling heterogeneity of $(\gamma_{ig},\delta_{ig})$ is important to capture the fact that not everyone is equally susceptible to peer influence.
\end{rem}

\section{Identification under Exchangeability}
\label{exchangeability}

The main barrier to identifying the average partial effect is that $\tau_{ig}$ and $X_{ig}$ might be correlated. To address this problem, I propose a control variable $V_{ig}$ that captures the main determinants of treatment and network formation. I assume that $V_{ig}$ satisfies the unconfoundedness condition $\tau_{ig} \ \indep  \ X_{ig} \mid V_{ig}$ and that $\mathcal{F}$ is $V_{ig}-$ measurable. For example, $\mathcal{F}$ could include gender and $V_{ig}$ could include a finer set of variables such as gender, age and wealth. I establish primitive assumptions on the network and treatment processes that justify these conditions in the next section. Under unconfoundedness \citep{wooldridge2003further,graham2022semiparametrically,wooldridge1999distribution}, we can identify average partial effects, as follows
\begin{thm}[Average Partial Effects]
	Suppose that (i) $Y_{ig} = X_{ig}'\tau_{ig}$, (ii) $X_{ig} \ \indep \ \tau_{ig} \mid V_{ig}$, (iii) $\mathcal{F}$ is $V_{ig}-$measurable and $\mathbf{Q}_{xx}(v) = \mathbb{E}[X_{ig}X_{ig}' \mid V_{ig} = v]$ is invertible almost surely over the support of $V_{ig} \mid \mathcal{F}$ . Then $\tau$ defined in \eqref{eq:averagepartialeffect} is equal to $\mathbb{E}[\mathbf{Q}_{xx}(V_{ig})^{-1}X_{ig}Y_{ig} \mid \mathcal{F}]$.
	\label{lem:randomcoef_identification}
\end{thm}
Intuitively, Theorem \ref{lem:randomcoef_identification} states that researchers can identify average partial effects by comparing the outcomes of individuals with similar values of $V_{ig}$ but different realizations of $X_{ig}$. As shown by \cite{graham2022semiparametrically}, this estimand is equivalent to computing an OLS coefficient for each subset $\{V_{ig} = v \}$ and averaging the results. 

 However, finding a $V_{ig}$ that meets these properties in the network setting is challenging, at least without further assumptions on the outcome, treatment, and network processes. I will propose a strategy to construct $V_{ig}$ that depends on weak notions of exchangeability. Broadly, exchangeability refers to the idea that the distribution of variables for a set of individuals is ex-ante identical, and that the distribution is invariant to relabeling of the observations in the network.

\subsection{Sufficient Conditions}

I assume that the outcome model is determined by spillovers that are exchangeable in the identities of the individuals. Formally, let $\Pi_{ij}$ be an $N_g \times N_g$ rotation matrix. For an $n \times 1$ vector $x$, the matrix is designed in such a way that $\Pi_{ij}x$ exchanges the order of the $i^{th}$ and $j^{th}$ rows. Under this definition, the matrix $\Pi_{ij}A_g\Pi_{ij}'$ is an adjacency matrix that exchanges the neighbor connections of $\{ig\}$ and $\{jg\}$. Similarly, $\Pi_{ij}C_g$ exchanges the order of covariates. I assume that social interactions are exchangeable if the following assumption holds.

\begin{assumptxt}[Exchangeable interactions]
	\label{assumptxt:exchangeableinteractions}For all $i,j \in \{1,\ldots,N_g\}$, $$\varphi(i,A_g,D_g,C_g,N_g) = \varphi(j, \underbrace{\ \Pi_{ij}A_g\Pi_{ij}'}_{\substack{\text{Reordered} \\\text{Network}}},\ \underbrace{\Pi_{ij}D_g}_{\substack{\text{Reordered} \\\text{Treatment}}}, \ \underbrace{\Pi_{ij}C_g}_{\substack{\text{Reordered} \\\text{Covariates}}}, \ N_g).$$	
\end{assumptxt}
\nameref{assumptxt:exchangeableinteractions} states that if $\{jg\}$ had all of the same connections as $\{ig\}$, then they would have the same value of the exposure function. Therefore, the only thing that matters is the network structure and $\{ig\}$'s relative position within the network: not any particular labeling of the dataset. This assumption is clearly satisfied if $\varphi(i,A_g,C_g,N_g)$ is a function of the total number of treated friends, giving everyone equal weight. It can also holds in more complex cases where some neighbors are be more influential than others.

\begin{rem}[Covariate-weighting]
	The exposure mapping could include differential weights for each neighbor depending on their covariates.  For instance, it allows for exposures of the form $\varphi(i,A_g,D_g,C_g,N_g) = \sum_{j \ne i} A_{jg}C_{jg}D_{jg}$, where some neighbors are given more weight. This type of sum still produces the same value even if we exchange the order of the $j$ indexes. \cite{qu2021efficient} consider a setting where $A_{ijg} = 1$ for all $i \ne j$ (partial interference) and impose a version of exchangeability of $D_g$ withing each cluster that they call ``conditional exchangeability''. \nameref{assumptxt:exchangeableinteractions} is more general than that because it allows for sparse networks with $A_{ijg} \in \{0,1\}$, and more general types of covariate-weighting.
\end{rem}	

\begin{rem}[Rooted Networks]
	\cite{auerbach2021local} propose a spillovers model where the outcome depends on the local network structure. A ``rooted'' network is the subnetwork that is generated by starting in node $\{ig\}$ and constructing first, second, and multi-order connections. In this way, spillovers depend on the local composition of treated individuals, up to permutations of the local network structure. Under suitable restrictions, this type of model is exchangeable.
\end{rem}

\begin{rem}[Social Interactions Models]
	In a typical social interactions model, $Y_{ig} = \gamma_0 + \gamma_1 \sum_{j \ne i}A_{jg} + \gamma_2D_{ig} + \gamma_2 \sum_{j \ne i}A_{jg}D_{jg} + \varepsilon_{ig}$, where $\gamma = (\gamma_0,\gamma_1,\gamma_2,\gamma_3)$ is a vector of parameters and $\varepsilon_{ig}$ is an error term. \cite{bramoulle2009identification} show $Y_{ig}$ can be written in reduced-form as a weighted average of $D_g$, where the weights depend on $\gamma$ and a function of the network connections. In vector form, this is $Y_g = \tfrac{\gamma_0}{1-\gamma_1}1_{N_g \times 1} + \gamma_2 D_g + (\gamma_1\gamma_2 + \gamma_3)\sum_{k=0}^{\infty} \gamma_1^kA_g^{k+1}D_g +\sum_{k=0}^{\infty} \gamma_1^kA_g^{k}\varepsilon_g$, where $1_{N_g \times 1}$ is a vector of ones of length $N_g$. By construction, this model is exchangeable. While the standard version assumes homogeneous coefficients, this could be extended to version where the reduced-form has heterogeneous coefficients and/or truncates the order of the sum. Relaxations under approximate network interference have also been considered by \cite{leung2022causal}.	
\end{rem}

I assume that the researcher has auxiliary covariates that explain $\{ig\}'s$ participation in the treatment and choice of friends. As discussed before, let  $C_{ig} \in \mathbb{R}^{d_c}$ be a vector of individual characteristics that are sampled at random from a super-population and define $\Psi_g^* \in \mathbb{R}^{d_\Psi}$ to be a vector of group characteristics. I next describe assumptions on the core structure that provide guidance on the choice of $V_{ig}$.
\begin{assumptxt}[Sampling Exchangeability]\text{ }
	\begin{enumerate}[(i)]
		\setlength\itemsep{-0.3em}
		\item $\quad$(Across Groups) \ $\{\tau_{ig},D_{ig},C_{ig}\}_{i=1}^{N_g}$, $\Psi_g^*$ are i.i.d. across groups.
		\item $\quad$(Within Groups) \ $\{\tau_{ig},D_{ig},C_{ig}\}$ are i.i.d. within group given $\Psi_g^*$.
	\end{enumerate}
	\label{assumptxt:randomsampling}
\end{assumptxt}
The first part of \nameref{assumptxt:randomsampling} --stating that groups are i.i.d-- is plausible when the groups are spatially, economically or socially separated. The second part states that the covariates within a group are conditionally independent within groups, which is a common assumption in the literature on network formation \citep{johnsson2019estimation,graham2017econometric,auerbach2019identification}.
\begin{assumptxt}[Selection on Observables]
	$\quad \tau_{ig} \ \indep \ D_{ig} \mid C_{ig},\Psi_g^*$.
	\label{assumptxt:confounders}
\end{assumptxt}
\bigskip

The \nameref{assumptxt:confounders} assumption states that the treatment status is independent of the treatment effects, after controlling for baseline characteristics. It puts the burden on researchers to identify relevant confounding variables (such as gender, income or age) that are motivated by either theory or practice. For example, the confounders can emerge from well-defined institutional rules that constrain the assignment of slots to treatment or the stratifying variables in experiments with perfect compliance. \nameref{assumptxt:confounders} is the same assumption discussed by \citet{rosenbaum1983central}, which justifies propensity score analysis.

\begin{assumptxt}[Dyadic Network]\text{ } Suppose that there exists an unobserved vector of pair-specific shocks $\{U_{ijg}\}_{i,j=1}^{N_g} \in \mathbb{R}^{N_g}$ for $g = 1,\ldots, G$ and an \textit{unknown} link function $\mathcal{L}:\mathbb{R}^{k_c}\times \mathbb{R}^{k_c}\times \mathbb{R}^{k_\Psi} \times \mathbb{R} \to \{0,1\}$ such that 
	\begin{enumerate}[(i)]
		\setlength\itemsep{-0.3em}
		\item \quad (Pairwise Links) $A_{ijg} = \mathcal{L}(C_{ig},C_{jg},\Psi_g^*,U_{ijg})$.
		\item \quad (Shocks) $U_{ijg}$ are i.i.d. and mutually independent of $\{\tau_{ig},D_{ig},C_{ig}\}_{i=1}^{N_g}$ given $\Psi_g^*$.
	\end{enumerate}
	\label{assumptxt:dyadicnetwork}
\end{assumptxt}

The \nameref{assumptxt:dyadicnetwork} assumption states that friendships between pairs of individuals $\{ig\}$ and $\{jg\}$ depend on their observed characteristics $(C_{ig},C_{jg})$, a group component $\Psi_g^*$ and a pair-specific shock $U_{ijg}$. For example, let $\Vert c - c^* \Vert$ be the Euclidean distance between two sets of covariates $(c,c^*)$. In a random geometric graph,  $\mathcal{L}(c,c,\Psi^*,u) = \mathbbm{1}\{ \Vert c - c^* \Vert \le u \}$, which implies that individuals are more likely to be friends if their characteristics are similar. In economics, dyadic networks have been used to analyze risk sharing agreements, political alliances and business partnerships \citep{graham2017econometric, fafchamps2007risk,attanasio2012risk,lai2000democracy, fafchamps2018networks}. The function $\mathcal{L}$ can be interpreted as a decision rule that encodes preferences over friends, as a random meeting process that brings two people together \citep{mele2017structural}, or a combination of both.

Dyadic networks can also be motivated as reduced form objects by appealing to exchangeability. In two influential papers, \cite{aldous1981representations} and \cite{hoover1979relations} showed that any network whose distribution is invariant to the ordering of the sample (exchangeability) can be represented as a dyadic network, where some of the components of $C_{ig}$ are possibly unobserved. From a practical point of view, the \nameref{assumptxt:dyadicnetwork} assumption states that the relevant determinants are indeed observed by the researcher. Therefore it can be interpreted as a network analog of the \nameref{assumptxt:confounders} assumption.

\subsection{Control Variable Results}

I show that $V_{ig} = (C_{ig},\Psi_g^*)$ satisfies the key unconfoundedness condition of Theorem \ref{lem:randomcoef_identification} and can be used as a matching variable to compute the average partial effect.

\begin{thm}[Direct Confounders]
	\label{thm:fullcontrolfunction}	
	Suppose that $Y_{ig}$ is generated by \eqref{eq:linearoutcome}. If \nameref{assumptxt:randomsampling}, \nameref{assumptxt:confounders} and \nameref{assumptxt:dyadicnetwork} hold,  then (i) $(X_{ig},L_{ig}) \ \indep \ \tau_{ig} \mid C_{ig},\Psi_g^*$, and (ii) $\mathbb{P}(X_{ig} \le x, L_{ig} \le \ell, Y_{ig} \le y \mid C_{ig},\Psi_g^*)$ does not depend on $\{ig\}$.
\end{thm}

Theorem \ref{thm:fullcontrolfunction} suggests that the researcher should include all of $\{ig\}$'s covariates that she considers relevant for treatment participation and network formation in $V_{ig}$. However, if the assumptions of \ref{thm:fullcontrolfunction} hold, then there is no need to control for the covariates of others. The variables $(C_{ig},\Psi_g^*)$ control for $\{ig\}'s$ friend preferences, and hence all the residual variation in $X_{ig}$ is exogenous. In practice, observing $(C_{ig},\Psi_g^*)$ may be a strong requirement and I propose some ways to relax the assumptions in Section \ref{weakcontrols}.

Intuitively, \nameref{assumptxt:randomsampling} and  \nameref{assumptxt:dyadicnetwork} imply that $(C_{ig},\Psi_g^*)$ controls for others' treatment whereas \nameref{assumptxt:confounders} ensures that it controls for own selection. That means that $\{D_{jg},C_{jg},N_g\}_{j \ne i},\{A_{jj'g}\}_{j \ne j'} \indep \tau_{ig} \mid C_{ig},\Psi_g^*$. Property (i) holds because $X_{ig}$ is a function of the left-hand side terms. Property (ii) by \nameref{assumptxt:randomsampling} and because \nameref{assumptxt:exchangeableinteractions} ensures that function $\varphi(\cdot)$ is exchangeable. Property (ii) is particularly important because it shows that if $(C_{ig},\Psi_g^*)$ is observed, then we can identify the conditional distribution of $(Y_{ig},X_{ig})$ by pooling different observations. If exchangeability did not hold, this probability would be $\{ig\}$--specific and the quantity $\mathbf{Q}_{xx}(v)$ in Theorem \ref{lem:randomcoef_identification} would not be well-defined.

\section{The Network Propensity Score}
\label{networkpropensityscore}

Most of the recent literature has focused on a particular type of exchangeable interactions, known as anonymous interference \citep{aronowsamii2017estimating,leung2019treatment,spilloversnoncompliance,forastiere2020identification,liu2019doubly,sussman2017elements,tchetgen2017auto}. This assumes that individual $\{ig\}$ is only affected by the total number (or proportion) of treated friends. In this case neighbors are given equal weight, regardless of their characteristics, and no-weight is put on higher connections. This type of assumption is reasonable in cases where the spillovers are local and there is no reason to believe that one particular friend has more influence than others.

With a slight abuse of notation, let $\varphi(t,\ell)$ be an exposure function that depends on the total number of treated friends $t$ and the total number of friends $\ell$,
\begin{assumptxt}[Anonymous interference]
	\label{assumptxt:anonymousinteractions}For all $i\in \{1,\ldots,N_g\}$, $$\varphi(i,A_g,D_g,C_g,N_g) = \varphi\Bigg(\underbrace{\sum_{j \ne i}A_{jg}D_{jg}}_{T_{ig}},\underbrace{\sum_{j \ne i}A_{jg}}_{L_{ig}}\Bigg).$$	
\end{assumptxt}
By design, \nameref{assumptxt:anonymousinteractions} is a special case of \nameref{assumptxt:exchangeableinteractions}. This added structure will lead to a lower-dimensional control variable. Define the propensity score and the friend propensity score, respectively as
\begin{align*}
p_{dig} &\equiv  \mathbb{P}(D_{ig} = 1 \mid C_{ig},\Psi_g^*),	\\
p_{fig} &\equiv  \mathbb{P}(D_{jg} = 1 \mid C_{ig},\Psi_g^*,A_{ijg} =1).
\end{align*}
The scalar $p_{dig}$ is the probability of treatment given individual characteristics, whereas $p_{fig}$ is the probability that a potential friend is treated. The \nameref{assumptxt:randomsampling} assumption ensures that every friend is ex-ante identical and hence the probability does not depend on the subscript $\{jg\}$. I call the three dimensional vector $P_{ig} = (p_{dig},p_{fig},L_{ig})$ the \textit{network propensity score}. Before presenting the general results I focus on a special case where $\tau$ has a closed form expression. The following result in Theorem \ref{thm:closedformtau} is a special case of Theorem \ref{lem:randomcoef_identification}, by setting $V_{ig} = (C_{ig},\Psi_g^*,L_{ig})$ and imposing a particular set of basis functions.

\begin{thm}[Closed form $\tau$]
	\quad  If $\varphi(t,l) = t/l$, $\mathcal{F} = \mathbbm{1}\{L_{ig} > 0\}$, $\mathbf{Q}_{xx}(V_{ig})$ is almost surely full rank and \nameref{assumptxt:randomsampling}, \nameref{assumptxt:confounders} and \nameref{assumptxt:dyadicnetwork} hold, then the average partial effects equal
	\begin{enumerate}[(i)]
		\item \quad $\alpha = \mathbb{E}\left[\left( 1 - \frac{T_{ig}-L_{ig}p_{fig}}{1-p_{fig}} \right)\left( \frac{(1-D_{ig})Y_{ig}}{1-p_{dig}}\right) \mid \mathcal{F} \right]$,
		\item \quad $\beta = \mathbb{E}\left[\left( 1 - \frac{T_{ig}-L_{ig}p_{fig}}{1-p_{fig}} \right)\left( \frac{D_{ig}Y_{ig}}{p_{dig}} - \frac{(1-D_{ig})Y_{ig}}{1-p_{dig}}\right) \mid \mathcal{F} \right]$,
		\item \quad $\gamma = \mathbb{E}\left[\left( \frac{T_{ig}-L_{ig}p_{fig}}{p_{fig}(1-p_{fig})} \right)\left(  \frac{(1-D_{ig})Y_{ig}}{1-p_{dig}}\right) \mid \mathcal{F} \right]$,
		\item \quad $\delta = \mathbb{E}\left[\left( \frac{T_{ig}-L_{ig}p_{fig}}{p_{fig}(1-p_{fig})} \right)\left( \frac{D_{ig}Y_{ig}}{p_{dig}} - \frac{(1-D_{ig})Y_{ig}}{1-p_{dig}}\right) \mid \mathcal{F} \right]$.		
	\end{enumerate}
	\label{thm:closedformtau}
\end{thm}
Theorem \ref{thm:closedformtau} shows that the average partial effects can be identified from $(p_{dig},p_{fig})$ and $(T_{ig},L_{ig},D_{ig},Y_{ig})$ for the subsample of individuals with at least one friend. The network propensity score is not observed directly but it can be identified from the data.

\begin{rem}[Ignoring spillovers in PSM]
The treatment effect $\beta$, in particular, looks very similar to its counterpart $\beta^{ATE}$ in the absence of spillovers. \cite{robins1994estimation} and many others have shown that

$$\beta^{ATE} = \mathbb{E}\left[ \frac{D_{ig}Y_{ig}}{p_{dig}} - \frac{(1-D_{ig})Y_{ig}}{1- p_{dig}} \mid \mathcal{F} \right].$$

By plugging in the outcome from \eqref{eq:linearoutcome}, and applying the law of iterated expectations, it is possible to show that $\beta^{ATE} = \beta + \mathbb{E}[p_{fig} \times \delta_{ig} \mid \mathcal{F} ]$. In the special case where the friend propensity score is independent of the spillover effect on the treated $(\delta_{ig})$, then this expression simplifies to $\beta + \mathbb{E}[D_{jg} \mid \mathcal{F}] \times \delta$. That means that the treatment effect that is recovered from traditional propensity score matching can be interpreted for the average effect when $\mathbb{E}[D_{jg} \mid \mathcal{F}]$ friends are treated. This quantity is not directly policy relevant because it does not reflect the average outcome when the program is implemented at a smaller or larger scale. In Appendix \ref{spuriouseffects}, I explore an example where OLS identifies spurious effects, even in the absence of spillovers.

\end{rem}

\begin{rem}[Rank conditions]
The example in Lemma \ref{lem:conditionaldist} also highlights some of the relevant rank conditions for identification that hold for more general settings. As in standard propensity score matching the overlap condition $0 < p_{dig} < 1$ needs to hold, otherwise the denominator is not well defined. There is a similar overlap condition for potential friends, where $0 < p_{fig} < 1$. This means that $\{ig\}'s$ friend cannot all be part of the treatment or control with probability approaching one. Otherwise, there is no residual variation to identify the spillover effects. Lastly, the distribution of $(T_{ig},L_{ig})$ needs to have thin tails (not too many friends), otherwise expectation may not be well defined. This suggests a potential weak identification problem in dense network limits where $L_{ig} \to \infty$. This is not a problem for networks with a bounded number of friends.
\end{rem}

\begin{lem}[Conditional Distribution]
	If \nameref{assumptxt:randomsampling} and \nameref{assumptxt:dyadicnetwork}, then (i)  $D_{ig} \mid T_{ig},L_{ig},C_{ig},\Psi_g^* \sim  Bernoulli(p_{dig})$, and (ii) $T_{ig} \mid L_{ig},C_{ig},\Psi_g^* \sim Binomial(p_{fig},L_{ig})$.
	\label{lem:conditionaldist}
\end{lem}
Lemma \ref{lem:conditionaldist} shows that the distribution of $(D_{ig},T_{ig})$ given $(C_{ig},\Psi_g^*)$ can be parametrized in terms of $P_{ig}$. Part (i) is an extension of the canonical result of \cite{rosenbaum1983central}, whereas par (ii) is a new result. This factorization holds regardless of the primitive function $(\mathcal{L})$ and shock distribution of network formation. The proof builds on the insight that $T_{ig}$ is a sum of conditionally independent Bernoulli variables after conditioning on the key variables of network formation. Under model \eqref{eq:linearoutcome}, $X_{ig}$ is a deterministic function of $(D_{ig},T_{ig},L_{ig})$ which means that $P_{ig}$ also parametrizes the distribution of $X_{ig} \mid C_{ig},\Psi_g,L_{ig}$.

\begin{thm}[Balancing] If \nameref{assumptxt:randomsampling} and \nameref{assumptxt:dyadicnetwork} hold, then $P_{ig}$ is a balancing score, in the sense that $X_{ig} \ \indep \ (C_{ig},\Psi_g^*) \mid P_{ig}$. If \nameref{assumptxt:confounders} also holds, then $X_{ig} \ \indep \ \tau_{ig} \mid P_{ig}$.
	\label{thm:expostbalancing}
\end{thm}
Theorem \ref{thm:expostbalancing} shows that $P_{ig}$ is a suitable generalization of the propensity score to setting with spillovers and network formation by showing that inherits two key properties. First, it is a balancing score which means that two individuals with the same value of $P_{ig}$ are guaranteed to have the same distribution of covariates $(C_{ig},\Psi_g^*)$. This property is important for causal analyses because it ensures that any matching procedure based on $P_{ig}$ will compare similar individuals. Second, it shows that $P_{ig}$ satisfies the unconfoundedness property required to identify the average partial effect $\tau$ in Theorem \ref{lem:randomcoef_identification}. The selection on observables ties the observed characteristics $(C_{ig},\Psi_g^*)$ to the random coefficients and is therefore crucial to prove the final step.

\subsection{A Structural Interpretation}

From an economic point of view, the network propensity score can be interpreted as a function of agents' underlying preferences. To this end, it is convenient to represent $\{ig\}'s$ treatment indicator as $D_{ig} = \mathcal{H}(C_{ig},\Psi_g^*,\eta)$ where $\mathcal{H}$ is a measurable function and $\eta_{ig} \mid C_{ig},\Psi_g^* \sim F(\eta \mid c,\Psi^*) $ is an unobserved participation shock. Since we can always define the participation shock as $\eta = D_{ig} - \mathbb{P}(D_{ig} =1 \mid C_{ig} = c,\Psi_g^* = \Psi^*)$, this form does not entail any loss of generality. The function $\mathcal{H}$ can also take the form of a threshold utility model or an institutional assignment rule based on observables. The first component of the network propensity score is the propensity score conditional on $(C_{ig},\Psi_g^*)$, which is defined as
\begin{align}
\begin{split}
\mathbb{P}(D_{ig} = 1 \mid C_{ig} = c,\Psi_g^* = \Psi^*) =  \int \mathcal{H}(c,\Psi^*,\eta) \ dF(\eta \mid c,\Psi^*)
\end{split}
\label{eq:pscore}
\end{align}
The propensity score depends on the preference function $\mathcal{H}$ and the distribution of selection shocks. The integral averages out the individual heterogeneity $\eta$, holding the characteristics $(c,\Psi^*)$ fixed.

The friend propensity score can be written in a similar way. Let $F(\eta^*,c^*,u \mid \Psi^*)$ be the distribution of traits of a potential friend in each group $(\eta^*,c^*)$ and the friendship shock $(u)$ given $\Psi_g^*$. By Bayes' rule
\begin{align}
\begin{split}
&\mathbb{P}(D_{jg} = 1 \mid C_{ig} = c,\Psi_g^* = \Psi^*,A_{ijg} = 1) \\
&= \int \mathcal{L}(c,c^*,\Psi^*,u)\ \mathcal{H}(c^*,\Psi^*,\eta^*) \ \frac{ dF(\eta^*,c^*,u \mid \Psi^*)}{\int \mathcal{L}(c,c^*,\Psi^*,u) \ dF(c^*,u^* \mid \Psi^*)}.
\end{split}
\label{eq:explanationfscore}
\end{align}

The friend propensity combines $\{ig\}'s$ friendship preferences/meeting likelihood and $\{jg\}'s$ preferences for participation in the program. In the extreme case that $L = \mathbbm{1}\{c = c^*\}$, agents only befriend others with exactly the same characteristics and the friend propensity score is equal to the propensity score. At the other extreme, when $L = \mathbbm{1}\{u > 0\}$ the network is exogenous then \eqref{eq:explanationfscore} reduces to $\int \mathcal{H}(c^*,\Psi^*,\eta) dF(\eta^*,c^* \mid \Psi^*)$, which is a group-level constant. Conversely, when the treatment is exogenous, that is when $\mathcal{H}(c^*,\Psi^*,\eta^*) = \eta$ and $\eta$ is independent of the other characteristics, then the propensity score and the friend propensity score are constant. For intermediate cases the friend propensity score will not contain the same information as the propensity score.

In the microfinance example, homophily suggests people tend to associate with people in the same caste whereas selection implies that certain castes are more likely to participate in the information sessions. Consequently, the likelihood of having a treated friend depends on a household's caste. This is where homophily interacts with selection. Pairs of friends tend to have similar traits $(c,c^*)$ on average and consequently similar partipation probabilities, via the function $h$.

\subsection{Relationship to Network Pseudo-Metrics}

One potential problem with identification via network propensity score matching is that there may be some latent confounding variables. In some cases, researchers may be able to recover some unobservables from the network structure. In this section I explain the relationship between graphon notions of network distance and the network propensity score. 

\cite{auerbach2022identification} defines a pseudo-metric as follows\footnote{A similar metric was used by \cite{zeleneev2020identification} to identify network formation models with interactive fixed effects.}
\begin{equation} d_{\Psi_g^*} = \left[ \int (\mathcal{L}(C_{ig},C^*)-\mathcal{L}(C_{jg},C^*))^2 dF(C^* \mid \Psi_g^*)\right]^{1/2}.
\end{equation}
The pseudo-metric $d_{\Psi_g^*}$ measures measures the $L_2$ difference in link functions between someone with characteristics $C_{ig}$ and $C_{jg}$. Two individuals are close in this sense, if they have the same revealed preferences for neighbors, at least in reduced-form. In practice this can be computed by counting the relative number of friends in common between $\{ig\}$ and $\{jg\}$. \cite{auerbach2022identification} proposed matching estimators for a partially linear model with an exogenous treatment based on the pseudo-metric, but noted that it could not be used to study spillovers because of the failure of a particular rank condition. I show that this finding extends to a more general setting with arbitrary functional form and selection-on-observables by establishing its connection to the friend propensity score.

The difference in the friend propensity scores of two individuals is equal to
\begin{equation}
d_f =  \left\Vert \int \mathcal{H}(C^*,\Psi_g^*)\left[ \frac{\mathcal{L}(C_{ig},C^*,\Psi_g^*)}{p_\ell(C_{ig},\Psi_g^*)} - \frac{\mathcal{L}(C_{jg},C^*,\Psi_g^*)}{p_\ell(C_{jg},\Psi_g^*)}\right]dF(C^* \mid \Psi_g^*) \right\Vert.
\end{equation}
The following theorem establishes the relationship between the two metrics.
\begin{thm}[Bounds Pseudo-Metrics]
$d_{f}\le \frac{1}{p_{\ell}(C_{ig},\Psi_g^*)}d_{\Psi_g^*}$.
\label{thm:bounds_pseudometrics}
\end{thm}
Theorem \ref{thm:bounds_pseudometrics} states that if two individuals are close in the pseudo-metric then they also have identical friend propensity scores. In principle, if $d_{\Psi_g^*}$ were known, then researchers could match on the pseudo-metric rather than on the friend propensity score. The converse does not necessarily hold. For example, suppose that $\mathcal{H}(C^*,\Psi_g^*) = 0.5$, then $d_f = 0$ if and only if $\int \mathcal{L}(C_{ig},C^*,\Psi_g^*) = \int \mathcal{L}(C_{jg},C^*, \Psi_g^*)$. This occurs when two individuals have the same proportion of friends. However, $d_{\Psi_g^*} > 0$ if the $\{ig\}$ and $C_{jg}$ have different preferences over specific friends.

However, the pseudo-metric is not typically known. The problem is that current methods to consistently estimate $d_{\Psi_g^*}$ require a dense network with $p_{\ell}(C_{ig},\Psi_g^*) \to \infty$ in the large-network limit. This means that the number of friends $L_{ig}$ grows proportionately with the sample size, which violates the rank conditions to identify even simple estimands as those in Theorem \ref{thm:closedformtau}. In essence, two individuals that are close in the pseudo=metric will have almost identical values of the regressors $X_{ig}$, and there is no residual variation to identify the average partial effects. This makes it difficult to use network structure to recover unobserved heterogeneity.

\subsection{Weaker Controls and Mixture Representation of $\mathbf{Q}_{xx}$}
\label{weakcontrols}

To compute the network propensity score, $(C_{ig},\Psi_g^*)$ needs to be fully observed or be consistently estimated. Unobserved heterogeneity can be addressed in a variety of ways. The researcher may still be able to identify average partial effects, even in settings with unobserved heterogeneity, under further restrictions.
\begin{lem}[Weaker Control]
	\label{lem:quasisaturation}
	If \nameref{assumptxt:randomsampling}, \nameref{assumptxt:confounders}, \nameref{assumptxt:dyadicnetwork} hold, and $\tau_{ig} \ \indep \ P_{ig} \mid V_{ig}$, where $P_{ig} = (p_{dig},p_{fig},L_{ig})$, then $X_{ig} \ \indep \ \tau_{ig} \mid V_{ig}$.	
\end{lem}
Lemma \ref{lem:quasisaturation} provides a high-level condition stating that
any residual variation in the network propensity is exogenous after conditioning on $V_{ig}$. Since the network propensity score is itself a function of $(C_{ig},\Psi_g^*)$ this means that there are exogenous shifters in individual behavior $(C_{ig})$ or group contextual factors $(\Psi_g^*)$.

There are two examples in the literature that could satisfy this requirement.  \cite{johnsson2019estimation} propose a restriction on the class of network models that satisfy a particular monotonicity restriction, extending prior work by \cite{goldsmith2013social}. In this case, the total number of friends is a sufficient statistic for network unobservables. \cite{spilloversnoncompliance} find a single-dimensional control variable in experiments with spillovers and non-compliance, which is the share of individuals that accept treatment offers in each group. I fill in some of the details in Section \ref{experiments_illustration}. \cite{spilloversnoncompliance} find a sufficient statistic $V_{ig}$ exploiting the binomial form of their key endogenous variable. \cite{arkhangelsky2018role} show how to exploit this type of structure to identify direct effects in models with selection-on-observables and fixed effects. They rely on the idea of obtaining sufficient statistics of the unobserved heterogeneity. Similar strategies could be applied in the spillovers case by imposing particular functional forms on the network or selection processes.

Imposing the conditions of Lemma \ref{lem:quasisaturation} has implications for the structure of the matrix $\mathbf{Q}_{xx}(v)$ defined in Theorem \ref{thm:fullcontrolfunction}. Define the functions $\widetilde{\varphi}_1(p_f,l) = \mathbb{E}[\varphi(T_{ig},L_{ig}) \mid p_{fig} = p_f,L_{ig} = l]$ and $\widetilde{\varphi}_2(p_f,l) = \mathbb{E}[\varphi(T_{ig},L_{ig})\varphi(T_{ig},L_{ig})' \mid p_{fig} = p_f,L_{ig} = l]$ which are the conditional first and second moments given the friend propensity score and the total number of friends. Since Lemma \ref{lem:conditionaldist} shows that $(p_{fig},L_{ig})$ parametrizes the distribution of $(T_{ig},L_{ig})$ given $(C_{ig},\Psi_g^*)$, these are equivalent to conditioning on $(C_{ig},\Psi_g^*)$ directly by the decomposition axiom \citep{Constantinou2017}. Lemma \ref{lem:conditionaldist} also implies that $\widetilde{\varphi}_1$ and $\widetilde{\varphi}_2$ are \textit{known} functions that only change depending on the  basis $\varphi$. In our running example, where $\varphi(t,l) = t/l$ these function take a very simple form. In this case $\widetilde{\varphi}_1(p_f,l)$ and $\widetilde{\varphi}_2(p_f,l) = \frac{p_f(1-p_f)}{l} + p_f^2$.

Lemma \ref{lem:mixturepresentationQxx} shows that the matrix $\mathbf{Q}_{xx}$ can be expressed as a mixture of known functions of the network propensity score.

\begin{lem}[Mixture Representation]
	Suppose that \nameref{assumptxt:randomsampling} and \nameref{assumptxt:dyadicnetwork} hold, and that $V_{ig}$ is measurable with respect to $(C_{ig},\Psi_g^*,L_{ig})$, then
	\begin{align}
		\mathbf{Q}_{xx}(v) = \int \begin{pmatrix} 1 & \widetilde{\varphi}_1(p_{f},l)' \\ \widetilde{\varphi}_1(p_{f},l,p_{f},l) & \widetilde{\varphi}_2(p_{f},l) \end{pmatrix} \otimes \begin{pmatrix} 1 & p_d \\ p_d & p_d \end{pmatrix} dF(p_{d},p_{f},l \mid V_{ig} = v).		
	\end{align}
	\label{lem:mixturepresentationQxx}	
\end{lem}
In the special case where $V_{ig} = (C_{ig},\Psi_g^*)$ the distribution $F$ is degenerate and we can drop the integral sign. Therefore, observing the key variables for selection and network formation imposes over-identifying restrictions on the weighting matrix. The integral is non-degenerate when some of these key variables are unobserved by the researcher. This assumption is testable by comparing the entries of $\mathbf{Q}_{xx}$. For example, in a parametric model $F$ can be modeled as a latent distribution that nests the degenerate case and $(p_{dig},p_{fig})$ as link function such as probit or logit. In the empirical example, I exploit this structure to produce a feasible parametric model.


\section{Estimation}
\label{estimation}

I outline a two-step procedure to estimate the causal effects for linear models as a sample analog of the estimand of $\tau$. In the first stage, I fit a parametric model for $\mathbf{Q}_{xx}$ using data from the endogenous regressors $X_{ig}$ and the control variable $V_{ig}$. In the second stage, I substitute the estimated weighting matrix $\mathbf{Q}_{xx}$ to compute $\tau$ by inverse weighting.

\textbf{Notation:} Let $Z_{ig}$ denote a vector of individual variables, where $Z_{ig} \equiv (X_{ig},Y_{ig},V_{ig})$ includes the endogenous regresors, the outcome and the observed control variables. I let $\sum_{ig} f(Z_{ig})$ be the sum $\sum_{g=1}^G\sum_{i=1}^{N_g} f(Z_{ig})$, where $f(\cdot)$ is an arbitrary function. I also let $\bar{n} = \frac{1}{G}\sum_{g=1}^G N_g$ denote the average group size. By construction, $\bar{n}G$ is equal to the total sample size. For convenience, let $vec(\cdot)$ denote the vectorize operator, which stacks the columns of a matrix into a single vector. I also use $\Vert x \Vert $ to denote the Euclidean norm of the vector $x$, defined as $\Vert x \Vert = \sqrt{\sum_{k=1}^{K} x_{k}^2}$.

In the first stage, I consider a parametric class of functions to model the weighting matrix, $\{\mathbf{Q}_{xx}(v,\boldsymbol{\theta}): \boldsymbol{\theta} \in \boldsymbol{\Theta} \subseteq \mathbb{R}^{d_{\theta}} \}$, that nest the true model. This means that there is a $\boldsymbol{\theta}_0 \in \boldsymbol{\Theta}$ such that $\mathbf{Q}_{xx}(v,\boldsymbol{\theta}_0) = \mathbb{E}[X_{ig}X_{ig}' \mid V_{ig} = v]$. The matrix $\mathbf{Q}_{xx}$ has to be symmetric and positive semi-definite. If \nameref{assumptxt:randomsampling}, \nameref{assumptxt:confounders} and \nameref{assumptxt:dyadicnetwork} hold, and $V_{ig} = (C_{ig},\Psi_g^*)$ the choice of parametric family can be disciplined by imposing over-identifying restrictions of the network formation model, so that $\mathbf{Q}_{xx}(v,\boldsymbol{\theta})$ can be expressed as a function of the network propensity score. Alternatively we can use the mixture model representation of Lemma \ref{lem:mixturepresentationQxx} to inform the choice of $\mathbf{Q}_{xx}$ for other choices of $V_{ig}$. The control variable $V_{ig}$ is valid as long as the conditions of Lemma \ref{lem:quasisaturation} hold.

I define the vectorized residuals,
\[
r(Z_{ig},\boldsymbol{\theta}) \equiv \text{vec}(X_{ig}X_{ig}' - \mathbf{Q}_{xx}(V_{ig},\boldsymbol{\theta})).
\]
The residuals capture how well the control variables fit $X_{ig}$. The sample criterion function computes the average of square residuals as
\begin{equation}
\label{eq:samplecriterion}
\widehat{\mathcal{R}}(\boldsymbol{\theta}) \equiv \ \frac{1}{\bar{n}G}\sum_{ig} \Vert r(Z_{ig},\boldsymbol{\theta})\Vert^2.
\end{equation}

The sample criterion $\widehat{\mathcal{R}}(\boldsymbol{\theta})$ is an  approximation to $\mathcal{R}(\boldsymbol{\theta}) = \mathbb{E}[\Vert r(Z_{ig},\boldsymbol{\theta})\Vert^2]$. The least squares criterion is appropriate for three reasons. First, the population criterion $\mathcal{R}(\boldsymbol{\theta})$ is minimized at $\boldsymbol{\theta}_0$ because the conditional mean of $X_{ig}X_{ig}'$ given $V_{ig}$ is the optimal prediction. This provides a rationale for minimizing $\widehat{\mathcal{R}}(\boldsymbol{\theta})$. Second, joint-likelihood approaches are either impractical or infeasible without strong assumptions, particularly with more complex exposure functions. Third, quasi-likelihood approaches, such as those in \cite{tchetgen2017auto} and \cite{sofrygin2017semi} are valid under certain assumptions, but are more sensitive to the specification of the model. My approach is more robust than quasi-likelihood methods because it targets the conditional mean directly, which is the main object required for identification.

We can construct a feasible estimator by minimizing the sample criterion,
\begin{equation}
\label{eq:estimator_networkpropensity}
\widehat{\boldsymbol{\theta}}  = \arg\min_{\boldsymbol{\theta} \in \boldsymbol{\Theta}} \  \widehat{\mathcal{R}}(\boldsymbol{\theta}).
\end{equation}
The estimated parameter $\widehat{\boldsymbol{\theta}}$ can be plugged-in to compute a feasible weighting matrix $\mathbf{Q}_{xx}(V_{ig},\widehat{\boldsymbol{\theta}})$. I propose the following sample analog of the inverse-weighting estimand of $\tau$. 

\[
\widehat{\boldsymbol{\tau}} \equiv \frac{1}{\bar{n}G} \sum_{ig} \mathbf{Q}_{xx}(V_{ig},\widehat{\boldsymbol{\theta}})^{-1}X_{ig}Y_{ig}
\]
The vector $\widehat{\boldsymbol{\tau}}$ is a feasible estimator of the average partial effects defined in \eqref{eq:linearoutcome}. The estimator is subject to two sources of uncertainty. First, the sample average is an approximation to $\mathbb{E}[\mathbf{Q}_{xx}(V_{ig})^{-1}X_{ig}Y_{ig}]$. Second, the inverse weighting method is subject to first-stage uncertainty in the  estimation of $\widehat{\boldsymbol{\theta}}$. Under standard regularity conditions that I list in the Appendix, $\widehat{\boldsymbol{\theta}}$ and $\widehat{\boldsymbol{\tau}}$ are consistent but the standard errors need to be adjusted. This is analogous to the first stage uncertainty in propensity score methods, that can be corrected analytically or by bootstrap procedures \citep{abadie2016matching}.

To adjust the standard errors it is useful to view the first and second stages as a single system of equations. As before, let $z \equiv (x,y,v)$. I write down the first-order conditions in terms of the jacobian of the square residuals $\psi_q(z,\boldsymbol{\theta}) = \tfrac{\partial}{\partial \boldsymbol{\theta}'}\Vert r(v,\boldsymbol{\theta})\Vert^2$ and the second stage influence function $\psi_{IW}(z,\boldsymbol{\theta}) = \mathbf{Q}_{xx}(v,\boldsymbol{\theta})$. I stack the first and second stage equations in a single influence function $\psi \equiv [\psi_q,\psi_{IW}']'$. The estimated parameters solve
\begin{equation}
\frac{1}{\bar{n}G}\sum_{ig} \psi(Z_{ig},\widehat{\boldsymbol{\tau}},\widehat{\boldsymbol{\theta}}) = \boldsymbol{0}
\label{eq:sampleinfluence}
\end{equation}
To this end, I define the within-group average $\overline{\psi}_g(\boldsymbol{Z}_{g},\boldsymbol{\theta}) \equiv \frac{1}{N_{gt}}\sum_{i=1}^{N_{g}} \psi(Z_{ig},\boldsymbol{\theta})$, where $\boldsymbol{Z}_{g} \equiv \{ Z_{ig}\}_{i=1}^{N_{gt}}$ is a matrix of individual covariates for each group. This allows me to decompose \eqref{eq:sampleinfluence} into group averages as $\frac{1}{\bar{n}G}\sum_{ig} \psi(Z_{ig},\widehat{\boldsymbol{\tau}},\widehat{\boldsymbol{\theta}})  = \frac{1}{G}\sum_{g=1}^G \left(\tfrac{N_{g}}{\bar{n}}\right) \overline{\psi}_g(\boldsymbol{Z}_{g},\boldsymbol{\theta})$. The fraction $(N_g/\bar{n})$ denotes the relative size of each group.

For inference, I compute heteroskedasticity-robust standard errors, clustered at the group level. Let $\widehat{\Omega}$ be an estimate of the second moments of the influence function \eqref{eq:sampleinfluence} and let  $\widehat{H}$ be a sample analog of the expected jacobian, defined as
\begin{align}
\widehat{H} &\equiv \frac{1}{G}\sum_{g=1}^G \left(\tfrac{N_{g}}{\bar{n}}\right) \frac{\partial}{\partial (\boldsymbol{\theta},\boldsymbol{\beta})'} \overline{\psi}_g(\boldsymbol{Z}_g,\widehat{\boldsymbol{\tau}},\widehat{\boldsymbol{\theta}}) \\
\widehat{\Omega} &\equiv \frac{1}{G}\sum_{g=1}^G \left(\tfrac{N_{g}}{\bar{n}}\right)^2 \overline{\psi}_g(\boldsymbol{Z}_{g},\widehat{\boldsymbol{\tau}},\widehat{\boldsymbol{\theta}}) \overline{\psi}_g(\boldsymbol{Z}_{g},\widehat{\boldsymbol{\tau}},\widehat{\boldsymbol{\theta}})'
\end{align}
Then the covariance of the estimators is computed by the sandwich form $\widehat{\Sigma} \equiv \tfrac{1}{G}\widehat{H}^{-1}\widehat{\Omega}\widehat{H'}^{-1}$ and the standard errors can be recovered from the square root of the diagonal of $\widehat{\Sigma}$. Since the estimator $\widehat{\tau} \in \mathbb{R}^{d_{\tau}}$ only enters the second stage linearly,
\[
\widehat{H} = \frac{1}{G}\sum_{g=1}^G \left(\tfrac{N_{g}}{\bar{n}}\right) \begin{pmatrix} \frac{\partial}{\partial \boldsymbol{\theta}'} \overline{\psi}_{q,g}(\boldsymbol{Z}_g,\widehat{\boldsymbol{\tau}},\widehat{\boldsymbol{\theta}}) & \boldsymbol{0}'_{d_{\tau}} \\
\frac{\partial}{\partial \boldsymbol{\theta}'} \overline{\psi}_{IW,g}(\boldsymbol{Z}_g,\widehat{\boldsymbol{\tau}},\widehat{\boldsymbol{\theta}}) & I_{d_{\tau}}
\end{pmatrix}
\]
Here, $\overline{\psi}_{q,g}$ and $\overline{\psi}_{IW,g}$ decompose the within-group average influence functions into the first and second stages, respectively. Both $\widehat{H}$ and its inverse are lower triangular, which means that the limiting covariance matrix of $\tau$ depends on the upper-left block of $\widehat{\Omega}$ (which captures the first-stage uncertainty).

\subsection{Large Sample Theory}

For the remainder of this section I propose inference procedures for a setting with many groups $G \to \infty$ and allow for the possibility that $N_g$ is either fixed or growing with $G$. This is intended to approximate the situation faced by empirical researchers who randomly collect data from distinct geographic units, with few individuals (classrooms) or many individuals (villages, cities), which matches the data that I use in the empirical example. Formally, I assume that there is a sequence of probability distributions that is indexed by $t$, with $G_t$ groups of unequal size $N_{gt}$, and let $N_t \equiv \mathbb{E}[N_{gt}]$ denote the expected group size. There is a triangular array of covariates for individual $\{ig\}$ for the point $t$ in the sequence, which I denote by $Z_{igt} = (X_{igt},Y_{igt},V_{igt})$. The variables $(L_{igt},T_{igt})$ are the number of treated friends and number of friends, respectively. Similarly, for each $t$, I compute estimators $(\widehat{\boldsymbol{\theta}}_t,\widehat{\boldsymbol{\tau}}_t)$. The estimator $\widehat{\boldsymbol{\tau}}_t$, in particular is compared to the population quantity $\boldsymbol{\tau}_{0t} = \mathbb{E}[\tau_{igt} \mid \mathcal{F}_t]$. Centering the estimator around the mean of the triangular array is important to derive the right rate of convergence. For simplicity, I define $\rho_{gt}$ as the relative group size. Let $0 < \underline{\rho} < \overline{\rho} < 1$ be an arbitrary constant that I use throughout the derivation.

\begin{assumptxt}[Bounded Group Ratios]
	\label{assumptxt:boundedrelative_n}
	\quad $\rho_{gt} \equiv (N_{gt}/N_t) \in [\underline{\rho}, \overline{\rho}] \subset (0,1)$ almost surely.
\end{assumptxt}
\nameref{assumptxt:boundedrelative_n} implies that all groups are approximately the same size, within a range. It implies that the ratio of the largest to the smallest group is bounded by $\overline{\rho}/\underline{\rho}$. This assumption is automatically satisfied when $N_{gt}$ is bounded. However, if $N_{t} \to \infty$ as $t \to \infty$, then the assumption implies that the smallest group size is growing, because $\inf_{g} N_{gt} \ge \underline{\rho}N_t \to \infty$ as $N_t \to \infty$. \cite{bester2016grouped} propose a weaker assumption for large unbalanced panels, where the bounds hold in the limit experiment rather than for each point along the sequence, which leads to qualitatively similar conclusions.

My asymptotic results allow for some or all of the regressors in $V_{igt}$ to be estimated. For example, \cite{johnsson2019estimation} show the estimator $L_{ig}/(N_g-1)$ converges uniformly to a measure of unobserved degree heterogeneity in dense networks, at rate $\sqrt{(\log N_t) / N_t}$ in sup-norm. In related work, \cite{spilloversnoncompliance} find that in randomized experiments with non-compliance, the key dimensions of heterogeneity in spillover models is unobserved but can be consistently estimated in large groups, with a $\sqrt{(\log G_t) /N_t}$ uniform rate of convergence. Finally, researchers may also want to estimate group-level averages of the covariates that are consistent in large groups. I define $V_{igt}^{0}$ as the true, but unobserved value of the regressors. My asymptotic results simply require that $\max_{g=1,\ldots,G_t}\max_{i=1,\ldots,N_{gt}} \Vert V_{igt} - V_{igt}^{0} \Vert = O_p(\lambda_t)$ and that $\sqrt{G_t}\lambda_t = o(1)$. In the two examples above, this means that the expected size of each group needs to be large relative to the number of groups. This is plausible in situations where data is collected on large villages or other geographical units. If the key confounders are observed without error then $V_{igt} = V_{igt}^0$. Otherwise the condition holds trivially and $N_t$ does not need to grow with $G$ at any particular rate.

I list additional \nameref{assumptxt:regularity} in the Appendix, where I impose conditions on the moments of $(X_{igt},Y_{igt})$ and smoothness conditions on the function $\mathbf{Q}_{xx}(\cdot,\boldsymbol{\theta})$. In particular, I provide conditions that ensure that the weighting matrix is almost surely invertible, by imposing a lower bound on the eigenvalues of the matrix. When $V_{igt} = (C_{igt},\Psi_{gt}^*)$ and \nameref{assumptxt:randomsampling}, \nameref{assumptxt:confounders} and \nameref{assumptxt:dyadicnetwork} hold, this is equivalent to saying that $L_{igt}$ is bounded, and that the remaining components of the network propensity score are bounded in a compact subset of the unit interval, i.e. $p_d(C_{igt},\Psi_{gt}^*),p_f(C_{igt},\Psi_{gt}^*) \in [\underline{\rho},\overline{\rho}] \subset (0,1)$. That avoids boundary cases, where there is not enough residual variation in the regressors after conditioning on the controls. Finally, I define two more objects, $H_{0t} \equiv \mathbb{E}[(N_{gt}/N_t)\overline{\psi}_g(\mathbf{Z}_{gt},\widehat{\boldsymbol{\tau}},\widehat{\boldsymbol{\theta}})]$ and $\Omega_{0t} \equiv \mathbb{E}[(N_{gt}/N_t)^2\overline{\psi}_g(\mathbf{Z}_{gt},\widehat{\boldsymbol{\tau}},\widehat{\boldsymbol{\theta}})\overline{\psi}_g(\mathbf{Z}_{gt},\widehat{\boldsymbol{\tau}},\widehat{\boldsymbol{\theta}})']$ that are used to compute the covariance matrix $\Sigma_t \equiv H_{0t}^{-1}\Omega_{0t}H_{0t}^{-1}$.

\begin{thm}[Limiting Distribution Estimators]
	\label{thm:normality_averageeffects}
	Suppose that $V_{igt}$ satisfied the conditions of Theorem \ref{lem:randomcoef_identification} and define the covariance matrix $\Sigma_t \equiv H_{0t}^{-1}\Omega_{0t}H_{0t}^{-1}$. If  \nameref{assumptxt:boundedrelative_n} and  \nameref{assumptxt:regularity} hold, then as $t \to \infty$,  (i) $\widehat{\boldsymbol{\theta}}_t \to^{p} \boldsymbol{\theta}_{0t}$, $\widehat{\boldsymbol{\tau}}_t \to^{p} \boldsymbol{\tau}_{0t}$ and (ii)
	\[
	\sqrt{G_t}\Sigma_t^{-1/2}\left( \begin{matrix} \widehat{\boldsymbol{\theta}}_t - \boldsymbol{\theta}_{0t} \\
	 \widehat{\boldsymbol{\tau}}_t - \boldsymbol{\tau}_{0t}
	 \end{matrix} \right) \to^d \mathcal{N}(\boldsymbol{0},I)
	\]
\end{thm}

Theorem \ref{thm:normality_averageeffects} shows that the estimators are consistent and converge to a normal distribution. The estimator is centered around the value of $(\boldsymbol{\theta}_{0t},\boldsymbol{\beta}_{0t})$ that solves the population criterion, at each point of the sequence. This allows for estimators that are consistent, even if the networks itself does not converge to any particular structure. Theorem \ref{thm:normality_averageeffects} can be viewed as an approximation to the finite sample behavior. Researchers can construct test statistics by substituting $\Sigma_t$ with a sample analog $\widehat{\Sigma}_t$ to confidence intervals.

My results are agnostic about the dependence structure across groups, but it may be possible to improve the $\sqrt{G_t}$ to $\sqrt{G_tN_t}$ under stronger conditions. For example, \cite{kojevnikov2020limit} develop a central limit theorem for network dependence and provide specific regularity conditions for a single  \nameref{assumptxt:dyadicnetwork}. This requires the network to be sparse $L_{igt}$ small relative to $N_{gt}$ so that individuals far apart in the network are approximately independent. In practice, this does not change the estimation procedure but rather the way in which we construct confidence intervals. \cite{kojevnikov2020limit} propose a Network-HAC estimator and \cite{kojevnikov2019bootstrap} proposes a bootstrap procedure. \cite{leung2019weak} proposes similar limiting theory for spillover effects when the treatment is exogenously assigned, and \cite{chandrasekhar2014tractable} propose alternative limit theorems under network dependence.

\subsection{Covariate Balancing (``Placebo'') Test}

The balancing property in Theorem \ref{thm:expostbalancing} is testable. Parametric propensity score analyses typically conduct so-called covariate balancing tests. I propose an analogous ``placebo'' test, where the pretreatment covariates serve as an outcome variable. Let $\widetilde{V}_{ig}\in \mathbb{R}$ be a variable in the covariate set $V_{ig} = (C_{ig},\Psi_g)$. My test relies on the simple idea that $\widetilde{V}_{ig}$ can be decomposed as
\[
\widetilde{V}_{ig} = \underbrace{\widetilde{V}_{ig}}_{\widetilde{\alpha}_{ig}} + 0 \times D_{ig} + 0 \times \varphi(i,A_g,D_g,C_g) +  0 \times D_{ig} \times \varphi(i,A_g,D_g,C_g)
\]
Let $\widetilde{\tau}_{ig} = (\widetilde{V}_{ig},0,0,0)$ is the vector of coefficients of the placebo outcome. It is easy to verify that $X_{ig} \ \indep \ \widetilde{\tau}_{ig} \mid V_{ig}$ since $\widetilde{\tau}_{ig}$ is a measuarable function of $V_{ig}$. Therefore by  Theorem \ref{lem:randomcoef_identification},   $\mathbb{E}[\mathbf{Q}_{xx}(V_{ig})^{-1}X_{ig}\widetilde{V}_{ig}] = (\mathbb{E}[\widetilde{V}_{ig}],0,0,0)$. Therefore, when $\mathbf{Q}_{xx}(v)$ is properly specified the researcher can test the null hypothesis that the slope coefficients are zero. This test only uses information about the treatment, the network and the covariates but not the outcome. In practice the test could be rejected in a parametric settings if the functional form is not flexible enough. However, it could also be rejected because a violation of the over-identifying restrictions imposed by the \nameref{assumptxt:randomsampling} and \nameref{assumptxt:dyadicnetwork} assumptions. The researcher may want to check whether there are omitted variables that might influence network formation or treatment.


\section{Empirical Examples}
\label{example}

\subsection{Political Participation in Uganda}

I evaluate the role of an intervention on political participation in Uganda  \citep{eubank2019viral,ferrali2020takes}. U-Bridge is a novel political communications technology that allows citizens to contact district officials via text-messages. In a pilot program, individuals in 16 villages were invited to participate in quarterly meetings, at a central location, where they received information about national service delivery standards and ways
to communicate with local officials. The Governance, Accountability, Participation, and Performance (GAPP) program collected survey data on 82\% of adults in the 16 villages as well as social network data. \cite{ferrali2020takes} evaluated the adoption patterns of U-Bridge a couple years later. \cite{eubank2019viral} study the role of social network structure on voting patterns. For my analysis, I evaluate the impact of attendance to UBridge meetings on political participation using the network propensity score matching methodology. Spillovers are likely to occur in this context because non-participants can receive information about ways to engage in politics from their friends, which can increase their own political activity. 
 
The data collected by the researchers contains four types of social networks: Family ties, friendships, lenders and problem solvers. In my analysis, $\{ig\}$ is an identifier for an adult in the pilot villages. The indicator $A_{ijg}$ equals one if $\{ig\}$ and $\{jg\}$ have a connection along any of the four dimensions and zero otherwise. Under this definition, individuals have 10 connections on average. The indicator $D_{ig}$ equals one if $\{ig\}$ attended the Ubridge meetings, which is around 8.6\% of the sample. The outcome is a continuous variable $Y_{ig}$ that denotes a political participation index constructed by \cite{ferrali2020takes}. Table \ref{table:descriptivesuganda} presents summary statistics comparing the treatment and control group. The average adult in the sample is around 40 years old. Men are more likely to attend the session than women. Individuals that a leader position and/or completed their secondary education are more likely to attend as well.

I estimate the following linear model with random coefficients.
\begin{equation}
Y_{ig} = \alpha_{ig} + \beta_{ig}D_{ig} + \gamma_{ig}\left( \frac{T_{ig}}{L_{ig}} \right) + \delta_{ig}\left( \frac{T_{ig}}{L_{ig}} \right)
\label{eq:equationuganda}
\end{equation}
Heterogeneity of $\beta_{ig}$ means that agents engage in varying levels of political activity after attending the meeting. In this case, we expect $\beta_{ig}$ to be close to zero because individuals that are already politically engaged are the ones opting to go to the meetings. Conversely, $\gamma_{ig}$ is the effect of peers on non-participant adults. If $\gamma_{ig} > 0$, then individuals with a larger fraction of treated friends are more politically active. The coefficient $\gamma_{ig} + \delta_{ig}$ captures the spillovers for participants. In this case we expect $\delta_{ig} <0 $ because the marginal effect of attending friends is lower because they are already receiving the information first hand.

There is a potential identification in this example because individuals select connections with similar preferences.  We expect $(\gamma_{ig},\delta_{ig})$ to be correlated with $(T_{ig}/L_{ig})$. To address this problem I leverage additional covariates collected by the researcher to tease out the causal effects. The network propensity score matching methodology is the appropriate tool to identify the average partial effect $\tau$ because it allows to incorporate additional covariates while allowing for heterogeneous causal effects $\tau_{ig} = (\alpha_{ig},\beta_{ig},\gamma_{ig},\delta_{ig})$.

\subsection{Feasible Network Propensity Score and Causal Effects}

The propensity score in this case describes the probability of attending an Ubridge meeting given covariates $C_{ig} = (C_{ig1},\ldots,C_{igK})$. These include an indicator for holding a leadership position in the village, gender, an indicator for secondary education, a self-reported relative income measure, distance to the meeting place, number of friends and age. \cite{ferrali2020takes} also incorporated a public goods question where participants were asked to donate part of their remuneration to the village that were match researchers. The donation amount is meant to capture pro-sociability attitudes.

I assume that the group-level variation $\Psi_g^*$ has an observed and an unobserved component. For the observed component, I include a vector of group-level averages of the key variables in $C_{ig}$, which I denote by $\Psi_g$. I assume that $\Psi_g^*$ has a bivariate structure with mean $(\Psi_g'\boldsymbol{\theta}_{d\Psi},\Psi_g'\boldsymbol{\theta}_{d\Psi})'$, where $(\boldsymbol{\theta}_{d\Psi},\boldsymbol{\theta}_{f\Psi})$ is a vector of parameters to be estimated. The error term of $\Psi_g^*$ follows a normally distributed random-effects structure with covariance matrix $\Sigma \equiv (\sigma_{11}^2,\sigma_{12},\sigma_{12},\sigma_{22}^2)$, that is assumed to be independent of the observed covariates and the random coefficients $\tau_{ig}$. The coefficient $\sigma_{12}$ captures the correlation between the two unobserved components of $\Psi_{g}^*$. Formally,
\[
\Psi_g^* = \left( \begin{matrix} \Psi_{gd}^* \\  \Psi_{gf}^* \end{matrix} \right) \sim \mathcal{N}(\mu_g,\Sigma), \qquad \mu_g = \begin{pmatrix}  \Psi_{g}'\theta_{d\Psi} \\  \Psi_{g}'\theta_{f\Psi} \\  \end{pmatrix}, \qquad \Sigma = \left( \begin{matrix} \sigma_1^2 & \sigma_{12} \\ \sigma_{12} & \sigma_{2}^2 \end{matrix}\right).
\]
I assume that the own propensity score takes the form of a logit function with an associated vector of parameters $\boldsymbol{\theta}_{d} = (\theta_{d0},\theta_{d1},\ldots,\theta_{dK})$ as follows
\[
p_d(C_{ig},\Psi_g^*;\boldsymbol{\theta}_d) = \frac{\exp(\theta_{d0} + \sum_{k=1}^{K} C_{igk}\theta_{dk} + \Psi_{gd}^* ) }{1+\exp(\theta_{d0} + \sum_{k=1}^{K} C_{igk}\theta_{dk} + \Psi_{gd}^* )}
\]
I similarly construct the friend propensity score using a logit link function. I use the same observables variables as the friend friend propensity score with different coefficients $\boldsymbol{\theta}_{f} = (\theta_{f0},\theta_{f1},\ldots,\theta_{df})$ as follows 
\[
p_f(C_{ig},\Psi_g^*; \boldsymbol{\theta}_f) = \frac{\exp(\theta_{f0} + \sum_{k=1}^{K} C_{igk}\theta_{fk}+ \Psi_{gf}^*   )}{1+\exp(\theta_{f0} + \sum_{k=1}^{K} C_{igk}\theta_{fk} + \Psi_{gf}^* )}
\]

The full vector of parameters to be estimated is
\[
\boldsymbol{\theta} \equiv (\boldsymbol{\theta}_{d\Psi},\boldsymbol{\theta}_{f\Psi},\sigma_{1}^2,\sigma_{12},\sigma_2^2,\boldsymbol{\theta}_{d},\boldsymbol{\theta}_{f}).
\]
Let $F(\Psi_g^* ; \boldsymbol{\theta})$ is the distribution of unobserved heterogeneity, which corresponds to that of a normal distribution with parameters $(\mu_g,\Sigma)$. I construct a weighting matrix that satisfies the mixture model representation of Lemma \ref{lem:mixturepresentationQxx}, where $V_{ig} = (C_{ig},\Psi_g,L_{ig})$. To simplify notation I define the auxiliary matrix
\[
\Lambda(C_{ig},\Psi_g^*,L_{ig};\boldsymbol{\theta}) = \begin{pmatrix} 1 & p_f(C_{ig},\Psi_g^*;\boldsymbol{\theta}_f)\\ p_f(C_{ig},\Psi_g^*;\boldsymbol{\theta}) & \frac{p_f(C_{ig},\Psi_g^*;\boldsymbol{\theta}_f)(1-p_f(C_{ig},\Psi_g^*;\boldsymbol{\theta}_f))}{L_{ig}} + p_f(C_{ig},\Psi_g^*;\boldsymbol{\theta}_f)^2  \end{pmatrix}
\]
The feasible weighting matrix is equal to
\begin{align}
	\mathbf{Q}_{xx}(V_{ig};\boldsymbol{\theta}) = \int \Lambda(C_{ig},\Psi_g^*,L_{ig};\boldsymbol{\theta}) \otimes \begin{pmatrix} 1 & p_d(C_{ig},\Psi_g^*;\boldsymbol{\theta}) \\ p_d(C_{ig},\Psi_g^*;\boldsymbol{\theta}) & p_d(C_{ig},\Psi_g^*;\boldsymbol{\theta}) \end{pmatrix} dF(\Psi_g^* ; \boldsymbol{\theta}).		
\end{align}
where I evaluate the integral numerically via quadrature methods and estimate the parameter $\boldsymbol{\theta}$ by minimizing the sample criterion function in \eqref{eq:samplecriterion}.

\begin{table}[t]
	\centering
	\begin{tabular}{lcccc}
		\hline
		& \multicolumn{2}{c}{\textit{Network Propensity}} & \multicolumn{2}{c}{\textit{OLS with   covariates}}          \\
		& \textit{Coefficients}                & \textit{Std. Error}              & \textit{Coefficients}                     & \textit{Std. Error} \\
		\hline
		Direct Effect $(\beta)$  & 0.270            & (0.165)           & \multicolumn{1}{c}{-0.010} & (0.060)                  \\
		Spillover Effect $(\gamma)$ &       0.348***          & (0.116)             & \multicolumn{1}{c}{0.156**} & (0.068)   \\
		Interaction $(\delta)$ & -0.199               & (0.862)             & \multicolumn{1}{c}{0.563***} & (0.165)  \\
		\hline
		N & 2831 & & 2831 & \\
		Villages & 16 & & 16 \\
		\hline
	\end{tabular}
	\caption[Average partial effects political participation in Uganda]{\footnotesize \textbf{(Average Partial Effects Political Participation in Uganda)} * Significant at 10\%. ** Significant at 5\%. *** Significant at 1\%. The second and third columns show the coefficients and standard errors of the inverse-weighted estimator, respectively. The fourth and fifth columns are the coefficients of am additive ordinary least squares (OLS) regression that regresses $Y_{ig}$ on a constant, $D_{ig}$, $(T_{ig}/L_{ig}), D_{ig} \times (T_{ig}/L_{ig})$ and the observed controls used in the inverse-weighting procedure.}
	\label{table:causaleffectsunweighteduganda}
\end{table}

Table \ref{table:ugandanetworkpropensityscore} reports the estimated parameters. Column (2) shows the coefficients of the propensity score. None of the variables in $C_{ig}$ appears to be statistically significant. Column (3) reports the coefficients of the friend propensity score, which are far more interesting. The evidence suggests that individuals that hold a leadership position and have completed a higher education or more likely to have a treated friend. Similarly individuals in villages where individuals perceive themselves as wealthier are more likely to see engagement with the U-Bridge sessions. Figure \ref{fig:estimatednetpscoreuganda} plots the propensity score and friend propensity score, integrating out the heterogeneity $\Psi_g^*$. Each score contains complementary information about the selection patterns. Finally to test the fit of the model I run a covariate test / placebo test by replacing the outcome variable in \eqref{eq:equationuganda} with each of the controls used in the analysis. None of the placebo coefficients are statistically significant for 16 out of the 20 variables. There are slight imbalances on one the relative income indicators, the distance to meeting and the average sociability.

\begin{figure}[t]
	\centering
	\begin{subfigure}{.45\textwidth}
		\centering
		\includegraphics[scale=0.09]{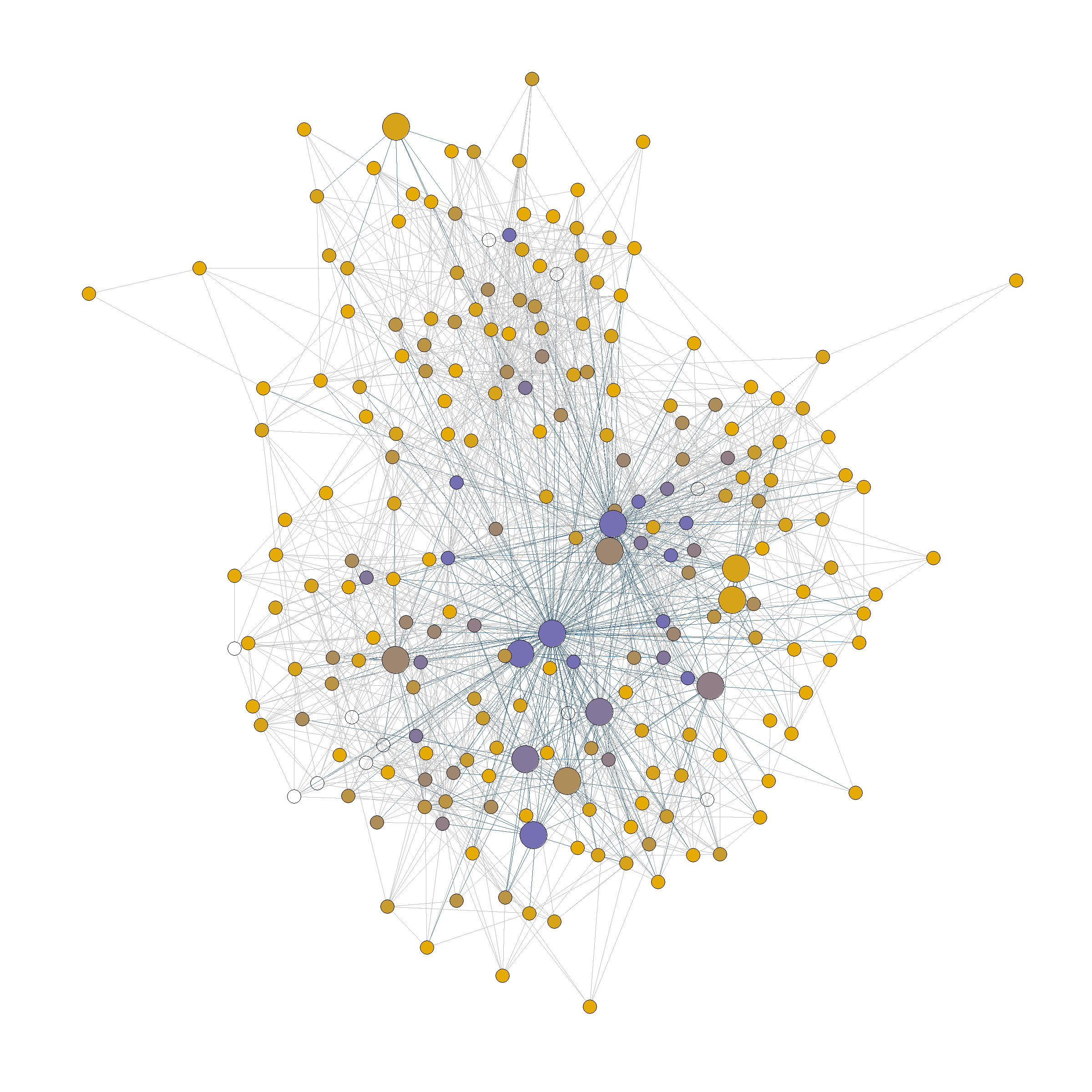}
		\caption{Histogram of the number of friends}
		\label{fig:histogramnumfriends1}
	\end{subfigure}%
	\begin{subfigure}{.45\textwidth}
		\centering
		\includegraphics[scale=0.09]{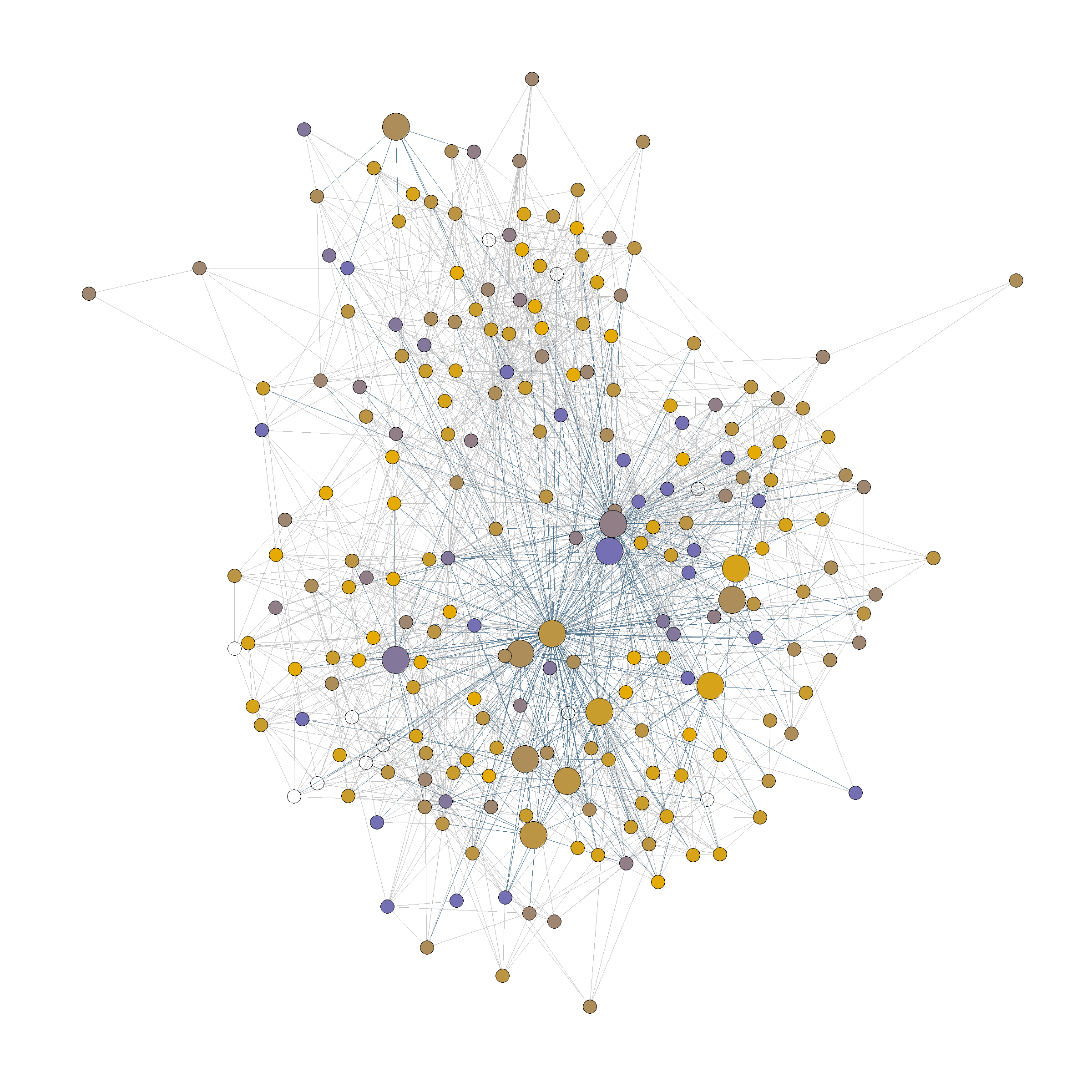}	
		\caption{Histogram of same-caste friends}
		\label{fig:histogramcastefriends2}
	\end{subfigure}
	\caption[Example estimated network propensity score]{The figure shows the estimated $p_{dig}$ and $p_{fig}$ for the network graph of one Ugandan village. Individuals are represented as nodes, and the links between them represent the relationships reported in the baseline survey. Treated individuals are represented with larger nodes. In figure (a) a darker shade of blue indicates a higher estimated probability of treatment, whereas a darker shade of yellow indicates a low probability. Analogously, in figure (b) a darker shade of blue indicates a larger probability of friend treatment.}
	\label{fig:estimatednetpscoreuganda}
\end{figure}

Table \ref{table:causaleffectsunweighteduganda} reports estimates of the average partial effects. Column (2) shows the coefficients under the network propensity approach. The direct effect $\beta$ is positive but not statistically significant at the 10\% level. The spillover effect $\delta$ increases the participation index by 0.348 points, which is significant at the 1\% level. This effect is quantitatively large relative to the standard deviation of the political participation index, which is around 0.567 points. This finding appears to suggest that the intervention had a large spillovers on non-participants, who increased their political activity. The interaction coefficient $\delta$ is negative but not statistically significant at the 1\% level. The results are consistent with the idea that the intervention had limited effects direct treatment effects, but promoted spillover effects on participants' social connections. Column (3) shows benchmark coefficients from an OLS regression with additive covariates. On one hand, the OLS coefficient of $\beta$ is also not statistically significant at the 10\% level. On the other hand, the OLS coefficient of $\gamma$ is statistically significant but roughly half the size of the network propensity estimate. Finally, the coefficient of $\delta$ is positive and statistically significant. The discrepancies in the results for $\gamma$ and $\delta$ can be explained by interactive spillover effects $\gamma_{ig}$ and $\delta_{ig}$ that are not captured by the additive OLS model.

\subsection{Microfinance Adoption in India}

In this section I re-evaluate a program that encouraged the adoption of microfinance in rural areas of Southern India, by inviting select households to participate in an information about the program \citep{banerjee2013diffusion}. Participant households were more likely to take out a loan. Spillovers are likely to occur in this context due to information transmission between participants and non-participants, and peer pressure to adopt.

The outcome is a binary variable $Y_{ig}$ that is equal to one if household $\{ig\}$ took out a loan when researcher followed-up a few months later. I estimate the following linear probability model with random coefficients.
\begin{equation}
Y_{ig} = \alpha_{ig} + \beta_{ig}D_{ig} + \gamma_{ig}\left( \frac{T_{ig}}{L_{ig}} \right) + \delta_{ig}\left( \frac{T_{ig}}{L_{ig}} \right)
\label{eq:equationbanerjee}
\end{equation}
Heterogeneity of $\beta_{ig}$ in the microfinance example means that some households are more likely to take-out a loan after the information session than others. Conversely, heterogeneity of $\gamma_{ig}$ and $\delta_{ig}$ means that not every household is equally likely to get in debt after receiving information from their friends. The coefficient $\delta_{ig}$ is the difference in spillovers effects between participant and non-participant households.

Identification of the average partial effect $\tau \equiv (\alpha,\beta,\gamma,\delta)$ is particularly challenging in this setting, however, because the treatment was not randomly assigned.  The microfinance organization followed a fixed targeting strategy in each village, that selected shopkeepers, teachers and related occupations. However, Table \ref{table:table1summary} shows that treated households were wealthier; they were more likely to have stone or concrete houses as opposed to tile or thatch, have private electricity, more bedrooms, and own a latrine. For instance, the treated were 13.45\% more likely to have access to some form of sanitation, with either a private or public latrine. These differences are statistically significant at the 5\% level, using clustered standard errors by village. There were also significant differences by caste, a hereditary social category that still defines many social boundaries, with household of so-called ``general caste'' more likely to be treated as opposed to minorities.

To measure social network links, \cite{banerjee2013diffusion} collected  twelve different definitions of the network at baseline, including favor exchange, commensality and community activities. I choose a conservative definition of the network, such that $A_{ijg}$ is equal to one if respondents reported a link along any of the dimensions. Figure \ref{fig:histogramnumfriends} plots the resulting degree distribution, which shows that the treated had a higher number of friends. Households have around ten friends on average, which is around 5\% of the average village size. Figure \ref{fig:histogramcastefriends} shows that households reported that most of their friends were in the same broad caste category. A significant portion of the households reported that \textit{all} of their friends were in the same category. The histogram shows that the treated had more diversified friendships, in the sense that they had fewer friends of the same caste.

\begin{figure}[t]
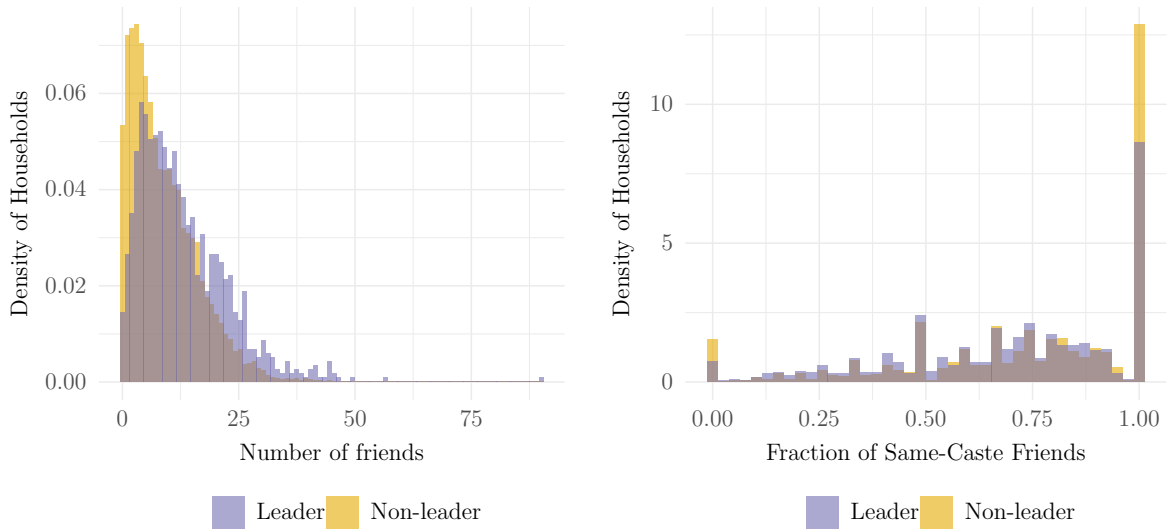

	\centering
	\begin{subfigure}{.5\textwidth}
		\centering
		\resizebox{\textwidth}{!}{
			\input{tables/outputfigures/histrogramfriends.tex}			
		}
		\caption{Histogram of the number of friends}
		\label{fig:histogramnumfriends}
	\end{subfigure}%
	\begin{subfigure}{.5\textwidth}
		\centering
		\resizebox{\textwidth}{!}{
		\input{tables/outputfigures/histogramfriendscaste.tex}
		}		
		\caption{Histogram of same-caste friends}
		\label{fig:histogramcastefriends}
	\end{subfigure}
	\caption[Histograms of overall friends and same-caste friends in India]{Figure (a) shows a histogram with the number of friends of each households, broken down by the treated and control households. Leaders tend to have a higher number of friends. Figure (b) shows a histogram with the fraction of same caste-friends. The general survey which contains information on five broad categories ``General'', ``Minority'', ``OBC'', ``Scheduled Caste'' and ``Scheduled Tribe''. I computed the fraction of treated friends for each household in the same caste category.}
	\label{fig:histogramnetworkmeasures}
\end{figure}

To estimate the network propensity score I use the same specification as in the example for Uganda. The second and third columns of Table \ref{table:tablenetworkpropensityunweighted} show the coefficients of the own propensity score and the corresponding standard errors. The structural parameters confirm the descriptive evidence. The number of rooms in the house, as well as the access to sanitation and electricity are statistically significant at the 5\% level. Individuals of general caste and more connections, are more likely to be part of the program, even after accounting for asset measures. The observed group covariates are not statistically significant at the 10\% level. Conversely, the fourth and fifth columns show estimated coefficients of the friend propensity score and their standard errors. Only the sociability index and the general caste indicator are statistically significant. This suggests that caste plays a crucial role on the interplay between homophily and selection. Treated individuals of general caste are more likely to befriend other treated individuals in their same caste category. The results also show that the unobserved heterogeneity parameters are not statistically significant at the 10\% level.

\begin{table}[t]
	\centering
	\begin{tabular}{lcccc}
		\hline
		& \multicolumn{2}{c}{\textit{Network Propensity}} & \multicolumn{2}{c}{\textit{OLS with   covariates}}          \\
		& \textit{Coefficients}                & \textit{Std. Error}              & \textit{Coefficients}                     & \textit{Std. Error} \\
		\hline
		Direct Effect $(\beta)$  & 0.096**            & (0.046)           & \multicolumn{1}{c}{0.077***} & (0.029)                  \\
		Spillover Effect $(\gamma)$ &       0.092          & (0.091)             & \multicolumn{1}{c}{0.026} & (0.036)   \\
		Interaction $(\delta)$ & -0.102               & (0.292)             & \multicolumn{1}{c}{0.000} & (0.121)  \\
		\hline
		Village Controls & Yes & & Yes \\
		N & 7480 & & 7480 & \\
		Villages & 43 & & 43 \\
		\hline
	\end{tabular}
	\caption[Average partial effects microfiance in India]{\footnotesize \textbf{(Average Partial Effects Microfinance in India)} * Significant at 10\%. ** Significant at 5\%. *** Significant at 1\%. The table shows the coefficients of the causal effects. The second and third columns show the coefficients and standard errors of the inverse-weighted estimator, respectively. The fourth and fifth columns are the coefficients of a ordinary least squares (OLS) regression that regresses $Y_{ig}$ on a constant, $D_{ig}$, $(T_{ig}/L_{ig}), D_{ig} \times (T_{ig}/L_{ig})$ and the observed controls used in the inverse-weighting procedure. This sample merges the census-level data with a detailed survey for a random subsample of households, to fill in missing caste data. The sample excludes households without friends, households with more than 30 friends, and those that have missing caste or electricity data, which is 0.77\% of the overall sample. The standard errors are clustered at the village level.}
	\label{table:causaleffectsunweighted}
\end{table}

Table \ref{table:causaleffectsunweighted} computes the treatment effects using my proposed inverse-weighting (IW) procedure and an ordinary least squares (OLS) regression that includes the covariates as additive controls. The IW results show that participants in the information session (leaders) are 8.5\% more likely to take-out a microfinance loan after controlling baseline characteristics, and is significant at the 1\% level. The value of the direct effect is 1\% higher than the effect estimated by OLS. The OLS regression only controls for additive heterogeneity, but it does not account for the possibility of heterogeneous slopes/treatment effects. The spillover effect is not significant in either case. That means that local variation in treated friends does not affect the outcome, on average.



\section{Conclusion}
\label{conclusion}

Many programs offered by governments and non-profit organizations are not randomly assigned. Individuals typically select into treatment based on a set of eligibility criteria. Furthermore, social networks typically exhibit homophily: individuals tend to befriend others with similar characteristics. The interaction between selection and friendship \textit{homophily} is not well understood, and hence the strategies to identify spillovers from observed social networks in this context are underdeveloped. This is particularly important for policy evaluators because estimating spillovers is crucial for cost-benefit calculations and understanding potential side-effects on non-participants.

This paper proposes a novel strategy for identifying average treatment effects and average spillover effects in settings with endogenous network formation and selection on observables. In particular I show that controlling for the key determinants of friendship decisions in a dyadic network model can account for possible confounders in the estimation of spillovers. I introduce a lower dimensional statistic, the network propensity score, which summarizes the key confounders and illustrates the crucial interplay between homophily and selection. I propose a two-step semiparametric estimator of the average effects in a class of random coefficient models, which is consistent as the number and size of the network grows.   I apply my estimator to an intervention to encourage political participation in Uganda where I find evidence of spillovers on non-participants, and a microfinance application in India, where I document large direct effects but no meaningful local spillovers.

\newpage

\appendix
\numberwithin{equation}{section}
\numberwithin{table}{section}
\numberwithin{figure}{section}

\newpage 
\onehalfspacing
\bibliographystyle{elsarticle-harv}
\bibliography{bibliographynetworks}
\newpage

\newpage

\begin{table}[ht]
	\small
	\centering
	\begin{tabular}{lcclc}
		\hline 
		& Treated & Control & Difference & Std. Error \\ 
		\hline
		Political Participation  & 0.370 & -0.035 & 0.406*** & (0.043) \\
		Leader & 0.288 & 0.132 & 0.156*** & (0.027) \\
		Prosociabililty Index & 0.198 & 0.196 & 0.001 & (0.011) \\
		Female & 0.275 & 0.606 & -0.331*** & (0.044)  \\
		Secondary Education & 0.458 & 0.207 & 0.251*** & (0.046) \\
		Relative income: Low & 0.296 & 0.278 & 0.018 & (0.043)\\
		Relative income: Avg & 0.108 & 0.103 & 0.005 & (0.015) \\
		Relative income: High & 0.375 & 0.324 & 0.051 & (0.034) \\
		Relative income: Very High & 0.025 & 0.022 & 0.002 & (0.010) \\
		Distance to meeting & 1.702 & 1.788 & - 0.086 & (0.163) \\ 
		Number of Friends & 11.775 & 9.413 & 2.362*** & (0.326) \\
		Age & 40.504  & 37.090 & 3.415*** & (0.101) \\
		\hline
		N & 250 & 2591 \\
		Villages & 16 & 16 \\
		\hline
	\end{tabular}	
	\caption[Summary statistics Uganda]{\footnotesize \textbf{(Summary statistics political participation in Uganda)} Differences between leader households selected by the microfinance organization and non-leader households. All the variables are measured at baseline. This sample merges the census-level data with a detailed survey for a random subsample of households, to fill in missing caste data. The sample excludes households without friends, households with more than 30 friends, and those that have missing caste or electricity data, which is 0.77\% of the overall sample. The standard errors are clustered by village.}
	\label{table:descriptivesuganda}
\end{table}

\begin{table}[ht]
	\small
	\centering
	\begin{tabular}{lcccc}
		\hline
		& \multicolumn{2}{c}{\textit{Own Propensity Score}} & \multicolumn{2}{c}{\textit{Friend Propensity Score}} \\
		\textbf{}                & \textit{Coefficient}                             & \textit{Std. Error.}                   & \textit{Coefficient}                                  & \textit{Std. Error}          \\
		\hline
Leader                              & 0.662  & (1.154)  & 0.116*** & (0.037) \\
Sociability Index                   & -1.122 & (1.428)  & -0.157   & (0.128) \\
Female                              & -1.383 & (1.712)  & -0.185   & (0.055) \\
Has secondary education             & 0.876  & (1.503)  & 0.179*** & (0.050) \\
Relative income: Somewhat worse         & 0.411  & (0.431)  & 0.013    & (0.047) \\
Relative income: About the same     & 0.15   & (0.367)  & 0.079*   & (0.046) \\
Relative income: Somewhat better    & 0.296  & (0.325)  & 0.022    & (0.057) \\
Relative income: Much better        & 0.137  & (0.758)  & -0.111   & (0.125) \\
Distance to meeting                 & -0.191 & (0.430)  & -0.084   & (0.042) \\
Number of friends                   & 0.187  & (0.263)  & -0.03    & (0.010) \\
Age                                 & 0.17   & (0.210)  & 0.027    & (0.020) \\
Share of leaders in village         & 0.167  & (10.215) & -2.406   & (2.433) \\
Average sociability index           & -8.028 & (11.269) & -8.034   & (3.334) \\
Share of women in village           & -2.332 & (12.777) & -9.321   & (4.281) \\
Share of high-school educated       & 1.428  & (4.434)  & 0.322    & (1.748) \\
Share reporting "Somewhat   worse"  & 0.913  & (15.164) & -2.484   & (4.092) \\
Share reporting "About the same"  & 16.559 & (17.025) & 9.536*** & (2.345) \\
Share reporting "Somewhat   better" & 4.024  & (17.623) & 1.501    & (4.555) \\
Average distance to meeting       & 0.233  & (0.701)  & -0.07    & (0.122) \\
Average age         & -1.192 & (2.598)  & -1.167   & (0.544) \\
$\log(\sigma_1)$                   & 1.697  & (3.071)  &          &         \\
$\sigma_{12}$                       &        &          & 0.178    & (0.190) \\
$\log(\sigma_2)$                    &        &          & -3.324   & (0.323) \\
Constant                            & -2.168 & (16.170) & 9.934*** & (3.328) \\
		\hline  
		Number of Observations & 2,831 &  & 2,831 \\
		Number of Villages & 16 &  & 16 \\
		\hline
	\end{tabular}
	\caption[Network propensity score parameters Uganda]{\footnotesize \textbf{(Network Propensity Score Parameters Uganda)} * Significant at 10\%. ** Significant at 5\%. *** Significant at 1\%. Columns (2) and (4) show the estimated coefficients for propensity score and friend propensity scores, respectively. Columns (3) and (5) show the corresponding standard errors, that are clustered by village. The relative income asks how an individual's perceives her household income relative the typical household. The baseline category is "Much worse than the typical household". I dropped the ``Share reporting: Much Better" variable because there was very little variation (only 2\% of the sample marked this category). The bottom half of the table reports village-level averages and shares of the key variables. I omit the share for the "Much better" category because there are two few individuals. The bottom rows displays the parameters of the covariance matrix of the unobserved heterogeneity parameters. The sample for the table excludes households without friends and missing data on distance to meeting, gender, age and income.}
	\label{table:ugandanetworkpropensityscore}	
\end{table}

\newpage
\begin{table}[ht]
	\small
	\centering
	\begin{tabular}{lcccccccc}
		\hline
		& \multicolumn{2}{c}{$\widetilde{\beta}$} & \multicolumn{2}{c}{$\widetilde{\gamma}$} & \multicolumn{2}{c}{$\widetilde{\delta}$} \\
		\textbf{}                & \textit{Coeff.}                             & \textit{Std.}                   & \textit{Coeff.}                                  & \textit{Std.}  & \textit{Coeff.}                                  & \textit{Std.}       \\
		& & \textit{Error} & & \textit{Error} & & \textit{Error} \\
		\hline
Leader                              & 0.039  & (0.217)  & 0.097   & (0.081) & 0.097   & (3.970)  \\
Pro-Sociability Index               & 0.627  & (1.961)  & -0.091  & (0.107) & -1.637  & (0.359)  \\
Female                              & 0.064  & (0.111)  & 0.083   & (0.150) & -0.242  & (2.802)  \\
Has secondary education             & 0.399  & (1.454)  & 0.079   & (0.102) & -0.946  & (3.161)  \\
Income: Somewhat worse     & 0.456  & (1.463)  & 0.123   & (0.097) & -1.323  & (14.137) \\
Income: About the same     & 1.509  & (6.494)  & 0.571   & (0.638) & -3.92   & (58.105) \\
Income: Somewhat better    & 7.395  & (26.897) & 4.249** & (1.917) & -20.639 & (18.984) \\
Income: Much better        & 2.877  & (9.529)  & 0.796   & (0.647) & -7.198  & (0.957)  \\
Distance to meeting                 & 0.129  & (0.462)  & 0.048** & (0.023) & -0.327  & (3.375)  \\
Number of friends                   & 0.453  & (1.687)  & 0.081   & (0.068) & -1.108  & (1.486)  \\
Age                                 & 0.235  & (0.749)  & 0.036   & (0.027) & -0.567  & (3.031)  \\
Share of leaders in village         & 0.401  & (1.507)  & 0.087   & (0.074) & -1.013  & (2.164)  \\
Average sociability index           & Index  & (1.080)  & 0.067** & (0.033) & -0.736  & (12.502) \\
Share of women in village           & 1.504  & (5.791)  & 0.562   & (0.745) & -3.377  & (22.216) \\
Share of high-school educated       & 3.032  & (11.044) & 0.597   & (0.422) & -7.452  & (53.987) \\
Share of "Somewhat   worse"  & 7.243  & (26.817) & 1.966*  & (1.157) & -18.381 & (0.343)  \\
Share of "About the   same"  & 0.007  & (0.080)  & 0.021   & (0.027) & -0.012  & (0.977)  \\
Share of "Somewhat   better" & 0.031  & (0.245)  & 0.063   & (0.072) & -0.08   & (7.409)  \\
Average distance to meeting         & -0.098 & (1.251)  & 0.708   & (0.794) & 0.776   & (11.434) \\
Average age                         & 0.23   & (2.839)  & 0.668   & (0.907) & -0.566  & 0.000   \\
		\hline
	\end{tabular}
	\caption[Covariate balancing participation Uganda]{\footnotesize \textbf{(Covariate Balancing Participation Uganda)} * Significant at 10\%. ** Significant at 5\%. *** Significant at 1\%. This table shows the coefficients of inverse-weighted estimators, where each of the baseline characteristics is treated as a (placebo) outcome variable. If the weighting matrix is correctly specified $\widetilde{\beta} =\widetilde{\gamma}= \widetilde{\delta} = 0$. The relative income asks how an individual's perceives her household income relative the typical household. The baseline category is "Much worse than the typical household". I dropped the ``Share reporting: Much Better" variable because there was very little variation (only 2\% of the sample marked this category). The bottom half of the table reports village-level averages and shares of the key variables. I omit the share for the "Much better" category because there are two few individuals.  The sample for the table excludes households without friends and missing data on distance to meeting, gender, age and income.}
	\label{table:ugandabalancing}	
\end{table}

\newpage
\begin{table}[t]
	\centering
	\small
	 \begin{tabular}{lcccc}  \hline & Non Leaders & Leaders & Difference & Std. Error \\ \  &  (N = 6,551) &  (N = 929) &  (N = 7,480) &  (N = 7,480)\\ 
 	\hline
 	\ \bf{Roof Type} & ~ & ~ & ~ & ~\\ \ ~~ Thatch & 2 \% & 1 \% & -1.12 \% & 0.43 \%\\ \ ~~ Tile & 38 \% & 31 \% & -6.32 \% & 2.42 \%\\ \ ~~ Stone & 26 \% & 30 \% & 4.2 \% & 2.19 \%\\ \ ~~ Sheet & 21 \% & 20 \% & -0.6 \% & 1.52 \%\\ \ ~~ RCC (Reinforced Concrete) & 10 \% & 15 \% & 4.69 \% & 1.2 \%\\ \ ~~ Other & 4 \% & 3 \% & -0.85 \% & 0.78 \%\\ \ \bf{No. Rooms} & ~ & ~ & ~ & ~\\ \ ~~ Mean & 0.77 & 1.06 & 0.29 & 0.06\\ \ ~~ Sd & 1.1 & 1.39 &  & \\ \ \bf{Electricity} & ~ & ~ & ~ & ~\\ \ ~~ Yes, Private & 61 \% & 72 \% & 10.94 \% & 1.98 \%\\ \ ~~ Yes, Government & 32 \% & 24 \% & -8.19 \% & 1.9 \%\\ \ ~~ No & 7 \% & 4 \% & -2.75 \% & 0.68 \%\\ \ \bf{Latrine} & ~ & ~ & ~ & ~\\ \ ~~ Owned & 25 \% & 39 \% & 13.5 \% & 1.7 \%\\ \ ~~ Common & 1 \% & 1 \% & -0.06 \% & 0.25 \%\\ \ ~~ None & 74 \% & 61 \% & -13.45 \% & 1.78 \%\\ \ \bf{Residence} & ~ & ~ & ~ & ~\\ \ ~~ Owned & 90 \% & 93 \% & 2.66 \% & 1.05 \%\\ \ ~~ Owned but shared & 1 \% & 1 \% & 0.34 \% & 0.35 \%\\ \ ~~ Rented & 6 \% & 3 \% & -2.65 \% & 0.76 \%\\ \ ~~ Leased & 0 \% & 0 \% & 0.08 \% & 0.16 \%\\ \ ~~ Government & 4 \% & 3 \% & -0.42 \% & 0.65 \%\\ \ \bf{Caste} & ~ & ~ & ~ & ~\\ \ ~~ General & 11 \% & 20 \% & 8.31 \% & 1.64 \%\\ \ ~~ Minority & 3 \% & 3 \% & -0.68 \% & 0.69 \%\\ \ ~~ OBC & 51 \% & 51 \% & 0.21 \% & 1.65 \%\\ \ ~~ Scheduled Caste & 29 \% & 22 \% & -6.69 \% & 1.57 \%\\ \ ~~ Scheduled Tribe & 5 \% & 4 \% & -1.14 \% & 0.79 \%\\ \ \bf{Religion} & ~ & ~ & ~ & ~\\ \ ~~ Hinduism & 95 \% & 95 \% & 0.09 \% & 0.87 \%\\ \ ~~ Islam & 5 \% & 5 \% & -0.1 \% & 0.91 \%\\ \ ~~ Christianity & 0.09 \% & 0.11 \% & 0.02 \% & 0.12 \%\\ \ \bf{Number of Connections} & ~ & ~ & ~ & ~\\ \ ~~ Mean & 9.91 & 12.5 & 2.59 & 0.25\\ \ ~~ Standard Deviation & 6.64 & 7.31 &  & \\ \hline \ \end{tabular}
 
	\caption[Summary statistics India]{\footnotesize \textbf{(Summary statistics microfinance in India)} Differences between leader households selected by the microfinance organization and non-leader households. All the variables are measured at baseline. This sample merges the census-level data with a detailed survey for a random subsample of households, to fill in missing caste data. The sample excludes households without friends, households with more than 30 friends, and those that have missing caste or electricity data, which is 0.77\% of the overall sample. The standard errors are clustered by village.}
	\label{table:table1summary}	
\end{table}

\newpage

\begin{table}[ht]
	\small
	\centering
	\begin{tabular}{lcccc}
		\hline
		& \multicolumn{2}{c}{\textit{Own Propensity Score}} & \multicolumn{2}{c}{\textit{Friend Propensity Score}} \\
		\textbf{}                & \textit{Coefficient}                             & \textit{Std. Error.}                   & \textit{Coefficient}                                  & \textit{Std. Error}          \\
		\hline
Tile Roof                                & -0.074   & (0.128) & -0.08    & (0.060)  \\
Stone Roof                               & 0.09     & (0.123) & 0.053    & (0.061)  \\
Sheet Roof                               & 0.013    & (0.134) & -0.059   & (0.061)  \\
No. Rooms                                & 0.124*** & (0.037) & 0.007    & (0.013)  \\
Access to Electricity                    & 0.226*   & (0.124) & 0        & (0.041)  \\
Access to Latrine                        & 0.321**  & (0.143) & 0.106**  & (0.052)  \\
General Caste (base OBC)                 & 0.602*** & (0.194) & 0.266*** & (0.094)  \\
Scheduled Caste (base OBC)               & -0.087   & (0.106) & -0.139   & (0.076)  \\
Scheduled Tribe (base OBC)               & -0.099   & (0.234) & 0.046    & (0.097)  \\
Share of general caste in village        & -0.19    & (0.856) & 0.253    & (0.402)  \\
Share of scheduled caste in village      & 0.032    & (0.354) & -0.208   & (0.241)  \\
Share of scheduled tribe in village      & 0.233    & (2.040) & 0.609    & (1.255)  \\
Share of latrine access in village       & 0.597    & (0.852) & 0.384    & (0.567)  \\
Share of electricity access in   village & -1.107   & (0.793) & -0.613   & (0.479)  \\
Total Friends / Village Size             & 9.371*** & (2.155) & 2.233*** & (0.858)  \\
$\log(\sigma_1) $                        & -0.447   & (2.087) &        &    \\
$\sigma_{12}$                            &          &         & 0.255    & (0.339)  \\
$\log(\sigma_2)$                         &          &         & -2.071   & (15.712) \\
Constant                                 & -2.663   & (0.677) & -1.565   & (0.391) \\
		\hline  
		Number of Observations & 7,480 &  & 7,480 \\
		Number of Villages & 43 &  & 43 \\
		\hline
	\end{tabular}
	\caption[Network propensity score parameters microfinance India]{\footnotesize \textbf{(Network Propensity Score Microfinance India)} * Significant at 10\%. ** Significant at 5\%. *** Significant at 1\%. Columns (2) and (4) show the estimated coefficients for the propensity score and friend propensity scores, respectively. Columns (3) and (5) show the corresponding standard errors, that are clustered by village. All the variables are measured at baseline. The bottom rows displays the parameters of the covariance matrix of the unobserved heterogeneity parameters. This sample merges the census-level data with a detailed survey for a random subsample of households, to fill in missing caste data. The sample excludes households without friends, households with more than 30 friends, and those that have missing caste or electricity data, which is 0.77\% of the overall sample.}
	\label{table:tablenetworkpropensityunweighted}	
\end{table}

\newpage
\begin{table}[ht]
	\small
	\centering
	\begin{tabular}{lcccccccc}
		\hline
		& \multicolumn{2}{c}{$\widetilde{\beta}$} & \multicolumn{2}{c}{$\widetilde{\gamma}$} & \multicolumn{2}{c}{$\widetilde{\delta}$} \\
		\textbf{}                & \textit{Coeff.}                             & \textit{Std.}                   & \textit{Coeff.}                                  & \textit{Std.}  & \textit{Coeff.}                                  & \textit{Std.}       \\
& & \textit{Error} & & \textit{Error} & & \textit{Error} \\
\hline
Tile Roof                                & 0.029  & (0.080) & 0.092  & (0.098) & 0.092  & (0.504) \\
Stone Roof                               & 0.025  & (0.095) & 0.05   & (0.085) & -0.136 & (0.381) \\
Sheet Roof                               & 0.033  & (0.067) & 0.085  & (0.122) & -0.147 & (4.082) \\
No. Rooms                                & 0.214  & (0.708) & 0.551  & (0.466) & -1.436 & (1.062) \\
Access to Electricity                    & 0.07   & (0.190) & 0.186  & (0.168) & -0.433 & (0.646) \\
Access to Latrine                        & 0.057  & (0.127) & 0.078  & (0.086) & -0.334 & (0.075) \\
General Caste (base OBC)                 & -0.027 & (0.016) & -0.05  & (0.128) & 0.13   & (0.710) \\
Scheduled Caste (base OBC)               & 0.069  & (0.114) & 0.161  & (0.198) & -0.551 & (0.090) \\
Scheduled Tribe (base OBC)               & 0.005  & (0.016) & 0      & (0.025) & 0.009  & (0.148) \\
Share of general caste in village        & 0.002  & (0.021) & -0.014 & (0.050) & 0.035  & (0.421) \\
Share of scheduled caste in village      & 0.013  & (0.069) & 0.091  & (0.099) & -0.144 & (0.102) \\
Share of scheduled tribe in village      & 0.008  & (0.017) & 0.01   & (0.010) & -0.036 & (0.438) \\
Share of latrine access in village       & 0.02   & (0.078) & 0.052  & (0.056) & -0.178 & (1.060) \\
Share of electricity access in   village & 0.051  & (0.178) & 0.144  & (0.133) & -0.391 & (0.118) \\
Total Friends / Village Size             & 0.009  & (0.021) & 0.016  & (0.014) & -0.05  & 0.000  \\
		\hline  
		Number of Observations & 7,480 &  & 7,480 & & 7,480 \\
		Number of Villages & 43 &  & 43 &  & 43 \\
		\hline
	\end{tabular}
	\caption[Covariate balancing microfinance India]{\footnotesize \textbf{(Covariate Balancing Microfinance India)} * Significant at 10\%. ** Significant at 5\%. *** Significant at 1\%. This table shows the coefficients of inverse-weighted estimators, where each of the baseline characteristics is treated as a (placebo) outcome variable. If the weighting matrix is correctly specified $\widetilde{\beta} =\widetilde{\gamma}= \widetilde{\delta} = 0$. This sample merges the census-level data with a detailed survey for a random subsample of households, to fill in missing caste data. The sample excludes households without friends, households with more than 30 friends, and those that have missing caste or electricity data, which is 0.77\% of the overall sample.}
	\label{table:unweightedbalancing}	
\end{table}

\begin{table}[t]
	\centering
	\begin{tabular}{lcc}
		\hline 
		& \multicolumn{2}{c}{Political Participation} \\
		& \multicolumn{1}{c}{\textit{Average Effect}} & \multicolumn{1}{c}{\textit{Average Effect}} \\	
		& \multicolumn{1}{c}{\textit{on Treated}} & \multicolumn{1}{c}{\textit{on Untreated}} \\
		\hline
		Direct Effect $(\beta)$  & 0.269            & 0.283     \\
		&(0.165)            & (0.166)                 \\
		Spillover Effect $(\gamma)$ &       0.410**  &  0.222  \\
		& (0.116)             & (0.170)   \\
		Interaction $(\delta)$ & -0.185  & -0.303 \\
		& (0.862) & (1.138)  \\
		\hline
		N & 2831 & 2831 \\
		Villages & 16 & 16 \\
		\hline
	\end{tabular}
	
	\vspace{0.025in}
	{\small * p $<$ 0.1, ** p $<$ 0.05. *** p $<$ 0.01}
	\caption[APT and APU political participation in Uganda]{\footnotesize \textbf{(Average partial effects by subpopulation Uganda)} The left column shows the estimated coefficients of the average partial effects on the treated $\mathbb{E}[(\beta_{ig},\gamma_{ig},\delta_{ig}) \mid D_{ig} = 1]$. The right column shows the estimated coefficients of $\mathbb{E}[(\beta_{ig},\gamma_{ig},\delta_{ig}) \mid D_{ig} = 0]$.}	
	\label{table:aptapu_uganda}	
\end{table}

\begin{table}[t]
	\centering
	\begin{tabular}{lcc}
		\hline 
		& \multicolumn{2}{c}{Political Participation} \\
		& \multicolumn{1}{c}{\textit{Average Effect}} & \multicolumn{1}{c}{\textit{Average Effect}} \\	
		& \multicolumn{1}{c}{\textit{on Treated}} & \multicolumn{1}{c}{\textit{on Untreated}} \\
		\hline
		Direct Effect $(\beta)$  & 0.089            & 0.102**     \\
		&(0.165)            & (0.052)                 \\
		Spillover Effect $(\gamma)$ &       0.090  &  0.093  \\
		& (0.099)             & (0.090)   \\
		Interaction $(\delta)$ & -0.094  & -0.108 \\
		& (0.304) & (0.300)  \\
		\hline
		N & 7480 & 7480 \\
		Villages & 43 & 43 \\
		\hline
	\end{tabular}
	
	\vspace{0.025in}
	{\small * p $<$ 0.1, ** p $<$ 0.05. *** p $<$ 0.01}
	\caption[APT and APU microfinance diffusion in India]{\footnotesize \textbf{(Average partial effects by subpopulation India)} The left column shows the estimated coefficients of the average partial effects on the treated $\mathbb{E}[(\beta_{ig},\gamma_{ig},\delta_{ig}) \mid D_{ig} = 1]$. The right column shows the estimated coefficients of $\mathbb{E}[(\beta_{ig},\gamma_{ig},\delta_{ig}) \mid D_{ig} = 0]$.}	
	\label{table:aptapu_india}

\end{table}

\newpage
\onehalfspacing
\newpage

\section{Appendix}
\label{secappendixidentification}

\subsection{Non-Separable Models}

In this section I relax the random coefficients assumption in \eqref{eq:linearoutcome} by assuming that $Y_{ig} = m(X_{ig},\tau_{ig})$ where $m$ is an unknown function and $\tau_{ig}$ is a vector of unobserved heterogeneity of arbitrary dimension. The researcher is interested in identifying the average structural function, defined as
\[
M(x) = \int m(x,\varepsilon) dF(\varepsilon)
\]
The function $M(x)$ identified the average effect if everyone was subject to the same exposure.

The proof of Theorem \ref{thm:fullcontrolfunction} does not make any explicit use of the functional form of the outcome. If the assumptions of the theorem hold, then $X_{ig} \ \indep \tau_{ig} \mid (C_{ig},\Psi_g^*)$ and 
\begin{align*}
\mathbb{E}[Y_{ig} \mid X_{ig}=x,C_{ig}=c,\Psi_g^* = \Psi^*] &= \int m(x,\varepsilon) dF(\varepsilon \mid x,c,\Psi^*) \\
&= \int m(x,\varepsilon) dF(\varepsilon \mid c,\Psi^*)
\end{align*}
This first stage is analogous to matching individuals with similar characteristics and similar levels of exposure.  The conditional mean is only identified over the conditional support of $(C_{ig},\Psi_g^*)$ given $X_{ig}$. When the conditional support of $(C_{ig},\Psi_g^*)$ given $X_{ig}$ equals the unconditional support we say that the system has \textit{full support}. This condition is similar to a rank condition. In that case the average structural function can be identified by integrating the conditional mean using standard arguments as in \cite{imbens2009identification}.
\begin{align*}
&\mathbb{E}[\mathbb{E}[Y_{ig} \mid X_{ig} = x,C_{ig},\Psi_g^*]] \\
&= \int \int m(x,\tau_{ig})dF(\tau_{ig} \mid c,\Psi^*) dF(c,\Psi^*) = M(x)
\end{align*}
Consequently, the average structural function is identified. \cite{imbens2009identification} show how to extend this idea to identify quantile effects in addition to average outcomes. We can also use the same set of arguments to prove identification of the average structural function for the network propensity score using the result of Theorem \ref{thm:expostbalancing}.

\subsection{Spurious Peer Effects}
\label{spuriouseffects}

Consider the following example where an ordinary least squares (OLS) regression recovers spurious peer effects. Suppose that $Y_{ig} = \alpha_{ig}$, and that \nameref{assumptxt:randomsampling}, \nameref{assumptxt:confounders} and \nameref{assumptxt:dyadicnetwork} are satisfied. Let $V_{ig} \equiv (C_{ig},\Psi_g^*,L_{ig})$ denote the confounders and $X_{ig} = (1,\bar{D}_{-ig})$, where $\bar{D}_{-ig}$ is the fraction of treated friends, defined as $T_{ig}/L_{ig}$. In this case there are no treatment effects, direct or indirect, but the outcome are correlated with the confounders. The researcher runs the following regression over the subset of individuals with at least one friend, $\mathcal{F} = \mathbbm{1}\{L_{ig} > 0\}$,
\[
Y_{ig} = \beta_0 + \beta_1 \bar{D}_{-ig} + \varepsilon_{ig}, \qquad \mathbb{E}[\varepsilon_{ig}] = 0.
\]
The true value of the intercept is $\beta_0 = 0$ and the slope is $\beta_1 = 0$. The population OLS coefficient is defined as
\[
\beta_1^{OLS} = \frac{Cov(\bar{D}_{-ig},Y_{ig}) \mid \mathcal{F}}{Var(\bar{D}_{-ig}\mid \mathcal{F})}
\]
Plugging in $Y_{ig} = \alpha_{ig}$ and using the law of total covariance,
\begin{equation}
\beta_1^{OLS} = \frac{\mathbb{E}[\overbrace{Cov(\bar{D}_{-ig},\alpha_{ig}\mid V_{ig},\mathcal{F})}^{(a)} \mid \mathcal{F}] + Cov(\overbrace{\mathbb{E}[\bar{D}_{-ig}\mid V_{ig}]}^{(b)},\overbrace{\mathbb{E}[\alpha_{ig} \mid V_{ig},\mathcal{F}]}^{(c)} \mid \mathcal{F}) }{Var(\bar{D}_{-ig} \mid \mathcal{F})}.
\label{eq:biasols}
\end{equation}
Theorem \ref{thm:expostbalancing} ensures that $\bar{D}_{-ig} \ \indep \ \alpha_{ig} \mid V_{ig}$, which means that (a) is equal to zero. The term $(b)$ equals  $p_{fig}$, the friend propensity. To simplify notation define $\alpha(V_{ig}) = \mathbb{E}[\alpha_{ig} \mid V_{ig},\mathcal{F}] = \mathbb{E}[\alpha_{ig} \mid V_{ig}]$. Consequently,
\begin{equation}
\beta_1^{OLS} = \frac{Cov\left(p_f(V_{ig}),\alpha(V_{ig}) \mid \mathcal{F} \right)}{Var(\bar{D}_{-ig} \mid \mathcal{F})}
\label{eq:biasols2}
\end{equation}
The OLS coefficient is biased when $\alpha(V_{ig})$ are correlated with $p_{fig}$. For example, suppose that $V_{ig}$ is a poverty index and that $p_{fig}$ is positively correlated with $V_{ig}$. That means that vulnerable individuals are more likely to have a higher fraction of friends who are targeted by the program. Similarly, suppose that $Y_{ig}$ is a measure of food insecurity and that $\alpha(V_{ig})$ is increasing in $V_{ig}$. Then $\beta_1^{OLS} > 0$ because $V_{ig}$ drives both the homophily/selection patterns and the baseline outcomes. Alternatively, when the network and treatment assignment are exogenous, $p_{fig}$ is a constant and the OLS estimator is unbiased because the covariance in the numerator of \eqref{eq:biasols2} equals zero.

\subsection{Effects by subpopulation}

In many cases social programs deliberately target individuals based on baseline characteristics, and the policy maker may not be interested in the effects for the overall population. The identification problem is that individuals are only observed in a single treatment status, which means that the researcher has to find appropriates comparison individuals in the control group that approximate the behavior of the treated under a different exposure. To this end, let us define average partial effect on the treated (APT) and untreated (APU)
\[
\tau_{APT} \equiv \mathbb{E}[\tau_{ig} \mid D_{ig} = 1, \mathcal{F}]
\]
\[
\tau_{APU} \equiv \mathbb{E}[\tau_{ig} \mid D_{ig} = 0, \mathcal{F}]
\]
Theorem \ref{thm:att_atu_identification} presents identification results for $\tau_{APT}$ and $\tau_{APU}$,
\begin{thm}[Identification Subpopulations]
	Suppose that (i) $(X_{ig},D_{ig}) \ \indep \ \tau_{ig} \mid V_{ig}$, (ii) $Y_{ig} = X_{ig}'\tau_{ig}$, and (iii) $\mathcal{F}$ is $V_{ig}-$measurable and $\mathbf{Q}_{xx}(v) = \mathbb{E}[X_{ig}X_{ig}' \mid V_{ig} = v]$ is invertible almost surely over the support of $V_{ig} \mid \mathcal{F}$, then  
	\begin{align*}
		\boldsymbol{\tau}_{APT} &= \frac{1}{\mathbb{E}[D_{ig}]}\times \mathbb{E}\left[p_d(V_{ig}) \times \mathbf{Q}_{xx}(V_{ig})^{-1}X_{ig}Y_{ig}  \mid \mathcal{F} \right] \\
		\boldsymbol{\tau}_{APU} &= \frac{1}{1-\mathbb{E}[D_{ig}]}\times \mathbb{E}\left[ (1-p_d(V_{ig})) \times \mathbf{Q}_{xx}(V_{ig})^{-1}X_{ig}Y_{ig} \mid \mathcal{F} \right]
	\end{align*}
	\label{thm:att_atu_identification}
\end{thm}
The main intuition is fairly similar to Theorem \ref{lem:randomcoef_identification}, in the sense that the inverse weighting ensures equal comparisons across with different strata of $V_{ig}$ whereas the own propensity $p_d(V_{ig})$ weights each strata by the relative number of treated individuals. Notice that the unconditional average partial effects and the $(APT,APU)$ are mutually constrained by the law of iterated expectations $\boldsymbol{\tau} = \mathbb{E}[D_{ig}]\boldsymbol{\tau}_{APT} + (1-\mathbb{E}[D_{ig}])\boldsymbol{\tau}_{APU}$.

Table \ref{table:aptapu_uganda}	computes the average partial effects by subpopulation for the political participation example in Uganda. The coefficients $(\beta_{APT},\delta_{APT})$ and $(\beta_{APU},\delta_{APU})$ have similar magnitudes, standard errors and significance. There are, however, large differences in the magnitudes and significance levels of the spillovers for the control group (in fact $\gamma_{APT} > \gamma_{APU}$). This suggested that individuals with a higher likelihood of participating in the session are more likely to change the behavior if one of their friends is treated. Analogously, in Table \ref{table:aptapu_india}, I compute $(\beta_{APT}, \gamma_{APT}, \delta_{APT})$ and $(\beta_{APU},\gamma_{APU},\delta_{APU})$. The coefficients are similar in magnitude, with comparable standard errors, which suggests that both groups of individuals are fairly similar. In both tables, I compute the standard errors by replacing the definition of $\psi_{IW}$ using a sample analog of the moment conditions in Theorem \ref{thm:att_atu_identification}.

\subsection{Network Propensity Score and Experiments}
\label{experiments_illustration}

One of the most effective ways to identify spillovers is to use a random saturation design. This a two-stage design rising in popularity in the empirical literature \citep{crepon2013labor,gine2018together,bursztyn2019} and studied in several recent econometrics papers \citep{baird,spilloversnoncompliance}. I establish a tight connection between the network propensity and identification in experiments. I show the applicability of my methods to study non-compliance in sparse networks.

In the first stage each group is randomly a saturation, a real number $S_g \in [0,1]$. In the second stage individuals within each group are randomly assigned to treatment with probability $S_g$. This design  is an extension of Bernoulli designs that treat individuals with a fixed probability, such as $S_g = 0.5$, and cluster design that assign complete groups to treatment or control, where $S_g \in \{0,1\}$. The more interesting case combines corner an interior saturations. For example, \cite{crepon2013labor} chooses $S_g \in \{0,0.25,0.5,0.75,1\}$, which generates more experimental variation. To simplify my analysis I focus on the case where the experimenter uses Bernoulli draws to offer treatment in the second stage.

The experimental setting relaxes the assumptions considerably. To discuss the identification of $\tau$ in this experimental context it is useful to assume that $C_{ig}$ includes both baseline  individual characteristics (observed and unobserved). Similarly, I assume that $\Psi_g^*$ includes group characteristics (observed or unobserved) heterogeneity and the exogenous saturations $S_g$. Under this definition it is easy to see that \nameref{assumptxt:confounders} is automatically satisfied because the treatment is exogenous. It is also easy to satisfy the \nameref{assumptxt:randomsampling} and \nameref{assumptxt:dyadicnetwork} assumptions. We can invoke the \cite{aldous1981representations} and \cite{hoover1979relations} theorems that state that any exchangeable network can be represented as a dyadic network with randomly sampled (and possibly unobserved) $C_{ig}$. The purpose of this exercise is to show that in certain experiments there is a simple set of conditioning statistics suffices to identify the treatment effects, even if there is rich unobserved heterogeneity determining the treatment and network choices. \\

\textbf{Example 1 (Perfect Compliance):} The random assignment of saturations and offers means that the propensity score is equal to the group saturation when there is perfect compliance. That means that individuals participate in the program when they are offered and are part of the control when they are not offered. In that case
\begin{equation}
p_{dig} = \underbrace{\mathbb{E}[D_{ig} \mid C_{ig},\Psi_g^*]}_{\text{Definition}} = \underbrace{\mathbb{E}[D_{ig} \mid C_{ig},\Psi_g^*,S_g]}_{\text{Redundancy}} = \underbrace{\mathbb{E}[D_{ig} \mid S_g]}_{\text{Second Stage}} =  \underbrace{S_g}_{\text{First Stage}}.
\label{eq:ownpropensity_rsdesign}
\end{equation}
Equation \eqref{eq:ownpropensity_rsdesign} breaks down the process to show that the propensity is equal to the group saturation. The first equality defines $p_d$. The second equality uses the fact that $S_g$ is a group characteristic that contains redundant information. The last two equality uses the property of the design, that the treatment probability only depends on a saturation which is independent of other characteristics. 

I perform a similar break down for the friend propensity score.
\[
p_{fig} =\underbrace{\mathbb{E}[D_{jg} \mid A_{ijg} =1,C_{ig},\Psi_g^*]}_{\text{Definition}} = \underbrace{\mathbb{E}[D_{jg} \mid A_{ijg} =1,C_{ig},\Psi_g^*,S_g]}_{\text{Redundancy}} = \underbrace{\mathbb{E}[D_{jg} \mid S_g]}_{\text{Second Stage}} = \underbrace{S_g}_{\text{First Stage}}
\]
Finally, the number of friends $L_{ig}$ is not randomly determined by the experimental design and can still be a source of homophily bias that the researcher needs to account for. In networks where everyone is connected $(L_{ig} = N_g - 1)$ this is equivalent to condition on the size of the group, such as classroom size.

The saturation $S_g$ is independent of the random coefficients $\tau_{ig}$ and the baseline information. Formally $\tau_{ig} \ \indep \ S_g \mid L_{ig}$ and hence we can apply Lemma \ref{lem:quasisaturation} to show that $X_{ig} \ \indep \ \tau_{ig} \mid L_{ig}$. That means that matching individuals with similar numbers of friends suffices to identify the average partial effects $\tau$ using Theorem \ref{lem:randomcoef_identification}. \\

\textbf{Example 2: (One-sided compliance)} In practice researchers randomly extend offers but subjects may not be compelled to accept them. Under one-sided compliance treatment status is defined by $D_{ig} = \widetilde{C}_{ig}Z_{ig}$ where $\widetilde{C}_{ig}$ is a binary indicator for whether $\{ig\}$ is a ``complier'' and $Z_{ig}$ is their offer. Compliers with $\widetilde{C}_{ig} = 1$ may perceive larger returns from the program and always participate if offered, where never-takers $\widetilde{C}_{ig} = 0$ do not consider the program worthwhile. In their empirical example from \citep{crepon2013labor}, $D_{ig}$ is a job placement program. The peer effects are potential displacement effect for non-participants that were disadvantaged in a tight labor market. To fit this example within my framework I assume that $\widetilde{C}_{ig}$ is a component of the individual covariates $C_{ig}$.

Non-compliance introduces additional complications because the treatment is no longer randomly assigned. To analyze this problem it is useful to first compute an infeasible propensity score that conditions on the latent complier indicator. If $\tilde{C}_{ig}$ were known
\[
p_{dig} =  \mathbb{E}[\widetilde{C}_{ig} Z_{ig} \mid A_{ijg} = 1,C_{ig},\Psi_g^*] = \widetilde{C}_{ig}S_g
\]
The propensity score for never-takers is always zero, whereas the propensity score for compliers depends on the saturation. The friend propensity equals
\[
p_{fig} = \mathbb{E}[\widetilde{C}_{jg} Z_{jg} \mid A_{ijg} = 1,C_{ig},\Psi_g^*] = \mathbb{E}[\widetilde{C}_{jg} \mid A_{ijg} =1,C_{ig},\Psi_g^*] \times S_g.
\]
The first equality applies the definition of the friend propensity and substitutes the expression for $D_{jg}$ under one-sided compliance. Theorem \ref{lem:quasisaturation} implies that the key dimensions of endogeneity are captured by the vector $V_{ig} = (\widetilde{C}_{ig},\mathbb{E}[\widetilde{C}_{jg} \mid A_{ijg} = 1,C_{ig},\Psi_g^*],L_{ig})$ since $S_g$ is exogenous. The second component of $V_{ig}$ can be interpreted as the probability that a potential complier is treated. This agrees with related work in \cite{spilloversnoncompliance}, where we show which causal effects are identified and show that $(S_g)$ for the spillover effects because of first-stage heterogeneity. They propose a procedure over subsets of the population to recover effects for compliers $\tilde{C}_{ig}=1$ and nervertakers $\tilde{C}_{ig} = 0$. They show that the share of compliers can be consistently using $T_{ig} / S_{ig}$ to construct a valid IV. The procedure relies on complete networks where $L_{ig} = N_g - 1$ and $N_g \to \infty$ in the asymptotic experiment.

\subsection{Regularity Conditions}

In this section I present conditions that are required to derive the asymptotic distribution of the estimator. In order to do so I assume that there is a sequence of distributions indexed by $t$. I denote the realization of variables of agent $\{ig\}$ at point $t$ in the sequence by including the subscript $\{igt\}$. I assume that one or more of the regressors need to be estimated. Let $V_{igt} = (V_{1igt}^0,V_{2igt})$ be the observed regressor and let $V_{igt}^{0} = (V_{1igt}^0,V_{2igt}^0)$. The first vector of regressors is observed without error, but the second estimator is estimated at rate $\max_{g=1,\ldots,G_t}\max_{i=1,\ldots,N_{gt}} \Vert V_{2igt} - V_{2igt}^0 \Vert  = O(\lambda_t)$. As in the main text, I assume that $Z_{igt} = (X_{igt},Y_{igt},V_{igt})$ is a vector of data.

I next outline the key regularity conditions for convergence. First, for the estimator to be consistent the weighting matrix needs to by almost surely full rank in a neighborhood of $\boldsymbol{\theta}$ around the true parameter. A positive semi-definite matrix $\mathbf{Q}_{xx}$ is full rank if and only if its smallest eigenvalue is positive. Consequently, I quantify the \textit{almost sure} requirement by imposing a lower bound on the eigenvalues of the estimated matrix.  Let $\lambda_{min}(v_1,v_2,\boldsymbol{\theta})$ denote the smallest eigenvalue of $\mathbf{Q}_{xx}((v_1,v_2),\boldsymbol{\theta})$ and let $B(\boldsymbol{\theta}_{0t},\delta)$ be a ball or radius $\delta >0$ around $\boldsymbol{\theta}_{0t}$ and suppose that $V_{2igt}$ belongs to a compact set $\mathcal{V}_2$ with probability approaching one. Let  $\underline{\lambda}(V_{1igt}^{0},\boldsymbol{\theta}_{0t},\delta) \equiv \inf_{\theta \in B(\boldsymbol{\theta}_{0t},\delta)} \inf_{v_2 \in \mathcal{V}_2}\lambda_{min}(V_{1igt}^{0},v_2,\boldsymbol{\theta})$ be a lower bound on the eigenvalues of $\mathbf{Q}_{xx}$. I assume infimum holds over all values of $v_2$ to ensure that the matrix is full rank, even if the regressors are noisily estimated.

Second, the weighting matrix also needs to be sufficiently smooth in order to reduce the impact of measurement error from estimating $V_{2igt}$ and $\boldsymbol{\theta}$. I define its Sobolev-norm as
\begin{equation}
Q_{xx}^{\partial}(v_1,v_2,\theta) \equiv \sup_{0 \le \alpha_1 + \alpha_2 \le 3, \ \alpha_1,\alpha_2 \le 2} \left\Vert \frac{\partial^{\alpha_1 + \alpha_2}\mathbf{Q}_{xx}(v_1,v_2,\theta)}{\partial v_2^{\alpha_1}\boldsymbol{\theta}^{\alpha_2}}  \right\Vert
\label{eq:sobolevenorm_Qxx}
\end{equation}
Equation \eqref{eq:sobolevenorm_Qxx} indicates the derivatives of the weighting matrix up to order three need to be bounded. In settings without a generated regressor problem, i.e. $V_{2igt} = V_{2igt}^0$, we typically only require smoothness conditions over $\boldsymbol{\theta}$. In this case, however, bounding the derivatives with respect to $v_2$ as well, allows us to control the generated regressor error. In particular, I require that certain moments of the Sobolev norm need to be bounded.

In addition, the following regularity conditions have to be satisfied.

\begin{assumptxt}[Regularity Conditions]  \quad (i) There exists a $\theta_{0t} \in int\left(\Theta\right)$ such that $\forall \delta > 0$, $
	\inf_{\Vert \boldsymbol{\theta} - \boldsymbol{\theta}_{0t} \Vert > \delta } \mathcal{R}_t(\boldsymbol{\theta}) > \mathcal{R}_t(\boldsymbol{\theta}_{0t})$,
	(ii) $\mathbf{Q}_{xx}(V_{igt};\boldsymbol{\theta})$ is three-times continuously differentiable almost surely and $\mathbb{E}[\sup_{\theta \in \Theta} \sup_{v_2 \in \mathcal{V}_2} \ (Q^{\partial}_{xx}(V_{1igt}^0,v_2,\theta))^4] < \infty$,
	(iii) $\mathbb{E}[\Vert X_{igt}\Vert^4], \mathbb{E}[\Vert Y_{igt} \Vert^2] < \infty$, (iv) $\underline{\lambda}(V_{1igt}^0,\boldsymbol{\theta}_{0t},\delta) > \underline{\lambda} > 0$ almost surely for $\nu > 0$. (v) $H_{0t} \equiv \mathbb{E}\left[\tfrac{\partial}{\partial (\boldsymbol{\theta},\boldsymbol{\beta})'}\psi(Z_{igt},\boldsymbol{\theta}_{0t}) \right]$ is full rank, (vi) $\Omega \equiv \mathbb{E}\left[ \rho_{gt} \overline{\psi}_{g}(\boldsymbol{Z}_{gt},\boldsymbol{\theta}_{0t}) \overline{\psi}_{g}(\boldsymbol{Z}_{gt},\boldsymbol{\theta}_{0t})' \right]$ is positive-definite, (vii) $\max_{ig} \Vert V_{2igt} - V_{2igt}^0 \Vert = O_p(\tau_t)$, and (viii) $\tau_t\sqrt{G_t} = o(1)$ and $(G_t,N_t) \to \infty$ as $t \to \infty$.
	\label{assumptxt:regularity}		
\end{assumptxt}

Condition (i) is an identification condition that says that the true weighting matrix is the unique minimizer of the residuals. This is satisfied as long as the parametric family nests the conditional mean and the true criterion has a unique minimum. Condition (ii) imposes bounds on the moments of the Sobolev norm that hold uniformly over $(\boldsymbol{\theta},v_2)$. Condition (iii) are bounds on the moments of the endogenous variable $X_{igt}$ and $Y_{igt}$. Condition (iv) is a full rank condition for the average causal effect. Condition (v) is a rank condition on the system of equation that is similar to non-colinearity. Condition (vi) says the group-level covariance matrix is non-degenerate and finite. Condition (vii) states the rate of convergence of the generated regressors. Condition (viii) states that the rate needs to be more accurate than the rate of growth of the groups $G_t$.

\newpage

\section{Proofs}

\subsection{Main Proofs}



\begin{proof}[Proof of Theorem \ref{lem:randomcoef_identification} (\nameref{lem:randomcoef_identification})]
	By (ii) $X_{ig} \ \indep \ \tau_{ig} \mid V_{ig}$. By the decomposition property in Lemma \ref{lem:axioms}, $X_{ig}X_{ig}' \ \indep \ \tau_{ig} \mid V_{ig}$ and by (i) $Y_{ig} = X_{ig}'\tau_{ig}$, which means that 
	\begin{align*}
		\mathbf{Q}_{xy}(v) \equiv \mathbb{E}[X_{ig}Y_{ig} \mid V_{ig} = v] &= \mathbb{E}[X_{ig}X_{ig}'\tau_{ig} \mid V_{ig} = v] \\
		&=   \mathbb{E}[X_{ig}X_{ig}'\mid V_{ig} = v] \mathbb{E}[\tau_{ig} \mid V_{ig} = v] \\
		&= \mathbf{Q}_{xx}(v)\tau(v).
	\end{align*}
	If $\mathbf{Q}_{xx}(v)$ is almost surely full rank then $\tau(v) = \mathbf{Q}_{xx}(v)^{-1}\mathbf{Q}_{xy}(v)$ almost surely. Since $\mathcal{F}$ is coarser than $V_{ig}$, $\mathbb{E}[\tau_{ig} \mid V_{ig},\mathcal{F}] = \mathbb{E}[\tau_{ig} \mid V_{ig}]$ and 
	\begin{align*}
		\int \mathbf{Q}_{xx}(v)^{-1}\mathbf{Q}_{xy}(v) \ dF(v \mid \mathcal{F}) &= \int \tau(v) \ dF(v \mid \mathcal{F}) = \tau
	\end{align*}
	Finally, by the law of iterated expectations
	\begin{align*}
	\mathbb{E}[\mathbf{Q}_{xx}(V_{ig})^{-1}X_{ig}Y_{ig} \mid \mathcal{F}] &= \mathbb{E}[\mathbb{E}[\mathbf{Q}_{xx}(V_{ig})^{-1}X_{ig}Y_{ig} \mid V_{ig},\mathcal{F}]\mid \mathcal{F}] \\
	&= \mathbb{E}[\mathbf{Q}_{xx}(V_{ig})^{-1}\mathbf{Q}_{xy}(V_{ig}) \mid \mathcal{F}] \\
	&=\tau.
	\end{align*}
	
\end{proof}

\begin{proof}[Proof of Theorem \ref{thm:fullcontrolfunction} (\nameref{thm:fullcontrolfunction}) ]
	Part (i): I represent $\{ig\}'s$ treatment indicator as $D_{ig} = \mathcal{H}(C_{ig},\Psi_g^*,\eta_{ig})$ where $\mathcal{H}$ is a measurable function and $\eta_{ig} \mid C_{ig},\Psi_g^* \sim F(\eta \mid c,\Psi^*) $ is an unobserved participation shock. Since we can always define the participation shock as $\eta = D_{ig} - \mathbb{P}(D_{ig} =1 \mid C_{ig} = c,\Psi_g^* = \Psi^*)$, this form does not entail any loss of generality.
	
	Let $\zeta_{ig} \equiv (\tau_{ig},\eta_{ig},C_{ig})$. By \nameref{assumptxt:randomsampling} and \nameref{assumptxt:dyadicnetwork},
	\begin{equation}
	\label{eq:indep_covariates}
	\zeta_{ig} \ \indep \ \{ U_{ijg} \}_{j\ne i}^{N_g}, \{\zeta_{jg}\}_{j \ne i}^{N_g} \mid \Psi_g^*
	\end{equation}
	By \eqref{eq:indep_covariates}, as well as the weak union and decomposition properties in Lemma \ref{lem:axioms},
	\begin{align*}
		\zeta_{ig} \ \indep \ \{ U_{ijg} \}_{j\ne i}^{N_g}, \{\zeta_{jg}\}_{j \ne i}^{N_g} &\mid \eta_{ig},C_{ig},\Psi_g^* \\ \implies \quad \tau_{ig} \ \indep \ \{ U_{ijg} \}_{j\ne i}^{N_g}, \{\eta_{jg},C_{jg}\}_{j \ne i}^{N_g} &\mid \eta_{ig},C_{jg},\Psi_g^*
	\end{align*}
	The second line subsets the relevant variables on either side of the independence relation. The participation decisions are functions of personal covariates and selection shocks. Similarly, the friendship vector $\{ig\}$ only depends on the list of preference shocks $(U)$ and covariates $(C)$. Since $X_{ig}' = (1,D_{ig})\otimes \left(1,\varphi\left(i,A_g,C_g,N_g\right) \right)$ and $L_{ig} = \sum_{j = 1,j \ne i}^{N_g}A_{ijg}$, that means that $(L_{ig},X_{ig})$ are both measurable with respect to $\{U_{ijg}\}_{j \ne i}^{N_g},\{\zeta_{jg}\}_{j=1}^{N_g}$.  Then by the decomposition property,
	\begin{equation}
	\tau_{ig} \ \indep \ (X_{ig},L_{ig}) \mid \eta_{ig}, C_{ig},\Psi_g^*  .
	\end{equation}
	By \nameref{assumptxt:confounders}, the outcome heterogeneity is conditionally independent of the selection unobservables, $\tau_{ig} \ \indep \ \eta_{ig} \mid C_{ig},\Psi_g^*$. By the contraction and decomposition properties,
	\begin{equation}
	\tau_{ig} \ \indep \ (X_{ig},L_{ig},\eta_{ig}) \mid C_{ig},\Psi_g^*  \quad \implies \quad \tau_{ig} \ \indep \ (X_{ig},L_{ig}) \mid C_{ig},\Psi_g^*.
	\end{equation}
	Part (ii): We will show that for all $i \in \{1,\ldots,N_g\}$,
	\begin{align*}
		& \mathbb{P}(X_{ig} \le x, L_{ig} \le \ell, Y_{ig} \le y \mid C_{ig} = c,\Psi_g^* = \Psi) \\
		&\quad = \mathbb{P}(X_{1g} \le x, L_{1g} \le \ell, Y_{1g} \le y \mid C_{1g} = c,\Psi_g^* = \Psi).
	\end{align*}
	Let $e_1$ be an $N_g \times 1$ vector with a one in the first entry and zero otherwise, and let $\tau_g$ be an $N_g \times (2+2k)$ matrix of random coefficients. By construction, $C_{1g}' = e_1'C_g$, $D_{1g} = e_1'D_g$, and $\tau_{1g}' = e_1'\tau_g$. Similarly, $L_{1g} = e_1'A_g1_{N_g \times 1}$, where $1_{N_g \times 1}$ is an $N_g \times 1 $ vector of ones.

	The key is to write the variables for $\{ig\}$ in terms of a permutation of the objects for individual $\{1g\}$, and then show that the distribution is invariant to permutations. By \nameref{assumptxt:exchangeableinteractions}, $X_{ig} = (1,D_{ig})\otimes\left( 1, \varphi(1,\Pi_{i1}A_g\Pi_{i1}',\Pi_{i1}D_g,\Pi_{i1}C_g,N_g) \right)$, where $\Pi_{i1}$ is the rotation matrix defined in the assumption. By construction, $C_{ig}'=e_1\Pi_{i1}C_g$, $D_{ig} = e_1'\Pi_{i1}D_g$, and $\tau_{ig}' = e_1'\Pi_{i1}\tau_g$, since $e_i = e_1\Pi_{i1}$. Similarly, $L_{ig} = e_1\Pi_{ij}A_g\Pi_{ij}'1_{N_g \times 1}$, since $\Pi_{ij}'1_{N_g \times 1} = 1_{N_g \times 1}$. By \nameref{assumptxt:randomsampling} and \nameref{assumptxt:dyadicnetwork} all the underlying shocks $\{\zeta_{jg}\}_{i=1}^{N_g},\{U_{ijg}\}_{i,j=1}^{N_g}$ are i.i.d., which means that 
	$$  (\Pi_{i1}A_g\Pi_{i1}',\Pi_{i1}D_g, \Pi_{i1}C_g,\Pi_{i1}\tau_g) \mid \{\Psi_g = \Psi\} \sim (A_g,D_g, C_g,\tau_g) \mid \{\Psi_g = \Psi\}. $$
	Since the key variables of interest are deterministic functions of the above variables,
	$$ (X_{ig},L_{ig},\tau_{ig},C_{ig},Y_{ig}) \mid \{\Psi_g = \Psi\} \sim (X_{1g},L_{1g},\tau_{1g},C_{1g},Y_{ig}) \mid \{\Psi_g = \Psi\}.$$
	To complete the proof we just compute the conditional probability.	
	
\end{proof}

\begin{proof}[Proof of Theorem \ref{thm:closedformtau} (\nameref{thm:closedformtau})] 	
	I make use of the mixture representation of $\mathbf{Q}_{xx}$ derived in Lemma \ref{lem:mixturepresentationQxx}, assuming \nameref{assumptxt:randomsampling}, \nameref{assumptxt:confounders} and \nameref{assumptxt:dyadicnetwork}. If $V_{ig} = (C_{ig},\Psi_g^*,L_{ig})$, then the conditional distribution of the network propensity score is degenerate and hence
	$$\mathbf{Q}_{xx}(v) =  \begin{pmatrix} 1 & \widetilde{\varphi}_1(p_{f},l) \\ \widetilde{\varphi}_1(p_{f},l,p_{f},l) & \widetilde{\varphi}_2(p_{f},l) \end{pmatrix} \otimes \begin{pmatrix} 1 & p_d \\ p_d & p_d \end{pmatrix}.$$
	When $\varphi(t,l) = t/l$, then $\widetilde{\varphi}_1(p_{f},l) = p_f$ and $\widetilde{\varphi}_1(p_{f},l) = p_f(1-p_f)/l + p_f^2$ by using the moments in Lemma \ref{lem:conditionaldist}. The inverse of kronecker product of matrices is equal to the inverse of the kronecker products, which means that
	\begin{align*}
		\mathbf{Q}_{xx}(v)^{-1} &= \begin{pmatrix} 1 & p_f \\ p_f & \frac{p_f(1-p_f)}{l} + p_f^2 \end{pmatrix}^{-1} \otimes \begin{pmatrix} 1 & p_d \\ p_d & p_d \end{pmatrix}^{-1} \\
		&= \left(\frac{1}{p_d(1-p_d)}\right)\left(\frac{l}{p_f(1-p_f)}\right)\begin{pmatrix} \frac{p_f(1-p_f)}{l} + p_f^2 & - p_f \\ - p_f & 1 \end{pmatrix} \otimes \begin{pmatrix} p_d & -p_d \\ -p_d & 1 \end{pmatrix}
	\end{align*}
	We can write the regressors in kronecker product form as $X_{ig}' = (1,T_{ig}/L_{ig}) \otimes (1,D_{ig})$. Hence $\mathbf{Q}_{xx}(V_{ig})^{-1}X_{ig}Y_{ig}$ multiplies two kronecker products. I use the property that for conformable matrices $(M_1,M_2,M_3,M_4)$, $(M_1 \otimes M_2)(M_3 \otimes M_4) = (M_1M_2) \otimes (M_3M_4)$. After some algebraic manipulations we can show that
	\begin{align*}
		\mathbf{Q}_{xx}(V_{ig})^{-1}X_{ig}Y_{ig}
		&= \begin{pmatrix} 1 + \frac{p_fL_{ig} - T_{ig}}{1-p_f} \\ \frac{-p_fL_{ig} + T_{ig}}{p_f(1-p_f)} \end{pmatrix} \otimes \begin{pmatrix} \frac{(1-D_{ig})Y_{ig}}{1-p_d} \\ \frac{D_{ig}Y_{ig}}{p_d} - \frac{(1-D_{ig})Y_{ig}}{1-p_d} \end{pmatrix}.
	\end{align*}
	By Theorem \ref{thm:fullcontrolfunction}, $V_{ig}$ satisfies $\tau_{ig} \ \indep \ X_{ig} \mid V_{ig}$. Assuming the inverse of $\mathbf{Q}_{xx}(V_{ig})$ is well defined then we can apply Theorem \ref{lem:randomcoef_identification} to show that $\tau = \mathbb{E}[\mathbf{Q}_{xx}(V_{ig})^{-1}X_{ig}Y_{ig} \mid \mathcal{F}]$. We can obtain the individual coefficients $(\alpha,\beta,\gamma,\delta)$ by expanding the kronecker product inside the expectation.
	
\end{proof}

\begin{proof}[Proof of Lemma \ref{lem:conditionaldist} (\nameref{lem:conditionaldist})]
	Let $\widetilde{C}_{ig} \equiv (C_{ig},\Psi_g^*)$ and $A_{ig} = \{A_{ijg}\}_{j=1,j \ne i}^{N_g}$. If \nameref{assumptxt:randomsampling} and \nameref{assumptxt:dyadicnetwork} holds, then we can apply Lemma \ref{lem:egocentric_factorization} (\nameref{lem:egocentric_factorization}) to show that
	\begin{equation}
	\label{eq:lik_factorization}
	\mathbb{P}(D_g,A_{ig} \mid \widetilde{C}_{ig}) = \mathbb{P}(D_{ig} \mid \widetilde{C}_{ig}) \prod_{j\ne i}^{N_g} \mathbb{P}(D_{jg},A_{ijg} \mid \widetilde{C}_{ig})
	\end{equation}
	By Bayes' rule, $\mathbb{P}(D_{jg},A_{ijg} \mid \tilde{C}_{ig}) = \mathbb{P}(D_{jg} \mid A_{ijg}, \tilde{C}_{ig}) \mathbb{P}(A_{ijg} \mid \tilde{C}_{ig})$ and substituting into \eqref{eq:lik_factorization}
	\begin{align*}
		&\mathbb{P}(D_g,A_{ig} \mid \widetilde{C}_{ig}) \\
		&= \mathbb{P}(D_{ig} \mid \widetilde{C}_{ig}) \prod_{j\ne i}^{N_g} \mathbb{P}(A_{ijg} \mid \widetilde{C}_{ig}) \prod_{j:A_{ijg} = 1}^{N_g} \mathbb{P}(D_{jg} \mid A_{ijg} = 1, \widetilde{C}_{ig}) \prod_{j:A_{ijg} = 0}^{N_g} \mathbb{P}(D_{jg} \mid A_{ijg} = 0, \widetilde{C}_{ig})
	\end{align*}
	This proves that $D_{ig} \ \indep \{D_{jg},A_{ijg}\}_{j \ne i}^{N_g} \mid \widetilde{C}_{ig}$. Let $L_{ig} \equiv \sum_{j \ne i}A_{ijg}$ be the total friends, $T_{ig} \equiv \sum_{j \ne i} D_{jg}A_{ijg}$ the total number of treated friends and $M_{ig} \equiv \sum_{j \ne i} D_{jg}(1-A_{ijg})$ be the total number of treated non-friends. Consequently, by the decomposition property in Lemma \ref{lem:axioms},
	\[
	D_{ig} \ \indep (L_{ig},T_{ig},M_{ig}) \mid \widetilde{C}_{ig} \quad \implies \quad D_{ig} \ \indep (L_{ig},T_{ig}) \mid \widetilde{C}_{ig} 
	\]
	Furthermore, the likelihood can be factorized in terms of four sets of Bernoulli random variables, with a distinct event probability and $(1,N_g,L_{ig},N_g-L_{ig})$ trials, respectively.
	
	Let $p_f(\tilde{C}_{ig})$ and $p_{m}(z)$ denote the participation probability of friends and non-friends. Then
	\begin{align}
		\begin{split}
			\mathbb{P}(A_{ijg} \mid \widetilde{C}_{ig})&= p_\ell(\widetilde{C}_{ig})^{A_{ijg}}(1-p_\ell(\widetilde{C}_{ig}))^{1-A_{ijg}}  \\
			\mathbb{P}(D_{jg} \mid A_{ijg} = 1, \widetilde{C}_{ig}) &= p_f(\widetilde{C}_{ig})^{D_{jg}}(1-p_f(\widetilde{C}_{ig}))^{1-D_{jg}}  \\
			\mathbb{P}(D_{jg} \mid A_{ijg} = 0, \widetilde{C}_{ig}) &= p_{m}(\widetilde{C}_{ig})^{D_{jg}}(1-p_{m}(\widetilde{C}_{ig}))^{1-D_{jg}}
		\end{split}
		\label{eq:proofcondprob}
	\end{align}
	The product of the probabilities is
	\begin{align}
		\begin{split}
			\prod_{j\ne i}^{N_g} \mathbb{P}(A_{ijg} \mid \widetilde{C}_{ig})
			&= p_\ell(\widetilde{C}_{ig})^{L_{ig}}(1-p_\ell(\widetilde{C}_{ig}))^{N_g-L_{ig}}  \\
			\prod_{j:A_{ijg} = 1}^{N_g} \mathbb{P}(D_{jg} \mid A_{ijg} = 1, \widetilde{C}_{ig}) &= p_f(\widetilde{C}_{ig})^{T_{ig}}(1-p_f(\widetilde{C}_{ig}))^{L_{ig}-T_{ig}}  \\
			\prod_{j:A_{ijg} = 0}^{N_g}\mathbb{P}(D_{jg} \mid A_{ijg} = 0, \widetilde{C}_{ig}) &= p_{m}(\widetilde{C}_{ig})^{M_{ig}}(1-p_{m}(\widetilde{C}_{ig}))^{N_g-L_{ig}-M_{ig}}
		\end{split}
		\label{eq:proofcondprob2}
	\end{align}
	Let $\mathcal{B}_{(d,l,t,m)}$ be the set of permutations of treatment and link formation decisions that produce $B_{ig} = (d,l,t,m)$, where  $B_{ig} \equiv (D_{ig},L_{ig},T_{ig},M_{ig})$. Then $\mathbb{P}_{B_{ig} \mid \widetilde{C}_{ig}}(d,l,t,m)$ is equal to $\sum_{(D_g,A_{ijg}) \in \mathcal{B}_(d,l,t,m)} \mathbb{P}(D_g,A_{ijg})$. The resulting distribution has the form
	\begin{align*}
		\label{eq:factorbinomialpoisson}
		\begin{split}
			L_{ig} \mid \widetilde{C}_{ig} &\sim \text{Binom}(p_\ell(\widetilde{C}_{ig}),N_g) \\
			T_{ig} \mid L_{ig},\widetilde{C}_{ig} &\sim \text{Binom}(p_f(\widetilde{C}_{ig}),L_{ig}) \\
			D_{ig} \mid T_{ig},L_{ig},\widetilde{C}_{ig} &\sim \text{Bernoulli}(p_d(\widetilde{C}_{ig})) \\			
			M_{ig} \mid D_{ig},T_{ig},L_{ig},\widetilde{C}_{ig} &\sim \text{Binom}(p_{m}(\widetilde{C}_{ig},N_g-L_{ig})) \\
		\end{split}
	\end{align*}
	To complete the statement of the lemma, we only report the distribution of $(D_{ig},T_{ig}) \mid \widetilde{C}_{ig},L_{ig}$, which does not depend on $M_{ig}$. The resulting distribution does not involve $p_m(\widetilde{C}_{ig})$.
	
\end{proof}

\begin{proof}[Proof Theorem \ref{thm:expostbalancing} (\nameref{thm:expostbalancing})]	
	If \nameref{assumptxt:randomsampling} and \nameref{assumptxt:dyadicnetwork} hold, then we can apply Lemma \ref{lem:conditionaldist} to show that $D_{ig} \ \indep \ (T_{ig},L_{ig}) \mid C_{ig}$ and
	\begin{align*}
		\begin{split}
			D_{ig} \mid T_{ig},L_{ig},C_{ig},\Psi_g^*  &\sim\text{Bernoulli}(p_{dig}) \\
			T_{ig} \mid L_{ig},C_{ig},\Psi_g^*  &\sim \text{Binomial}(p_{fig},L_{ig})
		\end{split}
	\end{align*}
	The distribution of $(D_{ig},T_{ig},L_{ig})$ is parametrized by $P_{ig} \equiv (p_{dig},p_{fig},L_{ig})$, which means that $(D_{ig},T_{ig},L_{ig}) \mid C_{ig},\Psi_g^*, P_{ig} \sim (D_{ig},T_{ig},L_{ig}) \mid P_{ig}$. Consequently, the network propensity score and the group size summarizes all the pretreatment information and 
	\[
	(D_{ig},T_{ig},L_{ig}) \ \indep \ C_{ig},\Psi_g^* \mid P_{ig}.
	\]
	By construction $X_{ig}$ is a measurable function of  $(D_{ig},L_{ig},T_{ig})$. By applying the decomposition property in Lemma \ref{lem:axioms},
	\begin{equation}
	X_{ig} \ \indep \ C_{ig},\Psi_g^* \mid P_{ig}.
	\label{eq:balancingproof_exante}
	\end{equation}
	This shows that $P_{ig}$ is a balancing score.
	
	If \nameref{assumptxt:randomsampling},  \nameref{assumptxt:confounders} and \nameref{assumptxt:dyadicnetwork} hold, then Theorem \ref{thm:fullcontrolfunction} states that $\tau_{ig} \ \indep \ (X_{ig}, L_{ig}) \mid C_{ig},\Psi_g^*$ which implies $\tau_{ig} \ \indep \ \mid C_{ig},\Psi_g^*,L_{ig}$. By combining the redundancy and weak union properties, it follows that $\tau_{ig} \ \indep \ X_{ig} \mid C_{ig},\Psi_g^*,P_{ig}$. Consequently, by \eqref{eq:balancingproof_exante} and the contraction property, $(\tau_{ig},C_{ig},\Psi_g^*) \ \indep \ X_{ig} \mid P_{ig}$. We can simplify the final expression by the decomposition property,
	\[
	\tau_{ig} \ \indep \ X_{ig} \mid P_{ig}.
	\] 
	
\end{proof}

\begin{proof}[Proof of Theorem \ref{thm:bounds_pseudometrics} (\nameref{thm:bounds_pseudometrics})]
	
	The difference between two friend propensity scores, $d_f$, is equal to
	\begin{align*}
		&d_f =  \left\Vert \int \mathcal{H}(C^*,\Psi_g^*)\left[ \frac{\mathcal{L}(C_{ig},C^*,\Psi_g^*)}{p_\ell(C_{ig},\Psi_g^*)} - \frac{\mathcal{L}(C_{jg},C^*,\Psi_g^*)}{p_\ell(C_{jg},\Psi_g^*)}\right]dF(C^* \mid \Psi_g^*) \right\Vert \\
		&\le \int \left\Vert \mathcal{H}(C^*,\Psi_g^*)\left[ \frac{\mathcal{L}(C_{ig},C^*,\Psi_g^*)}{p_\ell(C_{ig},\Psi_g^*)} - \frac{\mathcal{L}(C_{jg},C^*,\Psi_g^*)}{p_\ell(C_{jg},\Psi_g^*)}\right] \right\Vert dF(C^* \mid \Psi_g^*) \\
		&\le \int \left\Vert  \frac{\mathcal{L}(C_{ig},C^*,\Psi_g^*)}{p_\ell(C_{ig},\Psi_g^*)} - \frac{\mathcal{L}(C_{jg},C^*,\Psi_g^*)}{p_\ell(C_{jg},\Psi_g^*)} \right\Vert dF(C^* \mid \Psi_g^*) \\
		&\le \int \left\Vert  \frac{\mathcal{L}(C_{ig},C^*,\Psi_g^*)-\mathcal{L}(C_{jg},C^*,\Psi_g^*)}{p_\ell(C_{ig},\Psi_g^*)} + \left[ \frac{\mathcal{L}(C_{jg},C^*,\Psi_g^*)}{p_\ell(C_{ig},\Psi_g^*)}-\frac{\mathcal{L}(C_{jg},C^*,\Psi_g^*)}{p_\ell(C_{jg},\Psi_g^*)} \right]  \right\Vert dF(C^* \mid \Psi_g^*) \\
		&\le \frac{1}{p_{\ell}(C_{ig},\Psi_g^*)}\int \Vert \mathcal{L}(C_{ig},C^*,\Psi_g^*) - \mathcal{L}(C_{jg},C^*,\Psi_g^*) \Vert dF(C^* \mid \Psi_g^*) + \\
		&\quad + \left[ \frac{1}{p_\ell(C_{ig},\Psi_g^*)}-\frac{1}{p_\ell(C_{jg},\Psi_g^*)} \right] \int \left\Vert   \mathcal{L}(C_{jg},C^*,\Psi_g^*) \right\Vert dF(C^* \mid \Psi_g^*) \\
		&\le \frac{1}{p_{\ell}(C_{ig},\Psi_g^*)}d_{\Psi_g^*} + \left[ \frac{1}{p_\ell(C_{ig},\Psi_g^*)}-\frac{1}{p_\ell(C_{jg},\Psi_g^*)} \right] p_\ell(C_{jg},\Psi_g^*) \\
		&\le \frac{1}{p_{\ell}(C_{ig},\Psi_g^*)}d_{\Psi_g^*} + \frac{1}{p_{\ell}(C_{ig},\Psi_g^*)}(p_\ell(C_jg,\Psi_g^*)-p_\ell(C_{ig},\Psi_g^*)) \\
		&\le \frac{1}{p_{\ell}(C_{ig},\Psi_g^*)}d_{\Psi_g^*} + \frac{1}{p_{\ell}(C_{ig},\Psi_g^*)}d_{\Psi_g^*}
	\end{align*}
	
\end{proof}

\begin{proof}[Proof of Lemma \ref{lem:mixturepresentationQxx} (\nameref{lem:mixturepresentationQxx})]
	By construction we can write the covariates as $X_{ig}' = (1,\varphi(T_{ig},L_{ig})) \otimes (1,D_{ig})$. Therefore we can write $X_{ig}X_{ig}'$ as
	\[
	X_{ig}X_{ig}' = \begin{pmatrix}1 & \varphi(T_{ig},L_{ig})' \\ \varphi(T_{ig},L_{ig}) & \varphi(T_{ig},L_{ig})\varphi(T_{ig},L_{ig})' \end{pmatrix} \otimes \begin{pmatrix} 1 & D_{ig} \\ D_{ig} & D_{ig} \end{pmatrix}
	\]
	Define the functions
	\begin{align*}
	\widetilde{\varphi}_1(p_f,l) &= \mathbb{E}[\varphi(T_{ig},L_{ig}) \mid p_{fig} = p_f,L_{ig} = l] \\
	\widetilde{\varphi}_2(p_f,l) &= \mathbb{E}[\varphi(T_{ig},L_{ig})\varphi(T_{ig},L_{ig})' \mid p_{fig} = p_f,L_{ig} = l].
	\end{align*}
	Under Lemma \ref{lem:conditionaldist}, $D_{ig}$ is conditionally independent of $(T_{ig},L_{ig})$ given $(C_{ig},\Psi_g^*,L_{ig})$, and the distributions are parametrized by the components of the network propensity score. Therefore we can decompose the conditional moments of $X_{ig}X_{ig}'$ as
	\[
	\mathbb{E}[X_{ig}X_{ig}' \mid C_{ig} = c,\Psi_g^* = \Psi, L_{ig} = l] = \begin{pmatrix}1 & \widetilde{\varphi}_1(p_f,l)'\\ \widetilde{\varphi}_1(p_f,l) & \widetilde{\varphi}_2(p_f,l)' \end{pmatrix} \otimes \begin{pmatrix} 1 & p_d \\ p_d & p_d \end{pmatrix}
	\]	
	Since $V_{ig}$ is measurable with respect to $(C_{ig},\Psi_g^*,L_{ig})$ we can apply the law of iterated expectations to obtain
	\begin{align}
		\mathbf{Q}_{xx}(v) = \int  \begin{pmatrix} 1 & \widetilde{\varphi}_1(p_{f},l) \\ \widetilde{\varphi}_1(p_{f},l,p_{f},l) & \widetilde{\varphi}_2(p_{f},l) \end{pmatrix} \otimes \begin{pmatrix} 1 & p_d \\ p_d & p_d \end{pmatrix} dF(p_{d},p_{f},l \mid V_{ig} = v).		
	\end{align}

\end{proof}

\begin{proof}[Proof of Lemma \ref{lem:quasisaturation} (\nameref{lem:quasisaturation})]
	Let $\widetilde{C}_{ig} \equiv (C_{ig},\Psi_g^*)$ and $X_{ig}^* \equiv (X_{ig},L_{ig})$. If \nameref{assumptxt:randomsampling}, \nameref{assumptxt:confounders} and \nameref{assumptxt:dyadicnetwork} hold, then we can apply Theorem \ref{thm:expostbalancing} to show that $(X_{ig},L_{ig}) \ \indep \ \widetilde{C}_{ig} \mid p_d(\widetilde{C}_{ig}),p_f(\widetilde{C}_{ig}),L_{ig}$.
	
	Under \nameref{assumptxt:randomsampling}, \nameref{assumptxt:confounders}  and \nameref{assumptxt:confounders} we can apply Theorem \ref{thm:fullcontrolfunction} to show that $(X_{ig},L_{ig}) \ \indep \ \tau_{ig} \mid \widetilde{C}_{ig}$. By the weak union property, $(X_{ig},L_{ig}) \ \indep \ \tau_{ig} \mid \widetilde{C}_{ig}, p_d(\widetilde{C}_{ig}),p_f(\widetilde{C}_{ig}),L_{ig}$. Applying the contraction axiom,
	\[
	(X_{ig},L_{ig}) \ \indep \ (\tau_{ig},\widetilde{C}_{ig}) \mid p_d(\widetilde{C}_{ig}),p_f(\widetilde{C}_{ig}),L_{ig}
	\]	
	Since $V_{ig}$ is $(\widetilde{C}_{ig},L_{ig})-$measurable, we can apply the weak union property, as
	\[
	(X_{ig},L_{ig}) \ \indep \ (\tau_{ig},\widetilde{C}_{ig}) \mid p_d(\widetilde{C}_{ig}),p_f(\widetilde{C}_{ig}),L_{ig},V_{ig}
	\]
	` By decomposition $X_{ig} \ \indep \ \tau_{ig} \mid p_d(\widetilde{C}_{ig}),p_f(\widetilde{C}_{ig}),L_{ig},V_{ig}$. Since by assumption of the theorem $p_d(\widetilde{C}_{ig}),p_f(\widetilde{C}_{ig}),L_{ig} \ \indep \ \tau_{ig} \mid V_{ig}$, we can apply the contraction axiom again to show that
	\[
	(X_{ig},p_d(\widetilde{C}_{ig}),p_f(\widetilde{C}_{ig}),L_{ig}) \ \indep \ \tau_{ig} \mid V_{ig}.
	\]
	Finally, by the decomposition property, $X_{ig} \ \indep \ \tau_{ig} \mid V_{ig}$.
	
\end{proof}

\newpage

\subsection{Proof Asymptotics}


\begin{proof}[Proof of Theorem \ref{thm:normality_averageeffects} (\nameref{thm:normality_averageeffects}) ]
	Define the square residual function as $\mathcal{R}(z,v_2,\theta) \equiv r((x,y,(v_1,v_2)),\theta)^2$, so that the estimated and population criterion functions can be written as
	\[
	\widehat{\mathcal{R}}_t(\boldsymbol{\theta}) \equiv \frac{1}{G_t\bar{n}}\sum_{ig} \mathcal{R}(Z_{igt},V_{2igt},\boldsymbol{\theta}),
	\]
	\[
	\mathcal{R}_t(\boldsymbol{\theta}) \equiv \mathbb{E}[\mathcal{R}(Z_{igt},V_{2igt}^0,\boldsymbol{\theta})].
	\]
	Our first task is to prove uniform convergence of the criterion function by verifying the conditions of Lemma \ref{lem:uniformconvergence_sampleaverages}. First, by Assumption (vii) $\max_{ig} \Vert V_{2igt} - V_{2igt}^0 \Vert = O_p\left(\lambda_t \right)$. By assumption (viii), $\sqrt{G_t}\lambda_t = o(1)$ which means that the maximum discrepancy is $\tau_t = o(1)$, as required.
	
	Second we verify the uniform bounds on the moments. Assumptions (ii) and (iii) in \nameref{assumptxt:regularity} imply that $\mathcal{R}_t(Z_{igt},V_{2igt}^0,\boldsymbol{\theta})$ has bounded moments. Conversely, let $\mathcal{R}_{igt}^V$ and $S_{igt}$ be uniform bounds on the derivatives $\tfrac{\partial \mathcal{R}}{\partial v_2}$ and the score $\psi_{q} \equiv \tfrac{\partial \mathcal{R}}{\partial \theta}$ as defined in \eqref{eq:envelope_RA} and \eqref{eq:envelope_S}. These bounds hold uniformly over $\boldsymbol{\tau}$ because the average effect parameter does not enter $\mathcal{R}$. The bound on the expectation of the Sobolev-norm in \nameref{assumptxt:regularity} part (ii) and Lemma \eqref{lem:uniformboundsderivatives} imply that $\mathbb{E}[\mathcal{R}_{igt}^V]<\infty$ and $\mathbb{E}[S_{igt}]<\infty$. Consequently, $\mathcal{R}$ satisfies the requirements of Lemma \ref{lem:uniformconvergence_sampleaverages}, and hence
	\begin{equation}
	\sup_{\boldsymbol{\theta} \in \Theta} \Vert \widehat{\mathcal{R}}_t(\boldsymbol{\theta}) - \mathcal{R}_t(\boldsymbol{\theta}) \Vert \to^{p} 0
	\label{eq:uniformconsistency_criterion}
	\end{equation}
	Our next task is to show that $\widehat{\boldsymbol{\theta}}_t$ is consistent. By \nameref{assumptxt:regularity} part (i) for any $\delta > 0$ there exists a $\nu > 0$ such that 
	\begin{align*}
		&\mathbb{P}\left( \left\Vert \widehat{\boldsymbol{\theta}}_t-\boldsymbol{\theta}_{0t} \right\Vert > \delta \right) \\
		&\le \mathbb{P}(\mathcal{R}_t(\widehat{\boldsymbol{\theta}}_t) - \mathcal{R}_t(\boldsymbol{\theta}_{0t}) \ge \nu )\\
		&= \mathbb{P}(\mathcal{R}_t(\widehat{\boldsymbol{\theta}}_t) - \widehat{\mathcal{R}}_t(\widehat{\boldsymbol{\theta}}_t) + \widehat{\mathcal{R}}_t(\widehat{\boldsymbol{\theta}}_t) - \mathcal{R}_t(\boldsymbol{\theta}_{0t}) \ge \nu ) &\text{ Adding/subtracting } \mathcal{R}_t(\widehat{\boldsymbol{\theta}}_t)\\
		&\le \mathbb{P}(\mathcal{R}_t(\widehat{\boldsymbol{\theta}}_t) - \widehat{\mathcal{R}}_t(\widehat{\boldsymbol{\theta}}_t) + \widehat{\mathcal{R}}_t(\boldsymbol{\theta}_{0t}) - \mathcal{R}(\boldsymbol{\theta}_{0t}) \ge \nu ) &\text{ Since } \widehat{R}_t(\widehat{\boldsymbol{\theta}}_t) \le \widehat{R}_t(\boldsymbol{\theta}_{0t}) \\
		&\le \mathbb{P}\left(2 \sup_{\boldsymbol{\theta} \in \Theta} \left\Vert \widehat{\mathcal{R}}_t(\boldsymbol{\theta}) - \mathcal{R}_t(\boldsymbol{\theta}) \right\Vert \ge \nu \right) &\text{Uniform Bound}\\
		&\to^p 0 &\text{ By \eqref{eq:uniformconsistency_criterion}}
	\end{align*}
	Consequently $\widehat{\boldsymbol{\theta}}_t \to^{p} \boldsymbol{\theta}_{0t}$.
	
	We now turn to the task of proving asymptotic normality. In a slight abuse of notation, I use $\psi(z,v_2,\boldsymbol{\tau},\boldsymbol{\theta})$ to denote the influence function $\psi((x,y,(v_1,v_2)),\boldsymbol{\tau},\boldsymbol{\theta})$
	\[
	o_p\left( \tau \right) = \frac{1}{G_tN_t}\sum_{ig} \psi(Z_{igt},V_{2igt},\widehat{\boldsymbol{\tau}}_t,\widehat{\boldsymbol{\theta}}_t)
	\]	
	By a first-order expansion
	\begin{align}
	\begin{split}
	0&= \frac{1}{G_tN}\sum_{ig} \psi(Z_{igt},V_{2igt}^0,\boldsymbol{\tau}_{0t},\boldsymbol{\theta}_{0t})  \\
	&\quad + \frac{1}{G_tN}\sum_{ig} \tfrac{\partial}{\partial v_2}\psi(Z_{igt},\widetilde{V}_{2igt},\widetilde{\boldsymbol{\tau}},\widetilde{\boldsymbol{\theta}}_t)\Delta_{igt}
	+ 	\frac{1}{NG_t}\sum_{ig}\left( \begin{matrix}
	\frac{\partial}{\partial \boldsymbol{\tau}'} \psi(Z_{igt},\widetilde{V}_{2igt},\widetilde{\boldsymbol{\tau}},\widetilde{\boldsymbol{\theta}}_t) \\
	\frac{\partial}{\partial \boldsymbol{\theta}'} \psi(Z_{igt},\widetilde{V}_{2igt},\widetilde{\boldsymbol{\tau}},\widetilde{\boldsymbol{\theta}}_t)
	\end{matrix}\right) \
	\left( \begin{matrix}  \widehat{\boldsymbol{\tau}}_t - \boldsymbol{\tau}_{0t} \\
	\widehat{\boldsymbol{\theta}}_t - \boldsymbol{\theta}_{0t} \end{matrix} \right)
	\label{eq:firstorder_paramthetabeta}
	\end{split}
	\end{align}

	Our next task is to show that the second term is $O_p\left(\lambda_t \sqrt{G_t} \right)$. To this end it is useful to decompose that influence function into two sets of equations $\psi = [\psi_q,\psi_{IW}]'$, for the weighting matrix and the average effects, respectively. Let $B(\boldsymbol{\theta}_{0t},\nu)$ denote a ball of radius $\nu$ around the true parameter. By assumption (iv) the smallest eigenvalue of ${Q}_{xx}$ is bounded by a fixed constant for $\boldsymbol{\theta} \in B(\boldsymbol{\theta}_{0t},\nu)$. Since $\widehat{\boldsymbol{\theta}}_t$ and $\widetilde{\boldsymbol{\theta}}_t$ are both consistent, the estimator is contained in the ball with probability approaching one as $(G_t,N_t) \to \infty$.
		
	Define $S_{igt}^V$ and $\psi_{IW,ig}^{\partial}$ as uniform upper bounds for the partial derivatives of $s$ and $\psi_{IW}$ as defined in \eqref{eq:envelope_SA} and \eqref{eq:envelope_psitheta}. Furthermore, let $\Delta_{max} \equiv \max_{ig}\Vert V_{2igt} - V_{2igt}^0 \Vert$ be the maximum discrepancy between the generated and true regressors. By the triangle inequality.
	\begin{align}
	\begin{split}
	&\left\Vert	\frac{1}{G_t\bar{n}_t}\sum_{ig} \tfrac{\partial}{\partial v_2}\psi(Z_{igt},\widetilde{V}_{2igt},\widetilde{\boldsymbol{\tau}},\widetilde{\boldsymbol{\theta}}_t)\Delta_{igt} \right\Vert \\
	&\le \frac{1}{G_t\bar{n}_t}\sum_{ig} \left\Vert \tfrac{\partial}{\partial v_2}\psi(Z_{igt},\widetilde{V}_{2igt},\widetilde{\boldsymbol{\tau}},\widetilde{\boldsymbol{\theta}}_t) \right\Vert \ \cdot \ \Delta_{max} \\
	&\le \frac{1}{G_t\bar{n}_t}\sum_{ig} \left( \left\Vert \tfrac{\partial s(Z_{igt},\widetilde{V}_{2igt},\widetilde{\boldsymbol{\tau}},\widetilde{\boldsymbol{\theta}}_t)}{\partial v_2} \right\Vert +
	\left\Vert \tfrac{\partial \psi_{IW}(Z_{igt},\widetilde{V}_{2igt},\widetilde{\boldsymbol{\tau}},\widetilde{\boldsymbol{\theta}}_t)}{\partial v_2} \right\Vert \right) \ \cdot \ \Delta_{max} \quad \ \ \text{Component Bounds}
	\\
	&\le \left( \frac{1}{G_t\bar{n}_t}\sum_{ig} S_{igt}^V + \psi_{IW,ig}^{\partial} \right) \cdot \Delta_{max} + o_p(1) \qquad\qquad\qquad\qquad \text{Since $\widetilde{\boldsymbol{\theta}}_t \in B(\boldsymbol{\theta}_{0t},\nu)$ w.p.a.1}
	\end{split}
	\label{eq:bound_deriv_psi_A}
	\end{align}
	The discrepancy $\Delta_{max}$ is $O_p\left(\lambda_t\right)$ by Assumption (vii). Conversely, the bounds on the expectation of the Sobolev-norm in \nameref{assumptxt:regularity} part (ii) and the moments in (iii) can be used to show that $\mathbb{E}[S_{igt}^V],\mathbb{E}[\psi_{IW,ig}^\partial]<\infty$ and $\frac{1}{G_t\bar{n}_t}\sum_{ig}  \left(S_{igt}^V + \psi_{IW,ig}^{\partial}\right) = O_p(1)$, by Lemmas \ref{lem:uniformboundsderivatives}
and \ref{lem:Op_bounds_averages}, respectively. By combining the two findings we conclude that the right-hand side of \eqref{eq:bound_deriv_psi_A} is $O_p\left(\lambda_t\right)$.

	The partial derivative with respect to $\boldsymbol{\tau}$ in \eqref{eq:firstorder_paramthetabeta} has a simple form
	\[
	\frac{\partial}{\partial \boldsymbol{\tau}'} \psi(Z_{igt},\widetilde{V}_{2igt},\widetilde{\boldsymbol{\tau}},\widetilde{\boldsymbol{\theta}}_t) = \left( \begin{matrix} \boldsymbol{0} \\ -I \end{matrix} \right) \equiv H_{0,\tau}
	\]
	The first set of rows is zero because the equations to compute to the weighting matrix and the second rows is the identity because $\boldsymbol{\tau}$ enters linearly in $\psi_{IW}$. In this case the derivative is constant and crucially, does not depend on the estimated parameters.
	
	Since the components that contain $\boldsymbol{\theta}$ and $\boldsymbol{\tau}$ are additively separable, the partial derivative with respect to $\boldsymbol{\theta}$ in \eqref{eq:firstorder_paramthetabeta} does not depend on $\boldsymbol{\tau}$. We write this concisely as
	\begin{equation}
	\frac{\partial}{\partial \boldsymbol{\theta}'} \psi(Z_{igt},\widetilde{V}_{2igt},\widetilde{\boldsymbol{\tau}},\widetilde{\boldsymbol{\theta}}_t) = 	\frac{\partial}{\partial \boldsymbol{\theta}'} \psi(Z_{igt},\widetilde{V}_{2igt},\widetilde{\boldsymbol{\theta}}_t) = \left( \begin{matrix} \tfrac{\partial}{\partial \boldsymbol{\theta}'}\psi_q(Z_{igt},\widetilde{V}_{2igt},\tilde{\boldsymbol{\theta}}) \\ \tfrac{\partial}{\partial \boldsymbol{\theta}'}\psi_{IW}(Z_{igt},\widetilde{V}_{2igt},\tilde{\boldsymbol{\theta}}) \end{matrix} \right)
	\label{eq:deriv_psi_theta}
	\end{equation}
	Our next task is to impose integrable bounds on \eqref{eq:deriv_psi_theta} in order to apply the uniform consistency result in \ref{lem:uniformconvergence_sampleaverages}. On one hand, our bounds on the expectations in assumptions (ii) and (iii) allow us to apply the first part of Lemma \ref{lem:uniformboundsderivatives}. The lemma shows that $\tfrac{\partial \psi_q}{\partial \boldsymbol{\theta}},\tfrac{\partial^2 \psi_q}{\partial \boldsymbol{\theta}\partial \boldsymbol{\theta}'}$, $\tfrac{\partial^2 \psi_q}{\partial v_2 \partial \boldsymbol{\theta}'}$ are uniformly bounded over $(V_2,\boldsymbol{\theta}) \in \mathcal{V}_2\times \Theta$ by integrable random variables. On the other hand, assumption (ii), (iii) and (iv) allow us to apply the second part of the lemma, which implies $\frac{\partial^2 \psi_{IW}}{\partial \boldsymbol{\theta}\partial \boldsymbol{\theta}'}, \frac{\partial^2 \psi_{IW}}{\partial v_2 \partial \boldsymbol{\theta}'}$ are uniformly bounded over $(V_2,\boldsymbol{\theta}) \in \mathcal{V}_2\times B(\boldsymbol{\theta}_{0t},\nu)$ by an integrable random variable. Consequently, we can apply Lemma \ref{lem:uniformconvergence_sampleaverages} to show that 
	\begin{equation}
	\sup_{\boldsymbol{\theta} \in B(\boldsymbol{\theta}_{0t},\nu)} \left\Vert \frac{1}{G_t\bar{n}_t}\sum_{ig} \tfrac{\partial}{\partial \boldsymbol{\theta}'} \psi(Z_{igt},\widetilde{V}_{2igt},\widetilde{\boldsymbol{\theta}}_t) - \mathbb{E}\left[\tfrac{\partial}{\partial \boldsymbol{\theta}'} \psi(Z_{igt},V_{2igt}^0,\boldsymbol{\theta})\right] \right\Vert \to^{p} 0
	\label{eq:uniformconsistency_psi_theta}
	\end{equation}
	Since $\widehat{\boldsymbol{\theta}}_t$ is consistent $\Vert \widetilde{\boldsymbol{\theta}}_t-\boldsymbol{\theta}_{0t} \Vert \le \Vert \widehat{\boldsymbol{\theta}}_t-\boldsymbol{\theta}_{0t} \Vert = o_p(1)$. Therefore, by the uniform consistency result in \eqref{eq:uniformconsistency_psi_theta},
	\[
	\frac{1}{G_t\bar{n}_t} \sum_{ig} \tfrac{\partial}{\partial \boldsymbol{\theta}'} \psi(Z_{igt},\widetilde{V}_{2igt},\widetilde{\boldsymbol{\theta}}_t) \to^{p} \mathbb{E}\left[\tfrac{\partial}{\partial \boldsymbol{\theta}'} \psi(Z_{igt},V_{2igt}^0,\boldsymbol{\theta}_{0t})\right] \equiv H_{0,\theta}
	\]
	By assumption (iv), $H_0 = [H_{0,\tau},H_{0,\theta}]$ is full rank. Therefore, solving for the parameter in \eqref{eq:firstorder_paramthetabeta} and multiplying by $\sqrt{G_t}$,
	\[
	\sqrt{G_t}	\left( \begin{matrix}  \widehat{\boldsymbol{\tau}}_t - \boldsymbol{\tau}_{0t} \\
		\widehat{\boldsymbol{\theta}}_t - \boldsymbol{\theta}_{0t} \end{matrix} \right) =
	 -\left( H_0+ o_p(1) \right)^{-1}
	 \left[ \sqrt{G_t}\left( \frac{1}{NG_t}\sum_{ig} \psi(Z_{igt},V_{2igt}^0,\boldsymbol{\theta}_{0t})  \right)
	 + O_p\left(\lambda_t \sqrt{G_t} \right) \right]
	\]
	Let $\mathbb{E}_*$ and $\mathbb{E}$ denote the sampling (equal-weighted-group) measure and the population measure respectively. By construction, $\mathbb{E}_*[\rho_{gt} \psi(Z_{igt},\theta)] = \mathbb{E}[\psi(Z_{igt},\theta)]$, where $\rho_{gt} \equiv N_{gt} / N_t$ is the relative size of each group. By Lemma \ref{lem:Op_bounds_averages}, $\tfrac{\bar{n}_t}{N_t} \to^p 1$ as $(G_t,N_t) \to \infty$. Conversely, define the within-group average $\overline{\psi}_g(\boldsymbol{Z}_g,\boldsymbol{\theta}_{0t}) \equiv \frac{1}{N_{gt}}\sum_{i=1}^{N_{gt}} \psi(Z_{igt},V_{2igt}^0,\boldsymbol{\theta}_{0t})$, where $\boldsymbol{Z}_g \equiv \{ (X_{igt},Y_{igt},(V_{igt}^0))\}_{i=1}^{N_{gt}}$ is a matrix of individual covariates. By some algebraic manipulations
	\[
	\frac{\sqrt{G_t}}{\bar{n}_tG_t}\sum_{ig} \psi(Z_{igt},V_{2igt}^0,\boldsymbol{\theta}_{0t})   = \left( \frac{N_t}{\bar{n}_t} \right)\left(\frac{1}{\sqrt{G_t}}\sum_{g=1}^{G_t} \rho_{gt} \overline{\psi}_{g}(\boldsymbol{Z}_g,\boldsymbol{\theta}_{0t}) \right)
	\]
	Our final task is to apply the central limit theorem. First, we check that the influence function is mean zero. By distributing the expectation
	\[
	\mathbb{E}_*[\rho_{gt} \overline{\psi}_g(\boldsymbol{Z}_g,\boldsymbol{\tau}_{0t},\boldsymbol{\theta}_{0t})] = \mathbb{E}_*[\rho_{gt} \psi(Z_{igt},V_{2igt}^0,\boldsymbol{\tau}_{0t},\boldsymbol{\theta}_{0t})] = \mathbb{E}[\psi(Z_{igt},V_{2igt}^0,\boldsymbol{\tau}_{0t},\boldsymbol{\theta}_{0t})]
	\]
	Recall that $\psi = [s,\psi_{IW}]$. The mean of $s$ is equal to zero at the true value when the weighting matrix is properly specified. Similarly, $\psi_{IW}$ is equal to zero by Theorem \ref{lem:randomcoef_identification}.
	
	Finally, by assumption (v), $\mathbb{E}_*[\rho_{gt}^2 \overline{\psi}_{g}\overline{\psi}_{g}'] = \mathbb{E}[\rho_{gt} \overline{\psi}_{g}\overline{\psi}_{g}'] \equiv \Omega_{0t}$ is a positive-definite matrix. By the Lindenber-Feller central limit theorem, as $(G_t,N_t) \to \infty$,
	\[
	\Omega_{0t}^{-1}\left(\frac{1}{\sqrt{G_t}}\sum_{g=1}^{G_t} \rho_{gt} \overline{\psi}_{g}(\boldsymbol{Z}_g,\boldsymbol{\theta}_{0t})\right) \to^d \mathcal{N}(\boldsymbol{0},I).
	\]
	Combining the results we prove that the estimator converges to a normal distribution plus a bias term, 
	\[
	\sqrt{G_t}\Sigma_t^{-1/2}\left( \begin{matrix} \widehat{\boldsymbol{\theta}}_t - \boldsymbol{\theta}_{0t} \\
	\widehat{\boldsymbol{\tau}}_t - \boldsymbol{\tau}_{0t} 
	\end{matrix} \right)  \to^d \mathcal{N}(\boldsymbol{0},I) + O_p\left(\lambda_t \sqrt{G_t} \right).
	\]
	where $\Sigma_t = H_{0t}^{-1}\Omega_{0t}H_{0t}^{-1}$. Under assumption (viii) the second term is $o_p(1)$,
	
\end{proof}

\newpage

\subsection{Proofs Extensions and Experiments}

\begin{proof}[Proof of Theorem \ref{thm:att_atu_identification} (\nameref{thm:att_atu_identification})]
	For the APT, our first objective is to rewrite the inner term of the expectation in terms of the localized effect $\boldsymbol{\tau}(v)$, instead of $(\mathbf{Q}_{xx},X_{ig},Y_{ig})$. To this end we compute the conditional expectation given a particular value of the control variable $V_{ig}$.  In this case $p_d(V_{ig})$ is a constant given $V_{ig}$, so we can directly apply part (i) of Lemma \ref{lem:randomcoef_identification},
	\begin{align}
		\begin{split}
			\mathbb{E}[p_d(V_{ig}) \times \mathbf{Q}_{xx}(V_{ig})^{-1}X_{ig}Y_{ig} \mid V_{ig} = v]
			&= p_d(V_{ig}) \times \mathbf{Q}_{xx}(v)^{-1}\mathbb{E}[X_{ig}Y_{ig} \mid V_{ig}] \\
			&= p_d(V_{ig}) \times \mathbf{Q}_{xx}(v)^{-1}\mathbf{Q}_{xy}(V_{ig}) \\
			&= 	p_d(v) \times \boldsymbol{\tau}(v)
		\end{split}
		\label{eq:att_condexpect_initial}
	\end{align}
	The second task is to express \eqref{eq:att_condexpect_initial} in terms of of the primitives $(D_{ig},\tau_{ig})$. By definition $p_d(v) = \mathbb{P}(D_{ig} = 1 \mid V_{ig} = v)$. Since $V_{ig}$ is a control variable for $D_{ig}$, it follows that
	\[
	\boldsymbol{\tau} \equiv \mathbb{E}[\tau_{ig} \mid V_{ig} = v] = \mathbb{E}[\tau_{ig} \mid V_{ig} = v, D_{ig} = 1]
	\]
	Then by the law of iterated expectations $p_d(v) \times \boldsymbol{\tau}(v) $ equals
	\begin{align}
		\begin{split}
			\mathbb{P}(D_{ig} = 1 \mid V_{ig} = v) \times \mathbb{E}[\tau_{ig} \mid V_{ig} = v, D_{ig} = 1] = \mathbb{E}[ D_{ig} \tau_{ig} \mid V_{ig} = v]
		\end{split}
		\label{eq:att_condexpect_final}
	\end{align}
	Therefore \eqref{eq:att_condexpect_final} produces a simplified expression for the conditional expectation in \eqref{eq:att_condexpect_initial}. Applying the law of iterated expectations and substituting the expression in  \eqref{eq:att_condexpect_final}, 
	\[
	\mathbb{E}\left[p_d(V_{ig}) \times \mathbf{Q}_{xx}(V_{ig})^{-1}X_{ig}Y_{ig} \right] = \mathbb{E}[\mathbb{E}[ D_{ig} \tau_{ig} \mid V_{ig}]] = \mathbb{E}[D_{ig}\tau_{ig}]
	\]
	By Bayes's rule and the fact that $D_{ig}$ is binary,
	\begin{equation}
	\frac{\mathbb{E}[D_{ig} \tau_{ig}]}{\mathbb{P}(D_{ig} = 1)} = \mathbb{E}[\tau_{ig} \mid D_{ig} = 1] = \boldsymbol{\tau}_{APT}
	\label{eq:att_bayes}
	\end{equation}
	The unconditional effect, the APT and APU effects are mutually constrained by the law of iterated expectations, which implies that $
	\boldsymbol{\tau} = \mathbb{P}(D_{ig} = 1) \boldsymbol{\tau}_{APT} + \mathbb{P}(D_{ig} = 0) \boldsymbol{\tau}_{APU}$. Lemma \ref{lem:randomcoef_identification} implies that  $\boldsymbol{\tau} = \mathbb{E}[\mathbf{Q}_{xx}(V_{ig})^{-1}X_{ig}Y_{ig}]$. Therefore, we can solve for the APU effect by substituting the expressions for $(\boldsymbol{\tau},\boldsymbol{\tau}_{APT})$ and solving for $\boldsymbol{\tau}_{APU}$,
	\[
	\boldsymbol{\tau}_{APU} = \frac{1}{1-\mathbb{E}[D_{ig}]}\times \mathbb{E}\left[ (1-p_d(V_{ig})) \times \mathbf{Q}_{xx}(V_{ig})^{-1}X_{ig}Y_{ig} \right]
	\]		
\end{proof}

\newpage


\subsection{Supporting Lemmas}

\begin{lem}[Properties of Conditional Independence]
	\label{lem:axioms}
	Let $X,Y,Z,W$ be random vectors defined on a common probability space, and let $h$ be a measurable function.
	Then:
	\begin{enumerate}[(i)]
		\item (Symmetry): $X \indep Y|Z \implies Y \indep X|Z$.
		\item (Redundancy): $X \indep Y|Y$.
		\item (Decomposition): $X \indep Y|Z$ and $W = h(Y) \implies X \indep W|Z$.
		\item (Weak Union): $X\indep Y|Z$ and $W = h(Y) \implies X \indep Y|(W,Z)$.
		\item (Contraction): $X\indep Y|Z$ and $X\indep W|(Y,Z) \implies X\indep (Y,W)|Z$.
	\end{enumerate}
\end{lem}

\begin{proof}
	\cite{Constantinou2017}
\end{proof}

\begin{lem}[Combining Events]
	Let $E,E^*,U,U^*,\Psi$ be random variables on a common probability space. Suppose that (i) $E \ \indep \ E^* \mid \Psi$, (ii) $(E,E^*) \ \indep \ U^* \mid \Psi$ and (iii) $U \ \indep \ (U^*,E,E^*) \mid \Psi$. Then
	\[
	(E,U) \ \indep \ (E^*,U^*) \mid \Psi
	\]
	\label{lem:combiningevents}
\end{lem}

\begin{lem}[Egocentric Likelihood]
	Suppose that $D_{ig}$ is $(C_{ig},\Psi_g^*,\eta_{ig})-$measurable and $A_{ijg}$ is $(C_{ig},C_{jg},\Psi_g^*,U_{ijg})-$measurable. If \nameref{assumptxt:randomsampling} and \nameref{assumptxt:dyadicnetwork} hold, then for $V_{ig} \equiv (C_{ig},\Psi_g^*)$
	\[
	\mathbb{P}(D_g,A_{ig} \mid V_{ig}) = \mathbb{P}(D_{ig} \mid V_{ig}) \prod_{j\ne i}^{N_g} \mathbb{P}(D_{jg},A_{ijg} \mid V_{ig})
	\]
	\label{lem:egocentric_factorization}
\end{lem}


\begin{lem}[Bounds Quotients]
	\label{lem:boundsquotients} Let $a,b$ be non-zero scalars and suppose that $\Vert b \Vert \ge \underline{b} > 0$. Then
	\[
	\Vert a^{-1} - b^{-2} \Vert \le \frac{\underline{b}^{-2} \Vert b-a \Vert}{1-\underline{b}^{-1}\Vert b - a \Vert}.
	\]
\end{lem}

\begin{lem}[Derivative of Inverse Matrix]
	Let $v \in \mathbb{R}$ and suppose that $Q(v)$ is differentiable and full rank in an open set around $v_0$. Then $\tfrac{\partial}{\partial v} Q^{-1}(v_0) = - Q^{-1}(v_0) \tfrac{\partial Q(v_0)}{\partial v} Q^{-1}(v_0)$.
	\label{lem:derivativesinversematrix}
\end{lem}

\begin{lem}[Uniform Bounds Criterion Derivatives]
	
Let $\lambda_{min}(v_1,v_2,\boldsymbol{\theta})$ denote the smallest eigenvalue of $\mathbf{Q}_{xx}((v_1,v_2),\boldsymbol{\theta})$ and let $B(\boldsymbol{\theta}_0,\delta)$ be a ball or radius $\delta >0 $ around $\boldsymbol{\theta}_0$. Let  $\underline{\lambda}(V_{1igt}^{0},\boldsymbol{\theta}_{0t},\delta) \equiv \inf_{\theta \in B(\boldsymbol{\theta}_{0t},\delta)} \inf_{v_2 \in \mathcal{V}_2}\lambda_{min}(V_{1igt}^{0},v_2,\boldsymbol{\theta})$ be a lower bound on the eigenvalues of $\mathbf{Q}_{xx}$ for parameters in that set. Furthermore, define
\begin{align}
\begin{split}
\mathcal{R}_{ig}^A &\equiv \sup_{\theta \in \Theta} \ \sup_{v_2 \in \mathcal{V}_2} \ \left\Vert \tfrac{\partial}{\partial v_2}\mathcal{R}(Z_{ig},v_2,\theta) \right\Vert 
\end{split}
\label{eq:envelope_RA} \\
\begin{split}
S_{ig} &\equiv \sup_{\theta \in \Theta} \ \sup_{v_2 \in \mathcal{V}_2} \ \left\Vert s(Z_{ig},v_2,\theta)  \right\Vert 
\end{split}
\label{eq:envelope_S}
\\
\begin{split}
S_{ig}^{A} &\equiv \sup_{\theta \in \Theta} \ \sup_{v_2 \in \mathcal{V}_2} \ \left\Vert \tfrac{\partial}{\partial v_2}s(Z_{ig},v_2,\theta)  \right\Vert
\end{split}
\label{eq:envelope_SA}
\\
\begin{split}
S_{ig}^{A\theta} &\equiv \sup_{\theta \in \Theta} \ \sup_{v_2 \in \mathcal{V}_2} \ \left\Vert \tfrac{\partial^2}{\partial v_2\partial \theta}s(Z_{ig},v_2,\theta) \right\Vert
\end{split}
\label{eq:envelope_SATheta}
\\
\begin{split}
\psi_{IW,ig}^{\partial} &\equiv \sup_{\theta \in B(\boldsymbol{\theta}_0,\nu)} \ \sup_{v_2 \in \mathcal{V}_2} \ \sup_{0 \le \alpha_1 + \alpha_2 \le 2} \left\Vert \tfrac{\partial^{\alpha_1 + \alpha_2}}{\partial v_2^{\alpha_1}\partial\boldsymbol{\theta}^{\alpha_2}}\psi_{IW}(Z_{ig},v_2,\boldsymbol{\tau},\boldsymbol{\theta})  \right\Vert
\end{split}
\label{eq:envelope_psitheta}
\end{align}
then the following statements hold
\begin{enumerate}[(i)]
	\item If $\mathbb{E}[\Vert X_{ig} \Vert^4] < \infty$ and $\mathbb{E}[\sup_{\theta \in \Theta} \sup_{v_2 \in \mathcal{V}_2} \ (Q^{\partial}_{xx}(V_{1ig},v_2,\theta))^2] $ is bounded, then $\mathbb{E}[\mathcal{R}_{ig}^A]$ $\mathbb{E}[S_{ig}]$, $\mathbb{E}[S_{ig}^A]$ and $\mathbb{E}[S_{ig}^{A\theta}]$ are bounded.
	\item Suppose that in addition $\mathbb{E}[\Vert Y_{ig}\Vert^2] < \infty$, $\mathbb{E}[\sup_{\theta \in \Theta} \sup_{v_2 \in \mathcal{V}_2} \ (Q^{\partial}_{xx}(V_{1ig},v_2,\theta))^4]< \infty$ and $\underline{\lambda} > 0$ almost surely. Then $\mathbb{E}[\psi_{IW,ig}^{\partial}]$ is also bounded.
\end{enumerate}
\label{lem:uniformboundsderivatives}
\end{lem}

\begin{lem}[Stochastically Bounded Averages]
	Let $X_{igt}$ be a sequence of random variable such that $\mathbb{E}[\Vert X_{igt} \Vert] < \infty$, $\bar{X}_t = \frac{1}{GN}\sum_{ig} X_{igt}$ the sample average and $N_t = \mathbb{E}[N_{gt}]$ be the expected group size, and $(G_t,N_t) \to \infty$ as $t \to \infty$. Suppose that the groups are randomly sampled with equal weight from a superpopulation and that \nameref{assumptxt:boundedrelative_n} holds, then $\mathbb{E}_*[\tfrac{\bar{n}_t}{N_t}\bar{X}_t] = \mathbb{E}\left[ X_{igt}\right]$ and $\bar{X}_t = O_p(1)$ as $(G,N_*) \to \infty$, where $\mathbb{E}_*$ is the sampling (equal-group-weight) measure and $\mathbb{E}$ is the population measure. Furthermore, if \nameref{assumptxt:randomsampling} holds then $\bar{X}_t \to^{p} \mathbb{E}_t[X_{ig}]$ and $\frac{\bar{n}}{N_t} \to^{p} 1$.
	\label{lem:Op_bounds_averages}
\end{lem}

\begin{lem}[Uniform Consistency with Generated Regressors]
	Let $f$ be a measurable function of $(z,v_2,\boldsymbol{\tau},\boldsymbol{\theta})$ that is continuously differentially with respect to $(v_2,\boldsymbol{\tau},\boldsymbol{\theta})$. Suppose that
	\begin{enumerate}[(i)] 
		\item \quad $\max_{igt} \Vert V_{2igt} - V_{2igt}^0\Vert \to_p 0$
		\item \quad $\mathbb{E}\left[\sup_{(\boldsymbol{\tau},\boldsymbol{\theta},v_2)\in \mathcal{T} \times \Theta \times \mathcal{V}_2} \left\Vert f(Z_{igt},v_2,\boldsymbol{\tau},\boldsymbol{\theta})\right\Vert \right] < \infty$.
		\item \quad $\mathbb{E}\left[\sup_{(\boldsymbol{\tau},\boldsymbol{\theta},v_2)\in \mathcal{T} \times \Theta \times \mathcal{V}_2} \left\Vert \tfrac{\partial}{\partial (v_2,\boldsymbol{\tau},\boldsymbol{\theta})}  f(Z_{igt},v_2,\boldsymbol{\tau},\boldsymbol{\theta})\right\Vert \right] < \infty$.
	\end{enumerate}
	If \nameref{assumptxt:randomsampling} holds, and $(G_t,N_t) \to \infty$ as $t \to \infty$, then 
	\[
	\sup_{\boldsymbol{\tau} \in \mathcal{T}} \sup_{\boldsymbol{\theta} \in \Theta} \left\Vert \frac{1}{G_t\bar{n}_t}\sum_{ig} f(Z_{igt},V_{2igt},\boldsymbol{\tau},\boldsymbol{\theta}) - \mathbb{E}[f(Z_{igt},V_{2igt}^0,\boldsymbol{\tau},\boldsymbol{\theta})]. \right\Vert \to^p 0
	\]
	\label{lem:uniformconvergence_sampleaverages}
\end{lem}


\subsection{Proof Supporting Lemmas}


\begin{proof}[Proof of Lemma \ref{lem:combiningevents} (\nameref{lem:combiningevents})]
	By property (ii), weak union and decomposition
	\begin{equation}
	\label{eq:groupingevents_1}
	(ii) \quad \implies \quad (E,E^*) \ \indep \ U^* \mid E^*,\Psi \quad \implies \quad E \ \indep \ U^* \mid E^*,\Psi
	\end{equation}
	By property (i), \eqref{eq:groupingevents_1} and contraction, $E \ \indep \ (E^*,U^*) \mid \Psi$.
	
	Similarly, by property (iii), weak union and decomposition.
	\begin{equation}
	\label{eq:groupingevents_2}
	(iii) \quad \implies \quad U \ \indep \ (U^*,E,E^*) \mid E,\Psi \quad \implies \quad U \ \indep \ (E^*,U^*) \mid E,\Psi
	\end{equation}
	Combining the two results via the contraction property, $(E,U) \ \indep (E^*,U^*) \mid \Psi$.
	
\end{proof}

\begin{proof}[Proof of Lemma \ref{lem:egocentric_factorization} (\nameref{lem:egocentric_factorization})]
	Let $V_{ig} \equiv (C_{ig},\Psi_g^*)$. By Bayes' rule:
	\begin{equation}
	\label{eq:bayesegocentric}
	\mathbb{P}(D_g,A_{ig} \mid V_{ig}) = \prod_{j=1}^n \mathbb{P}(D_{jg},A_{ijg} \mid \{D_{kg},A_{ikg}\}_{k=1}^{j-1},V_{ig})
	\end{equation}
	We can factor the joint probability in any order, so I set $i = 1$ without loss of generality. By definition $G_{iig} = 0$ (no self-loops in the network), so $\mathbb{P}(G_{iig} = 0 \mid V_{ig}) = 1$ and we can denote the probability as $\mathbb{P}(D_{ig} \mid A_{ig},V_{ig}) = \mathbb{P}(D_{ig} \mid V_{ig})$ without loss of generality.
	
	For $j > 1$, define the random variables $E \equiv (\eta_{jg},C_{jg})$ and $E^* \equiv \{(\eta_{kg},C_{kg})\}_{k=1}^{j}$, which denote the personal covariates of $j$ and a vector of the covariates of agents 1 through $(j-1)$, respectively. Similarly, let $U \equiv U_{ijg}$ and $U^* \equiv \{ U_{ikg} \}_{k=1}^{(j-1)}$, denote the respective link formation shocks. \nameref{assumptxt:randomsampling}.(i) allows us to ignore covariates across groups. \nameref{assumptxt:randomsampling}.(ii) states that the covariates of different agents  are conditionally independent,  which implies $E \ \indep \ E^* \mid \Psi_g$. \nameref{assumptxt:dyadicnetwork} says that the personal covariates are conditionally independent of the link shocks, which implies that $(E,E^*) \ \indep \ U^* \mid \Psi_g$. Furthermore, the links are also mutually conditionally independent of each other, which means that $U \ \indep \ (U^*,E,E^*) \mid \Psi_g$.
	
	Consequently $(E,E^*,U,U^*,\Psi_g)$ meet the conditions of Lemma \ref{lem:combiningevents} and 
	\begin{equation}
	(\eta_{jg},C_{jg},U_{ijg}) \ \indep\  \{\eta_{kg},C_{kg}\}_{k=1}^{j-1},\{U_{ijg}\}_{k=2}^{j-1} \mid \Psi_g
	\label{eq:decomp_egocentriclikelihood_1}
	\end{equation}
	The right-hand side of \eqref{eq:decomp_egocentriclikelihood_1} contains enough information to compute $D_{kg} = h(C_{kg},\Psi_g^*,\eta_{kg})$ and $A_{ijg} =  L(C_{ig},C_{jg},\Psi_g^*,U_{ijg})$. Therefore, we can use the decomposition property to show that
	\begin{equation}
	\label{eq:decomp_egocentriclikelihood_2}
	(U_{ijg},\eta_{jg},C_{jg}) \ \indep\  \{D_{jg},A_{ijg}\}_{k=1}^{j-1},C_{ig}\mid \Psi_g^* 
	\end{equation}
	By combining  \eqref{eq:decomp_egocentriclikelihood_2}, weak union and decomposition,
	\begin{equation}
	\label{eq:decomp_egocentriclikelihood_3}
	(U_{ijg},\eta_{jg},C_{jg}) \ \indep\  \{D_{jg},A_{ikg}\}_{k=1}^{j-1} \mid C_{ig},\Psi_g^*
	\end{equation}
	Finally $D_{jg}$ is $(C_{jg},\Psi_g^*,\eta_{jg})$-measurable and $A_{ijg}$ is $(C_{ig},C_{jg},\Psi_g^*,U_{ijg})$-measurable. We can use the redundancy property to incorporate the variables in the conditioning set and then apply the decomposition property to show that 
	\[
	(D_{jg},A_{ikg}) \ \indep\  \{D_{kg},A_{ikg}\}_{k=1}^{j-1} \mid C_{ig},\Psi_g^*
	\] 
	By applying this argument recursively, we can show that potential link and participation decisions are conditionally independent. By \eqref{eq:bayesegocentric}
	\[
	\mathbb{P}(D_g,A_{ig} \mid V_{ig}) = \mathbb{P}(D_{ig} \mid V_{ig}) \prod_{j\ne i}^{N_g} \mathbb{P}(D_{jg},A_{ijg} \mid V_{ig})
	\]
\end{proof}

\begin{proof}[Proof of Lemma \ref{lem:boundsquotients} (\nameref{lem:boundsquotients})]
	By finding a common denominator, $a^{-1} - b^{-1} = b^{-1}(b-a)a^{-1}$. By the triangle inequality
	\begin{align*}
		\Vert a^{-1} - b^{-1} \Vert &= \Vert b^{-1} \Vert \ \Vert b - a \Vert \ \Vert a^{-1} \Vert \\
		&\le \Vert b^{-1} \Vert \ \Vert b - a \Vert \ (\Vert b^{-1} \Vert + \Vert a^{-1} - b^{-1}\Vert) \\
		&\le \underline{b}^{-1} \ \Vert b - a \Vert \ (\underline{b}^{-1} + \Vert a^{-1} - b^{-1}\Vert).
	\end{align*}
	Solving for $\Vert a^{-1} - b^{-1} \Vert$,
	\[
	\Vert a^{-1} - b^{-1} \Vert \le \frac{\underline{b}^{-2} \Vert b-a \Vert}{1-\underline{b}^{-1}\Vert b - a \Vert}.
	\]
\end{proof}

\begin{proof}[Proof of Lemma \ref{lem:derivativesinversematrix} (\nameref{lem:derivativesinversematrix})]
	Let $M(v) \equiv Q^{-1}(v)$ and define $F(v) \equiv Q(v)M(v) - I$. By construction $F(v) \equiv \boldsymbol{0}$ uniformly for $v$ in open set around $v_0$. Let $F_{i\ell}$ denote the entry in the $i^{th}$ and the $\ell^{th}$ column of $F$, which can be decomposed as
	\[
	F_{i\ell}(v) = \sum_{k\ell} Q_{ij}M_{k \ell} - a_{i\ell} = 0
	\] 
	where $a_{i\ell}$ are the entries of the identity matrix $I$. We can differentiate each component by the scalar $v$. By the product rule,
	\[
	\frac{\partial F_{i\ell}(v)}{\partial v} = \sum_{k\ell} \frac{\partial Q_{ij}}{\partial v}M_{k \ell} + Q_{ij}\frac{\partial M_{k\ell}}{\partial v} = 0
	\]
	Define $\frac{\partial F}{\partial v}$ denote the matrices with entries $\frac{\partial F}{\partial v}$. Define $\frac{\partial Q}{\partial v},\frac{\partial M}{\partial v}$ analogously. Then
	\[
	\frac{\partial F(v_0)}{\partial v} = \tfrac{\partial Q(v_0)}{\partial v} M(v_0) + Q(v_0) \frac{\partial M(v_0)}{\partial v} = 0 
	\]
	Solving the equation, $\frac{\partial M(v_0)}{\partial v} = -  Q(v_0)^{-1} \tfrac{\partial Q(v_0)}{\partial v} M(v_0)$ and substituting the definition of $M$,
	\[
	\tfrac{\partial}{\partial v} Q^{-1}(v_0) = - Q^{-1}(v_0) \tfrac{\partial Q(v_0)}{\partial v} Q^{-1}(v_0).
	\]
\end{proof}

\begin{proof}[Proof of Lemma \ref{lem:uniformboundsderivatives} (\nameref{lem:uniformboundsderivatives})]
	The first task is to express the derivatives of $\mathcal{R}(\cdot)$ and $s(\cdot)$ in terms of $X_{ig}$ and the weighting matrix $\mathbf{Q}_{xx}(\cdot)$. It will be convenient to work with the vectorized version of the weighting matrix, which I denote by $\mathbf{q}(Z_{ig},v_2,\boldsymbol{\theta})$. Similarly, define $\mathbf{x}_{ig} = vec(X_{ig}X_{ig}')$. In matrix form the criterion can be expressed as
	\[
	\mathcal{R}(Z_{ig},v_2,\boldsymbol{\theta}) = (\mathbf{x}-\mathbf{q})'(\mathbf{x}-\mathbf{q})
	\]
	Since the score function is defined as the jacobian of $\mathcal{R}$, then $s(Z_{ig},v_2,\boldsymbol{\theta}) = -2(\mathbf{x}-\mathbf{q})'\tfrac{\partial \mathbf{q}}{\partial \boldsymbol{\theta}'}$. We can compute the following derivatives by applying the chain rule. Let $(\mathbf{x}_k,\mathbf{q}_k)$ denote the $k^{th}$ rows of $(\mathbf{x},\mathbf{q})$, respectively. Then
	\begin{align*}
	\tfrac{\partial}{\partial v_2}\mathcal{R}(Z_{ig},v_2,\boldsymbol{\theta})  &=  -2(\mathbf{x}-\mathbf{q})'\tfrac{\partial \mathbf{q}}{\partial v_2} \\
	\tfrac{\partial}{\partial \boldsymbol{\theta}}s(Z_{ig},v_2,\boldsymbol{\theta}) &= 2 \tfrac{\partial \mathbf{q}}{\partial \boldsymbol{\theta}}\tfrac{\partial \mathbf{q}}{\partial \boldsymbol{\theta}'} - 2 \sum_{k} (\mathbf{x}_k-\mathbf{q}_k) \tfrac{\partial^2 \mathbf{q}_k}{\partial \boldsymbol{\theta}\partial \boldsymbol{\theta}'} \\
	\tfrac{\partial}{\partial v_2}s(Z_{ig},v_2,\boldsymbol{\theta})
	&= 2 \times \tfrac{\partial \mathbf{q}'}{\partial v_2}\tfrac{\partial \mathbf{q}}{\partial \boldsymbol{\theta}'} - 2 \sum_{k} (\mathbf{x}_k-\mathbf{q}_k) \tfrac{\partial^2 \mathbf{q}_k}{\partial v_2 \partial \boldsymbol{\theta}'} \\
	\tfrac{\partial}{\partial^2 A\partial \boldsymbol{\theta}}s(Z_{ig},v_2,\boldsymbol{\theta}) &= 2 \times \left( \tfrac{\partial^2 \mathbf{q}'}{\partial v_2^2}\tfrac{\partial \mathbf{q}}{\partial \boldsymbol{\theta}'} + \tfrac{\partial \mathbf{q}'}{\partial v_2}\tfrac{\partial^2 \mathbf{q}}{\partial v_2 \partial \boldsymbol{\theta}'}\right)
	- 2 \sum_{k} \left[ -\tfrac{\partial \mathbf{q}_k}{\partial v_2}\tfrac{\partial^2 \mathbf{q}_k}{\partial v_2 \partial \boldsymbol{\theta}'} +  (\mathbf{x}_k-\mathbf{q}_k) \tfrac{\partial^2 \mathbf{q}_k}{\partial v_2^2 \partial \boldsymbol{\theta}'} \right] 
	\end{align*}
	Let $\mathbf{Q}_{xx}^{\partial}(Z_{ig},v_2,\boldsymbol{\theta})$ denote the Sobolev norm, as defined in \eqref{eq:sobolevenorm_Qxx}, which is a bound on the derivatives of order $\{0,1,2,3\}$. Similarly, let $\Vert \mathbf{x} \Vert $ denote the Euclidean norm of $\mathbf{x}$. It is useful to use the fact that $\sum_{k} \Vert \mathbf{x}_k \Vert \le \kappa \Vert \mathbf{x} \Vert $, for some universal constant $\kappa$ that only depends on the dimension. We denote this inequality as $\sum_{k} \Vert \mathbf{x}_k \Vert \lesssim \Vert \mathbf{x} \Vert$.
	
	By the triangle inequality,
	\begin{align}
	\begin{split}
	\left\Vert \tfrac{\partial}{\partial v_2}\mathcal{R}(Z_{ig},v_2,\boldsymbol{\theta}) \right\Vert
	&\le  2 (\Vert \mathbf{x}_{ig} \Vert + \Vert \mathbf{q} \Vert)\left\Vert \tfrac{\partial \mathbf{q}}{\partial v_2} \right\Vert \le  2 \Vert \mathbf{x}_{ig} \Vert \  \mathbf{Q}_{xx}^{\partial}(Z_{ig},v_2,\boldsymbol{\theta}) + 2 \mathbf{Q}_{xx}^{\partial}(Z_{ig},v_2,\boldsymbol{\theta})^2 \\
	\left\Vert \tfrac{\partial}{\partial \boldsymbol{\theta}}s(Z_{ig},v_2,\boldsymbol{\theta})\right\Vert 
	&\le 2 \left\Vert \tfrac{\partial \mathbf{q}}{\partial \boldsymbol{\theta}}\right\Vert \left\Vert \tfrac{\partial \mathbf{q}}{\partial \boldsymbol{\theta}'}\right\Vert + 2 \sum_{k} (\Vert \mathbf{x}_k\Vert-\Vert \mathbf{q}_k \Vert) \left\Vert \tfrac{\partial^2 \mathbf{q}_k}{\partial \boldsymbol{\theta}\partial \boldsymbol{\theta}'} \right\Vert \\
	&\lesssim 4 \mathbf{Q}_{xx}^{\partial}(Z_{ig},v_2,\boldsymbol{\theta})^2 + 2 \Vert \mathbf{x}_{ig} \Vert \mathbf{Q}_{xx}^{\partial}(Z_{ig},v_2,\boldsymbol{\theta}) \\
	\left\Vert \tfrac{\partial}{\partial v_2}s(Z_{ig},v_2,\boldsymbol{\theta}) \right\Vert 
	&\le 2 \left\Vert \tfrac{\partial \mathbf{q}}{\partial v_2}\right\Vert\left\Vert \tfrac{\partial \mathbf{q}}{\partial \boldsymbol{\theta}'}\right\Vert + 2 \sum_{k} (\Vert \mathbf{x}_k \Vert-\Vert \mathbf{q}_k \Vert ) \left\Vert \tfrac{\partial^2 \mathbf{q}_k}{\partial v_2 \partial \boldsymbol{\theta}'} \right\Vert \\
	& \lesssim 4 \mathbf{Q}_{xx}^{\partial}(Z_{ig},v_2,\boldsymbol{\theta})^2 + 2 \Vert \mathbf{x}_{ig} \Vert \mathbf{Q}_{xx}^{\partial}(Z_{ig},v_2,\boldsymbol{\theta})
	\end{split}
	\label{eq:bounds_derivatives1}
	\end{align}
	At each step we bound the derivatives by the Sobolev-norm and Euclidean norms, respectively. By using a similar procedure we can show that
	\begin{equation}
	\left\Vert \tfrac{\partial}{\partial^2 A\partial \boldsymbol{\theta}}s(Z_{ig},v_2,\boldsymbol{\theta}) \right\Vert \lesssim 8 \mathbf{Q}_{xx}^{\partial}(V_{1ig}^0,v_2,\boldsymbol{\theta})^2 + 2 \Vert \mathbf{x}_{ig} \Vert \mathbf{Q}_{xx}^{\partial}(V_{1ig}^0,v_2,\boldsymbol{\theta})
	\label{eq:bounds_derivatives2}
	\end{equation}
	Our next task is to derive a uniform bounds for the expectations of the derivatives. All of the derivatives in \eqref{eq:bounds_derivatives1} and \eqref{eq:bounds_derivatives2} are bounded uniformly by combinations of $\Vert \mathbf{x}_{ig} \Vert$ and $\mathbf{Q}_{xx}^{\partial}(\cdot)$. By assumption $\mathbb{E}[\sup_{\theta \in \Theta} \sup_{v_2 \in \mathcal{V}_2} \ (Q^{\partial}_{xx}(V_{1ig}^0,v_2,\theta))^2] < \infty$, which allows us to bound some of the terms directly. To bound the rest of the terms we use the Cauchy-Schwartz inequality,
	\[
	\mathbb{E}\left[\sup_{\theta \in \Theta} \sup_{v_2 \in \mathcal{V}_2} \mathbf{Q}_{xx}^{\partial}(V_{1ig}^0,v_2,\boldsymbol{\theta})\right] \le \sqrt{\mathbb{E}\left[\sup_{\theta \in \Theta} \sup_{v_2 \in \mathcal{V}_2} \ (Q^{\partial}_{xx}(V_{1ig}^0,v_2,\theta))^2\right]} < \infty.
	\]
	\[
	\mathbb{E}\left[\sup_{\theta \in \Theta} \sup_{v_2 \in \mathcal{V}_2}\Vert \mathbf{x}_{ig} \Vert \mathbf{Q}_{xx}^{\partial}(V_{1ig}^0,v_2,\boldsymbol{\theta})\right] \le \sqrt{\mathbb{E}\left[ \Vert \mathbf{x}_{ig} \Vert^2 \right] \times \mathbb{E}\left[\sup_{\theta \in \Theta} \sup_{v_2 \in \mathcal{V}_2} \ (Q^{\partial}_{xx}(V_{1ig}^0,v_2,\theta))^2\right]} < \infty.
	\]
	Recall that $\mathbf{x}_{ig} = vec(X_{ig}X_{ig}')$ (the product of $X_{ig}$) which means that $\mathbb{E}[\Vert \mathbf{x}_{ig} \Vert^2] \lesssim \mathbb{E}[\Vert X_{ig} \Vert^4]$, which is finite by assumption.
	
	Define $(\mathcal{R}_{ig}^V,S_{ig},S_{ig}^V,S_{ig}^{V\theta})$ as in the statement of the Lemma. By \eqref{eq:bounds_derivatives1} and \eqref{eq:bounds_derivatives2}
	\[
	\mathcal{R}_{ig}^V,S_{ig},S_{ig}^V,S_{ig}^{V\theta}
	\lesssim 8 \sup_{\theta \in \Theta} \sup_{v_2 \in \mathcal{V}_2} \mathbf{Q}_{xx}^{\partial}(V_{1ig}^0,v_2,\boldsymbol{\theta})^2 + 2 \  \sup_{\theta \in \Theta} \sup_{v_2 \in \mathcal{V}_2} (1+\Vert \mathbf{x}_{ig} \Vert) \mathbf{Q}_{xx}^{\partial}(V_{1ig}^0,v_2,\boldsymbol{\theta}).
	\]
	The expectations of the right-hand side is bounded and therefore $\mathbb{E}[\mathcal{R}_{ig}^V], \mathbb{E}[S_{ig}],\mathbb{E}[S_{ig}^V],\mathbb{E}[S_{ig}^{V\theta}]$ are finite.
	
	Now we turn our attention to the derivatives of the influence function
	\[
	\psi_{IW}(Z_{igt},v_2,\boldsymbol{\beta},\boldsymbol{\theta}) = \mathbf{Q}_{xx}(Z_{igt},v_2,\boldsymbol{\theta})^{-1}X_{ig}Y_{ig}
	\]
	By Lemma \ref{lem:derivativesinversematrix} we can compute the derivatives of the inverse.
	\begin{align*}
		\tfrac{\partial \psi}{\partial v_2} &= -\mathbf{Q}_{xx}^{-1}\tfrac{\partial \mathbf{Q}_{xx}}{\partial v_2}\mathbf{Q}_{xx}^{-1} X_{ig}Y_{ig} \qquad	\tfrac{\partial \psi}{\partial \boldsymbol{\theta}_m} = -\mathbf{Q}_{xx}^{-1}\tfrac{\partial \mathbf{Q}_{xx}}{\partial v_2}\mathbf{Q}_{xx}^{-1} X_{ig}Y_{ig}
	\end{align*}
	Similarly, by applying the product rule and grouping terms
	\begin{align*}
		\tfrac{\partial^2 \psi}{\partial v_2 \partial \boldsymbol{\theta}} &= [ 2 \ \mathbf{Q}_{xx}^{-1}\tfrac{\partial \mathbf{Q}_{xx}}{\partial v_2}\mathbf{Q}_{xx}^{-1}\tfrac{\partial \mathbf{Q}_{xx}}{\partial \boldsymbol{\theta}_j}\mathbf{Q}_{xx}^{-1} + \mathbf{Q}_{xx}^{-1}\tfrac{\partial^2 \mathbf{Q}_{xx}}{\partial v_2 \partial \boldsymbol{\theta}_j}\mathbf{Q}_{xx}^{-1} ]X_{ig}Y_{ig} \\
		\tfrac{\partial^2 \psi}{\partial \boldsymbol{\theta}_m \partial \boldsymbol{\theta}_j} &= [ 2 \ \mathbf{Q}_{xx}^{-1}\tfrac{\partial \mathbf{Q}_{xx}}{\partial \boldsymbol{\theta}_m}\mathbf{Q}_{xx}^{-1}\tfrac{\partial \mathbf{Q}_{xx}}{\partial \boldsymbol{\theta}_j}\mathbf{Q}_{xx}^{-1} + \mathbf{Q}_{xx}^{-1}\tfrac{\partial^2 \mathbf{Q}_{xx}}{\partial \boldsymbol{\theta}_m \partial \boldsymbol{\theta}_j}\mathbf{Q}_{xx}^{-1} ]X_{ig}Y_{ig}
	\end{align*}
	By assumption, the smallest eigenvalue of $\mathbf{Q}_{xx}$ is bounded below by $\underline{\lambda} >\theta$ for $\boldsymbol{\theta} \in B(\boldsymbol{\theta}_0,\delta)$. For parameter values in this set, $\Vert \mathbf{Q}_{xx}^{-1} \Vert \le \underline{\lambda}^{-1}$ by Lemma X and 
	\begin{align*}
		\left\Vert \tfrac{\partial \psi}{\partial v_2} \right\Vert &\le \left\Vert \mathbf{Q}_{xx}^{-1}\right\Vert \ \left\Vert \tfrac{\partial \mathbf{Q}_{xx}}{\partial v_2} \right\Vert \ \left\Vert \mathbf{Q}_{xx}^{-1} \right\Vert  \ \Vert X \Vert \ \Vert Y \Vert \le \underline{\lambda}^{-2} \mathbf{Q}_{xx}^{\partial} \Vert X_{ig}Y_{ig} \Vert \\
		\left\Vert \tfrac{\partial \psi}{\partial \theta_m} \right\Vert &\le \left\Vert \mathbf{Q}_{xx}^{-1}\right\Vert \ \left\Vert \tfrac{\partial \mathbf{Q}_{xx}}{\partial \theta_m} \right\Vert \ \left\Vert \mathbf{Q}_{xx}^{-1} \right\Vert  \ \Vert X \Vert \ \Vert Y \Vert \le \underline{\lambda}^{-2} \mathbf{Q}_{xx}^{\partial} \Vert X_{ig}Y_{ig} \Vert
	\end{align*}
	By bounding the respective terms, we can also show that
	\begin{align*}
		\left\Vert \tfrac{\partial^2 \psi}{\partial v_2 \partial \theta_m} \right\Vert &\le (2\underline{\lambda}^{-3} (\mathbf{Q}_{xx}^{\partial})^2 + \underline{\lambda}^{-2}\mathbf{Q}_{xx}^{\partial}) \Vert X_{ig}Y_{ig} \Vert \\
		\left\Vert \tfrac{\partial^2 \psi}{\partial \theta_m \partial \theta_j}\right\Vert &= (2\underline{\lambda}^{-3} (\mathbf{Q}_{xx}^{\partial})^2 + \underline{\lambda}^{-2}\mathbf{Q}_{xx}^{\partial}) \Vert X_{ig}Y_{ig} \Vert
	\end{align*}
By the Cauchy Schwartz inequality, $\mathbb{E}[\Vert X_{ig}Y_{ig}\Vert] < \sqrt{\mathbb{E}[\Vert X_{ig}\Vert^2]\mathbb{E}[\Vert Y_{ig} \Vert^2]}$, which is bounded by assumption.

	By applying the Cauchy-Schwartz inequality a second time,
	\[
\mathbb{E}\left[\sup_{\theta \in \Theta} \sup_{v_2 \in \mathcal{V}_2} \mathbf{Q}_{xx}^{\partial}(V_{1ig}^0,v_2,\boldsymbol{\theta}) \Vert X_{ig}Y_{ig} \Vert \right] \le \sqrt{\mathbb{E}\left[\sup_{\theta \in \Theta} \sup_{v_2 \in \mathcal{V}_2} \ (Q^{\partial}_{xx}(V_{1ig}^0,v_2,\theta))^2\right] \mathbb{E}\Vert X_{ig}Y_{ig} \Vert} < \infty
\]
\[
\mathbb{E}\left[\sup_{\theta \in \Theta} \sup_{v_2 \in \mathcal{V}_2} \mathbf{Q}_{xx}^{\partial}(V_{1ig}^0,v_2,\boldsymbol{\theta})^2 \Vert X_{ig}Y_{ig} \Vert \right] \le \sqrt{\mathbb{E}\left[\sup_{\theta \in \Theta} \sup_{v_2 \in \mathcal{V}_2} \ (Q^{\partial}_{xx}(V_{1ig}^0,v_2,\theta))^4 \right] \mathbb{E}\Vert X_{ig}Y_{ig} \Vert} < \infty
\]
The fourth moment of the Sobolev norm is bounded by assumption. Consequently, $\mathbb{E}[\psi_{IW,ig}^{\partial}]$ as defined in \eqref{eq:envelope_psitheta} is bounded.
	
\end{proof}

\begin{proof}[Proof of Lemma \ref{lem:Op_bounds_averages} (\nameref{lem:Op_bounds_averages})]
	We start by writing $\bar{X}_t$ in terms of within-group averages $\bar{X}_g$.
	\[
	\bar{X}_t = \frac{1}{G}\sum_{g=1}^G \frac{N_{gt}}{\bar{n}_t} \left(\frac{1}{N_{gt}}\sum_{i=1}^{N_{gt}} X_{igt} \right) \equiv \frac{1}{G}\sum_{g=1}^{G} \frac{N_{gt}}{\bar{n}_t} (\bar{X}_{gt})
	\]
	Determining the properties of the average is slightly complicated by the fact that $N_{gt}$ is a random variable, which means that $\bar{X}_{gt}$ is an average with a random number of terms.
	
	Let $\mathbb{E}_*$ be a measure where groups are given equal weight regardless of their size, which satisfies two properties: (i) $\mathbb{E}_*[\rho_{gt} X_{igt}] = \mathbb{E}[X_{igt}]$ and (ii) $\mathbb{E}_*[X_{igt} \mid N_{gt} = n] = \mathbb{E}[X_{igt} \mid N_{gt} = n]$, where $\rho_{gt} \equiv N_{gt}/N_t$ and $N_t = \mathbb{E}_*[N_{gt}]$. Property (i) links the equal-weighted measure to the population measure by including importance weights, whereas property (ii) states the two measures are identical after conditioning on group size.
	
	Our first task is to show that $\tfrac{N}{N_t}\bar{X}$ is unbiased. By substituting the definition of $\rho_{gt}$,
	\begin{equation}
	\frac{N}{N_t}\bar{X}_t = \frac{1}{G}\sum_{g=1}^G \rho_{gt} \bar{X}_{gt}
	\label{eq:barX_groupdecomp}
	\end{equation}
	Conditional on group size, $\bar{X}_{gt}$ is an average of a fixed number of terms and hence $\mathbf{E}[\bar{X}_{gt} \mid N_{gt} = n] = \mathbb{E}[X_{igt} \mid N_{gt} = n]$. Therefore by the law of iterated expectations and distributing the expectation over each group
	\[
	\mathbb{E}_*\left[\frac{N}{N_t}\bar{X}_t\right] = \mathbb{E}_*[\rho_{gt} \bar{X}_{gt}] =   \mathbb{E}_*[\rho_{gt} \mathbb{E}_*[X_{igt} \mid N_{gt}]] = \mathbb{E}_*[\rho_{gt} X_{igt}] = \mathbb{E}[X_{igt}]
	\]
	Our next task is to show that $\bar{X}$ is bounded in probability. 
	By the triangle inequality and the law of iterated expectations.
	\[
	\mathbb{E}_*[\Vert \bar{X} \Vert  ] \le \frac{1}{G}\sum_{g=1}^G   \mathbb{E}_*\left[ \tfrac{N_{gt}}{N} \mathbb{E}_*\left[ \Vert \bar{X}_{gt} \Vert \mid \mathbf{N}\right]\right] \le \frac{1}{G}\sum_{g=1}^G   \mathbb{E}_*\left[ \tfrac{N_{gt}}{N} \mathbb{E}_*\left[ \Vert X_{igt} \Vert \mid \mathbf{N}\right]\right] = \mathbb{E}\left[ \tfrac{N_{gt}}{N} \Vert X_{igt}\Vert \right] 
	\]
	Assumption \nameref{assumptxt:boundedrelative_n} states that the $N_{gt} \in [\underline{\rho},\overline{\rho}] \subset (0,1)$ which means that the sample size average $\bar{n}_t \in [\underline{\rho},\overline{\rho}]$. Consequently,
	\[
	\frac{N_{gt}}{\bar{n}_t} = \frac{N_{gt}}{N_t} \times \frac{N_t}{\bar{n}} = \rho_{gt} \times \frac{N_t}{\bar{n}} \le \rho_{gt} \times (1/\underline{\rho})
	\]
	Hence $\mathbb{E}_*[\Vert \bar{X}_t \Vert ] \le (1/\underline{\rho})\mathbb{E}_*[\rho_{gt} \Vert X_{igt} \Vert] = 1/\underline{\rho}\mathbb{E}[\Vert X_{igt} \Vert]$, which is bounded. Then by Markov's inequality, for fixed $\delta > 0$,
	\[
	\mathbb{P}(\Vert \bar{X}_t \Vert > \delta ) \le \frac{\mathbb{E}[\Vert \bar{X}_t\Vert]}{\delta}
	\]
	Therefore $\bar{X}_t = O_p(1)$.
	
	Finally, under \nameref{assumptxt:randomsampling} the observations in each group are independent. Since the \eqref{eq:barX_groupdecomp} is an average of i.i.d variables with finite moments, then we can apply the strong law of large numbers in \cite[p.282]{ billingsley1995probability}, to show that $\tfrac{N}{N_t}\bar{X}_t \to^{p} \mathbb{E}[X_{igt}]$. As a special case, $\tfrac{N}{N_t} \to^{p} 1$. By combining the two results, we find that $\bar{X}_t \to^p \mathbb{E}[X_{igt}]$.	
	
\end{proof}

\begin{proof}[Proof of Lemma \ref{lem:uniformconvergence_sampleaverages} (\nameref{lem:uniformconvergence_sampleaverages})]
	I start by proving point-wise convergence of the criterion function. For simplcity define $\Delta_{igt} \equiv \Vert V_{2igt} - V_{2igt}^0\Vert $. By a first-order Taylor expansion
	\[
	\widehat{f}(\boldsymbol{\tau},\boldsymbol{\theta}) \equiv \frac{1}{G_t\bar{n}_t}\sum_{ig}f(Z_{igt},V_{2igt},\theta) =  \frac{1}{G_t\bar{n}_t}\sum_{ig} \left[ f(Z_{igt},V_{2igt}^0,\theta) +   \tfrac{\partial}{\partial v_2} f(Z_{igt}\tilde{V}_{2igt},\theta)\Delta_{igt} \right]
	\]
	I apply the triangle inequality to bound the second term. By the triangle inequality
	\begin{align}
		\begin{split}
			& \left\Vert \frac{1}{G_t\bar{n}_t}\sum_{ig}\tfrac{\partial}{\partial v_2} f(Z_{igt}\tilde{V}_{2igt},\boldsymbol{\tau},\boldsymbol{\theta}) \Delta_{igt} \right\Vert \\
			& \le \left( \frac{1}{G_t\bar{n}_t}\sum_{ig} \left\Vert \tfrac{\partial}{\partial v_2} f(Z_{igt}\tilde{V}_{2igt},\theta) \right\Vert  \right) \max_{ig} \Delta_{igt} \\
			&\le \left(\frac{1}{G_t\bar{n}_t}\sum_{ig} \sup_{(\boldsymbol{\tau},\boldsymbol{\theta},v_2)\in \mathcal{T} \times \Theta \times \mathcal{V}_2} \left\Vert \tfrac{\partial}{\partial v_2}  f(Z_{igt},v_2,\boldsymbol{\tau},\boldsymbol{\theta})\right\Vert\right) \ \max_{ig}\Delta_{igt} \\
			&\equiv \left(\frac{1}{G_t\bar{n}_t}\sum_{ig} f_{igt}^V \right) \max_{ig}\Delta_{igt}
		\end{split}
		\label{eq:bound_criterion_derivA}
	\end{align}
	The discrepancy $\max_{ig} \Vert \Delta_{igt} \Vert$ is $o_p(1)$ by assumption (i). Conversely, by assumption (iii) $\mathbb{E}[f_{igt}^V]<\infty$ and Lemma \ref{lem:Op_bounds_averages} imply that $\frac{1}{G_t\bar{n}_t}\sum_{ig} f_{igt}^V = O_p(1)$. Finally, by combining the two finding we conclude that the right-hand side of \eqref{eq:bound_criterion_derivA} is $o_p(1)$.
	
	Assumptions (i) implies that $f(Z_{igt},V_{2igt}^0,\theta)$ has bounded moments. Similarly, \nameref{assumptxt:randomsampling} implies that groups are independent. Therefore, we can apply a group level law of large numbers to show that as $(G,N) \to \infty$, 
	\[
	\widehat{f}(\boldsymbol{\tau},\boldsymbol{\theta}) =  \frac{1}{G_t\bar{n}_t}\sum_{ig} f(Z_{igt},V_{2igt}^0,\boldsymbol{\tau},\boldsymbol{\theta}) + o_p(1) = \mathbb{E}[f(Z_{igt},V_{2igt}^0,\boldsymbol{\tau},\boldsymbol{\theta})] + o_p(1)
	\]
	Our next task is to show that the criterion function is stochastically equicontinuous, in the sense defined by \cite{newey1991uniform}. Let $(\theta,\theta^*)$ be two distinct parameter values and define a uniform bound on the derivative $S_{igt}$ as in Lemma \ref{lem:uniformboundsderivatives}. Then 
	\begin{align}
		\begin{split}
			&\Vert \widehat{f}(\boldsymbol{\tau},\boldsymbol{\theta}) - \widehat{f}(\boldsymbol{\tau}^*,\boldsymbol{\theta}^*) \Vert \\
			&= \left\Vert\frac{1}{G_t\bar{n}_t}\sum_{ig} \tfrac{\partial}{\partial 
				(\boldsymbol{\tau},\boldsymbol{\theta})} f(Z_{igt},V_{2igt}^0,\widetilde{\boldsymbol{\tau}},\widetilde{\boldsymbol{\theta}}) [(\boldsymbol{\tau},\boldsymbol{\theta})' - (\boldsymbol{\tau}^*,\boldsymbol{\theta}^*)'] \right\Vert \\
			&\le \left(\frac{1}{G_t\bar{n}_t}\sum_{ig} \sup_{(\boldsymbol{\tau},\boldsymbol{\theta},v_2)\in \mathcal{T} \times \Theta \times \mathcal{V}_2} \left\Vert \tfrac{\partial}{\partial 
				(\boldsymbol{\tau},\boldsymbol{\theta})}  f(Z_{igt},v_2,\widetilde{\boldsymbol{\tau}},\widetilde{\boldsymbol{\theta}}) \right\Vert \right)\Vert (\boldsymbol{\tau},\boldsymbol{\theta})' - (\boldsymbol{\tau}^*,\boldsymbol{\theta}^*)' \Vert \\
			&\le \left(\frac{1}{G_t\bar{n}_t}\sum_{ig} f_{igt}^{(\boldsymbol{\tau},\boldsymbol{\theta})} \right)\Vert (\boldsymbol{\tau},\boldsymbol{\theta})' - (\boldsymbol{\tau}^*,\boldsymbol{\theta}^*)' \Vert
		\end{split}
	\end{align}
	By assumption (iv) $\mathbb{E}[f_{igt}^{(\boldsymbol{\tau},\boldsymbol{\theta})}] < \infty$ and by Lemma \ref{lem:Op_bounds_averages}, $\frac{1}{G_t\bar{n}_t}\sum_{ig} f_{igt}^{(\boldsymbol{\tau},\boldsymbol{\theta})} = O_p(1)$. This exactly fits the definition of stochastic equicontinuity. Since the parameter space is compact, the function converges point-wise and the sample-criterion is stochastically equicontinuous, then by Theorem 2.1 in \cite{newey1991uniform},
	\begin{equation}
	\sup_{\boldsymbol{\tau} \in \mathcal{T}} \sup_{\boldsymbol{\theta} \in \Theta} \Vert \widehat{f}(\boldsymbol{\tau},\boldsymbol{\theta}) - \mathbb{E}[f(Z_{igt},V_{2igt}^0,\boldsymbol{\tau},\boldsymbol{\theta})] \Vert \to^{p} 0
	\label{eq:uniformconsistency_criterion}
	\end{equation}
	
\end{proof}

\newpage

\newpage

\end{document}